\documentclass[12pt]{article}
\usepackage[a4paper,margin=2cm]{geometry}
\usepackage{amsmath}
\usepackage{graphicx}
\usepackage{url} 
\usepackage{changepage}
\usepackage[utf8x]{inputenc}
\usepackage{textcomp,marvosym}
\usepackage{fixltx2e}
\usepackage{amsmath,amssymb,amsthm}
\usepackage{nameref,hyperref}
\usepackage{authblk}
\usepackage{algorithm}
\usepackage{algpseudocode}
\usepackage{dsfont}
\usepackage{xcolor}
\usepackage{enumitem}
\usepackage{caption,setspace}
\usepackage{pdflscape}
\usepackage{graphicx}
\usepackage{placeins}
\usepackage{mathtools}
\usepackage{appendix}
\usepackage{extarrows}
\usepackage{cleveref}
\newtheorem{theorem}{Theorem}
\newcommand{\tp}{\mathrm{tp}}
\newcommand{\fp}{\mathrm{fp}}
\newcommand{\fn}{\mathrm{fn}}

\newcommand{\abc}{\mathrm{ABC}}
\newcommand{\mfabc}{\mathrm{MFABC}}

\newcommand{\dydx}[2]{\frac{\text{d} #1}{\text{d} #2}}
\newcommand{\pdydx}[2]{\frac{\partial #1}{\partial #2}}

\renewcommand{\eqref}[1]{Equation~(\ref{#1})}
\newcommand{\Prob}[1]{\mathbb{P}(#1)}
\newcommand{\CondProb}[2]{\mathbb{P}(#1 \mid #2)}
\newcommand{\CondE}[2]{\mathbb{E}\left[#1 \mid #2\right]}

\newcommand{\PDF}[1]{p(#1)}
\newcommand{\CondPDF}[2]{p(#1 \mid #2)}
\newcommand{\CondPDFsub}[3]{p_{#3}(#1 \mid #2)}

\newcommand{\like}[2]{\mathcal{L}(#2 ; #1)}
\newcommand{\approxlike}[2]{L(#2 ; #1)}

\newcommand{\E}[1]{\mathbb{E}\left[#1\right]}
\newcommand{\V}[1]{\text{Var}\left[#1\right]}

\newcommand{\C}[2]{\text{Cov}\left[#1,#2\right]}
\newcommand{\bvec}[1]{\mathbf{#1}}
\newcommand{\ind}[2]{\mathds{1}_{#1}\left(#2\right)}

\newcommand{\simProc}[2]{s(#1 \mid #2)}
\newcommand{\obsProc}[2]{g(#1 \mid #2)}
\newcommand{\approxsimProc}[2]{\tilde{s}(#1 \mid #2)}
\newcommand{\approxsimProct}[3]{s^{#3}(#1 \mid #2)}
\newcommand{\paramvec}{\boldsymbol{\theta}}

\newcommand{\param}{\theta}
\newcommand{\paramspace}{\boldsymbol{\Theta}}
\newcommand{\discrep}[2]{\rho( #1, #2)}
\newcommand{\approxdiscrep}[2]{\tilde{\rho}( #1, #2)}
\newcommand{\approxdiscrept}[3]{\rho^{#3}( #1, #2)}
\newcommand{\dat}{\mathcal{D}}
\newcommand{\simdat}{\mathcal{D}_s}
\newcommand{\approxsimdat}{\tilde{\mathcal{D}}_s}

\begin{document}

\title{Multifidelity multilevel Monte Carlo to accelerate approximate Bayesian parameter inference for partially observed stochastic processes}

\author[1,2]{David~J. Warne\footnote{To whom correspondence should be addressed. E-mail: david.warne@qut.edu.au}}
\author[3,4]{Thomas~P. Prescott}
\author[4]{Ruth~E. Baker}
\author[1,2]{Matthew~J.~Simpson}

\affil[1]{School of Mathematical Sciences, Queensland University of Technology, Brisbane, Queensland 4001, Australia}
\affil[2]{Centre for Data Science, Queensland University of Technology, Brisbane, Queensland 4001, Australia}
\affil[3]{The Alan Turing Institute, London, NW1 2DB, United Kingdom}
\affil[4]{Mathematical Institute, University of Oxford, Oxford, OX2 6GG, United Kingdom}

\maketitle

\begin{abstract}
	Models of stochastic processes are widely used in almost all fields of science. Theory validation, parameter estimation, and  prediction all require model calibration and statistical inference using data. However, data are almost always incomplete observations of reality. This leads to a great challenge for statistical inference because the likelihood function will be intractable for almost all partially observed stochastic processes. This renders many statistical methods, especially within a Bayesian framework, impossible to implement. Therefore, computationally expensive likelihood-free approaches are applied that replace likelihood evaluations with realisations of the model and observation process. For accurate inference, however, likelihood-free techniques may require millions of expensive stochastic simulations. 
To address this challenge, we develop a new method based on recent advances in multilevel and multifidelity methods for parameter inference using partially observed Markov processes.  Our novel approach combines the multilevel Monte Carlo telescoping summation, applied to a sequence of approximate Bayesian posterior targets, with a multifidelity rejection sampler that learns from computationally inexpensive model approximations to minimise the number of computationally expensive exact simulations required for accurate inference. We present the derivation of our new algorithm for likelihood-free Bayesian inference, discuss practical implementation details, and demonstrate substantial performance improvements. Using examples from systems biology, we demonstrate improvements of more than two orders of magnitude over standard rejection sampling techniques. Our approach is generally applicable to accelerate other sampling schemes, such as sequential Monte Carlo, to enable feasible Bayesian analysis for realistic practical applications in physics, chemistry, biology, epidemiology, ecology and economics. We provide source code implementations of our methods and demonstrations (available at \href{https://github.com/davidwarne/MLMCandMultifidelityForABC}{https://github.com/davidwarne/MLMCandMultifidelityForABC}).
\end{abstract}

\section{Introduction}

Stochastic processes are used to model complex systems in almost all fields of science and engineering. Partially observed stochastic processes result in some of the most computationally challenging problems for Bayesian inference~\cite{Meulen2017,Dellaportas2006,Golightly2006}. Given the ubiquity of stochastic processes for real applications it is essential that efficient methods are developed to enable the analysis of modern high resolution data sets without sacrificing accuracy.

An important application of stochastic processes occurs in the study of cellular processes~\cite{Elowitz2002,Raj2008}. Here, stochastic models of biochemical reaction networks often provide a more accurate description of system dynamics than deterministic models~\cite{Wilkinson2009}. 
This is largely due to intrinsic noise in the system dynamics of many biochemical processes that are significantly influenced by relatively low populations of certain chemical species~\cite{Fedoroff2002}. For example, in eukaryotic cells, molecules that regulate gene expression occur in relatively low numbers; as a result, stochastic fluctuations have a direct effect on the production rates of proteins~\cite{Elowitz2002,Blake2003,Braichenko2021}. In addition, there are other interesting phenomena that occur in biological systems that can only be captured by stochastic models, for example, self-induced stochastic resonance~\cite{Gammaitoni1998,Muratov2005,Wellens2004}, stochastic focusing~\cite{Paulsson2000}, and stochastic bi-stability~\cite{Schlogl1972,Schnakenberg1979}.

To quantify uncertainty in parameters or predictions, it is typical to consider an expectation of a function of an unknown vector of model parameters, $\paramvec \in \paramspace$, conditional on some observational data, $\dat$,
\begin{equation}
\E{f(\paramvec)} = \int_{\paramspace}  f(\paramvec)\CondPDF{\paramvec}{\dat} \,\text{d}\paramvec,
\label{eq:Ef}
\end{equation}
where $\CondPDF{\paramvec}{\dat}$ is the Bayesian posterior probability density~\cite{Gelman2014},
\begin{eqnarray}
\label{eq:bayes}
\CondPDF{\paramvec}{\dat} = \frac{\like{\dat}{\paramvec}\PDF{\paramvec}}{\PDF{\dat}}.
\end{eqnarray}
Here, $\PDF{\paramvec}$ is the \emph{a priori} probability density of unknown parameters, $\like{\dat}{\paramvec}$ is the likelihood function that maps a parameter vector $\paramvec$ to the probability of the observations $\dat$, and $\PDF{\dat}$ is a normalising constant often referred to as the evidence. The function $f(\cdot)$ can be any scalar function of the parameters $\paramvec$. For example, \eqref{eq:Ef} could represent a posterior moment, probability density or cumulative probability. Similarly $f(\cdot)$ could include a model prediction to enable equivalent expectation of the posterior predictive distribution.

Computational challenges arise due to the evidence term, given by
\begin{eqnarray}
\label{eq:evidence}
\PDF{\dat} = \int_{\paramspace} \like{\dat}{\paramvec}\PDF{\paramvec}\, \text{d}\paramvec.
\end{eqnarray}
This term is rarely tractable and the expectation in \eqref{eq:Ef} must be numerically estimated using sampling techniques, such as Markov chain Monte Carlo (MCMC)~\cite{Metroplis1953,Hastings1970} and sequential Monte Carlo (SMC)~\cite{DelMoral2006}, that only require point-wise evaluation of the likelihood function, $\like{\dat}{\paramvec}$. 

For partially observed Markov processes, point-wise evaluation of the likelihood requires the solution to the forward Kolmogorov equation, which must be computed approximately. Therefore, standard Bayesian tools cannot be applied and likelihood-free methods are needed, such as approximate Bayesian computation (ABC)~\cite{Sisson2018,Sunnaker2013,Warne2019a}, pseudo-marginal methods~\cite{Andrieu2009,Andrieu2010,Warne2020}, and Bayesian synthetic likelihood (BSL)~\cite{An2019,Ong2017,Price2017}. In the machine learning literature, there are a multitude of likelihood-free, or simulation-based, approaches available~\cite{Cranmer2020}, for example the use of deep neural networks to learn a surrogate model of the posterior~\cite{Lueckmann2017} or the likelihood~\cite{Papamakarios2019}. Regardless of the likelihood-free approach, realisations of the stochastic model are  generated by stochastic simulation in place of likelihood function evaluation. Thus, the computational cost of evaluating \eqref{eq:Ef} depends on the efficiency of the stochastic simulation algorithm and the posterior sampler used as a basis for likelihood-free inference. For example, ABC can be implemented using either rejection sampling~\cite{Pritchard1999,Tavare1997,Beaumont2002}, MCMC~\cite{Marjoram2003} or SMC approaches~\cite{Sisson2007,DelMoral2012,Drovandi2011}. However, if the model is even moderately expensive then such techniques are still infeasible.

Recently, there has been substantial research activity in the application of approximate model simulations and posterior samplers in combination with bias correction adjustments to accelerate likelihood-free applications that would be impractical otherwise.  Techniques such as transport maps~\cite{Parno2018} and moment-matching transforms~\cite{Warne2019b} aim to transform a set of approximate posterior samples, using a surrogate or reduced model, into posterior samples under an expensive exact model. Other approaches such as preconditioning utilise approximate models to inform efficient proposal mechanisms~\cite{Warne2019b}. Various early-rejection and delayed-acceptance schemes~\cite{Banterle2019,Everitt2020,Golightly2014,Prangle2014} probabilistically simulate the accurate model based on the rejection/acceptance status of an approximation; this family of methods is generalised for ABC schemes under the multifidelity framework~\cite{Prescott2020,Prescott2021,Peherstorfer2016}. 

There has also been substantial research in the last decade on the use of control variates to improve the rate of convergence in \emph{mean-square} of Monte Carlo estimates to expectations. In particular, the multilevel Monte Carlo (MLMC) method~\cite{Heinrich1998,Giles2015} expands an expectation as a telescoping summation of bias corrections and achieves variance reductions by exploiting path-wise convergence properties of numerical schemes for solving stochastic differential equations (SDEs)~\cite{Giles2008} or discrete-state Markov processes~\cite{Anderson2012,Lester2014,Lester2015}. Recently, the MLMC telescoping summation approach has also been successfully applied to Bayesian inference~\cite{Dodwall2015,Dodwell2019,Efendiev2015,Beskos2016,Latz2018}, as per \eqref{eq:Ef}, including ABC-based samplers~\cite{Guha2017,Jasra2019,Warne2018}. 

Many of the above approaches use approximate simulations in ways that are not always mutually exclusive. However, it is an open question as to how these various schemes could be combined to obtain compounding effects in performance.   

\subsection{Contribution}
In this work, we derive and demonstrate an effective method for approximate Bayesian inference that combines the benefits of multifidelity ABC sampling with that of variance reduction with MLMC applied to the expectations with respect to an ABC posterior density. Acknowledging the wide range of variations in these algorithm classes, we specifically consider a natural application of these approaches to rejection samplers and extend the methods of Prescott and Baker~\cite{Prescott2020}, and Warne et al.~\cite{Warne2018}. Specifically, we apply the MLMC telescoping summation of Warne et al.~\cite{Warne2018} to the multifidelity ABC sampler of Prescott and Baker~\cite{Prescott2020} including practical implementation of an adaptive tuning scheme for the multifidelity approach~\cite{Prescott2021b}. These methods combined achieve multiple orders of magnitude improvement over direct ABC rejection sampling. In addition, for situations where rejection sampling is practically intractable, our methodology is equally extendable to other sampling approaches, such as SMC, using developments by Prescott and Baker~\cite{Prescott2021} and Jasra et al.~\cite{Jasra2019}. Using biochemical reaction networks as a characteristic application area, we present key computational features of ABC methods and introduce the fundamental algorithms developed by Prescott and Baker~\cite{Prescott2020} and Warne et al.~\cite{Warne2018}. We then highlight the distinct features of the inference problem exploited by each approach to accelerate ABC rejection sampling, and derive our novel approach by leveraging the complementary nature of these features. We then explore various performance results using stochastic models of Michaelis--Menten enzyme kinetics and the repressilator gene regulatory network. These examples enable us to establish requirements for practical application of our approach and provide rules-of-thumb to simplify the tuning of algorithm parameters. Finally, we demonstrate the efficacy of our multifidelity MLMC approach using a challenging model of a two--step MAPK cascade reaction, which forms the basis for many real biological functions such as cell-to-cell signalling through the epithelial growth factor receptor (EGFR)~\cite{Dhananjaneyulu2012,Brown2004}. Our results demonstrate the potential of multilevel and multifidelity methods to accelerate ABC sampling by several orders of magnitude.

\section{Methods}

In this section we describe standard numerical methods for simulation and inference of stochastic biochemical network models~\cite{Warne2019a,Schnoerr2017,Higham2008}. Then we present recent developments in MLMC~\cite{Warne2019a,Warne2018} and multifidelity~\cite{Prescott2020} methods for inference, and develop novel extensions that achieve the benefits from both the MLMC and multifidelity approaches. Finally, we highlight practical challenges and provide guidelines for the application and tuning of our new method.

\subsection{Stochastic models of biochemical reaction networks}
We consider biochemical reaction network models that involve a well-mixed population of $\mathcal{N}$ chemical species, $\mathcal{X}_1,\mathcal{X}_2,\ldots,\mathcal{X}_\mathcal{N}$, that react via a network of $\mathcal{M}$ chemical reactions,
\begin{eqnarray*}
	\sum_{i=1}^\mathcal{N} \nu_{i,j}^- \mathcal{X}_i  \to \sum_{i=1}^\mathcal{N} \nu_{i,j}^+ \mathcal{X}_i, \quad j = 1,2,\ldots, \mathcal{M},
\end{eqnarray*}
where $\nu_{i,j}^-$ and $\nu_{i,j}^+$ are, respectively, the reactant and product stoichiometries for species $\mathcal{X}_i$ in reaction $j$. Given the state of the system at time $t$, $\bvec{X}_t = [X_{1,t},X_{2,t},\ldots,X_{\mathcal{N},t}]^\text{T}$  with $X_{i,t}$ denoting the copy number of species $\mathcal{X}_{i}$ at time $t$, it can be shown that, for a sufficiently small time interval $[t,t+\Delta t)$, the probability of reaction $j$ occurring within this interval is proportional to $ a_j(\bvec{X}_t)\Delta t$ where $a_j(\bvec{X}_t)$ is the propensity function of reaction $j$. Generally, the propensity function will take the form
\begin{eqnarray*}
	a_j(\bvec{X}_t,\paramvec) = k_j \prod_{i=1}^\mathcal{N} \nu_{i,j}^- ! \binom{X_{i,t}}{\nu_{i,j}^-},
\end{eqnarray*} 
where $k_j$ is the non-dimensionalised kinetic rate parameter for reaction $j$. However, other more complex forms are also possible that may depend on a general parameter vector such as $\paramvec = [k_1,k_2, \ldots,k_{\mathcal{M}},\boldsymbol{\lambda}]$ where $\boldsymbol{\lambda}$ could include other parameters like hill constants, observation error, or domain sizes.   

Mathematically, the stochastic dynamics of a biochemical network is governed by a discrete-state, continuous-time Markov process~\cite{Higham2008,Gillespie1977,Kurtz1972,Erban2007}. The transitional density function of this process, $\CondPDF{\bvec{X}_t}{\bvec{X}_s,\paramvec}$, describes the probability of the system state at time $t$ given the state at a previous time $s$. To obtain $\CondPDF{\bvec{X}_t}{\bvec{X}_s,\paramvec}$ one needs to solve the forward Kolmogorov equation, also referred to as the chemical master equation (CME)~\cite{Higham2008,Gillespie1992}, 
\begin{eqnarray}
\dydx{\CondPDF{\bvec{X}_t}{\bvec{X}_s},\paramvec}{t} = \sum_{j=1}^{\mathcal{M}} a_j(\bvec{X}_t -\boldsymbol{\nu}_{j};\paramvec) \CondPDF{\bvec{X}_t - \boldsymbol{\nu}_{j}}{\bvec{X}_s,\paramvec} - a_j(\bvec{X}_t ;\paramvec) \CondPDF{\bvec{X}_t}{\bvec{X}_s,\paramvec},
\label{eq:genCME}
\end{eqnarray} 
where $\bvec{\boldsymbol{\nu}}_{j}$ is the state change corresponding to the occurrence of reaction $j$, that is $\bvec{\boldsymbol{\nu}}_{j} = [\nu^+_{1,j} - \nu^-_{1,j}, \nu^+_{2,j} - \nu^-_{2,j}, \ldots, \nu^+_{\mathcal{N},j} - \nu^-_{\mathcal{N},j}]^\text{T}$. While we focus on Markov processes in this work, equivalent concepts exist for non-Markovian systems; though with more complexity~\cite{Braichenko2021}. Our new methods may be generalised to non-Markovian models without substantial modification.

\subsubsection{Stochastic simulation methods}
The CME is intractable for all but the simplest of models~\cite{Warne2019a,Schnoerr2017,Higham2008}. As a result, exact stochastic simulation~\cite{Gillespie1977,Gibson2000,Anderson2007} is required to study the system behaviour without introducing potentially substantial bias~\cite{Warne2019b}. Stochastic simulation schemes generate realisations of the network state, $\bvec{X}_t$, from some initial time $t_0$ up to a termination time $T$. Exact stochastic simulation schemes, such as Gillespie's direct method~\cite{Gillespie1977} (Algorithm~\ref{alg:essa}), simulate every reaction event, and are  computationally prohibitive for systems with very large copy numbers or very high reaction rates. 

\begin{algorithm}
	\caption{Gillespie's direct method for stochastic simulation}
	\begin{algorithmic}[1]
		\State Initialise $t = t_0$ and $\bvec{X} = \bvec{x}_0$;
		\Loop
		\State Set $a_0 \leftarrow \sum_{j=1}^\mathcal{M} a_j(\bvec{X})$;
		\State Sample next reaction time $\Delta t \sim \text{Exp}(a_0)$;
		\If{$t + \Delta t > T$}
		\State \Return
		\Else
		\State Select reaction $j \in [1,2,\ldots,\mathcal{M}]$ with probabilities
		\begin{displaymath}
		\Prob{j = 1} = a_1(\bvec{X})/a_0,\quad \Prob{j = 2} = a_2(\bvec{X})/a_0,\quad \ldots,\quad \Prob{j = \mathcal{M}} = a_\mathcal{M}(\bvec{X})/a_0;
		\end{displaymath}
		\State Set $\bvec{X} \leftarrow \bvec{X} + \boldsymbol{\nu}_j$ and $t \leftarrow t + \Delta t$.
		\EndIf
		\EndLoop
	\end{algorithmic}
	\label{alg:essa}
\end{algorithm}
Various approximate stochastic simulation schemes, such as the tau-leaping method~\cite{Gillespie2001} (Algorithm~\ref{alg:assa}), can be applied~\cite{Higham2008,Gillespie2001,Gillespie2000,Tian,Cao2004} to improve the computational performance, but there will be bias  incurred due to the simplifying approximations~\cite{Warne2019a,Anderson2011}. For example, assume the propensities do not change substantially over a time interval of length $\tau$. The resulting sample path $\bvec{Z}_t$ will be a discrete-time Markov chain approximation to a true path $\bvec{X}_t$ from the full continuous-time process. 
\begin{algorithm}
	\caption{The tau-leaping method for approximate stochastic simulation}
	\begin{algorithmic}[1]
		\State Initialise $t = t_0$, and $\bvec{Z} = \bvec{x}_0$;
		\While{$t + \tau > T$}
		\State Generate event counts, $Y_j \sim \text{Po}(a_j(\bvec{Z})\tau)$, for $j = 1,2,\ldots, \mathcal{M}$;
		\State Set $\bvec{Z} \leftarrow \bvec{Z} + \sum_{j=1}^\mathcal{M} Y_j \boldsymbol{\nu}_j$, and $t \leftarrow t + \tau$.
		\EndWhile
	\end{algorithmic}
	\label{alg:assa}
\end{algorithm}

\FloatBarrier

\subsubsection{Acceleration using multilevel Monte Carlo}
Often the goal of stochastic simulation is to estimate the expectation, $\E{f(\bvec{X}_T)}$, where $f(\cdot)$ is a function of the process state at time $t = T > 0$. Note that this expectation can, through the specification of $f(\cdot)$, resolve to any raw or central moment, it may also be used to estimate the full probability density or cumulative distribution of the system state $\bvec{X}_T$~\cite{Warne2018,Wilson2016,Giles2015b}. MLMC provides a mechanism to exploit approximations that computationally accelerate the estimation~\cite{Warne2019a,Giles2008,Lester2014}. Assume we have a stochastic process $\{\bvec{X}_t\}_{t\geq0}$, such as biochemical reaction network model, that is computationally expensive to simulate with the Gillespie direct method (Algorithm~\ref{alg:essa}) or equivalent. Now consider a sequence of $L$ stochastic processes, $\{\{\bvec{Z}_{t,\ell}\}_{t\geq0}\}_{\ell = 1}^{\ell = L}$, that approximate  $\{\bvec{X}_t\}_{t\geq0}$. This sequence is constructed such that the bias decreases and computational cost increases with $\ell$. For example, $\{\bvec{Z}_{t,\ell}\}_{t\geq0}$ could be a tau-leaping approximation with time step $\tau_\ell = c 2^{-\ell}$ and $c$ constant for all $\ell = 1,2,\ldots, L$. The insight of Giles~\cite{Giles2008} was to expand the desired expectation as a telescoping summation by exploiting linearity of expectation,    
\begin{equation}
\E{f(\bvec{X}_T)} = \E{f(\bvec{Z}_{T,1})} + \E{f(\bvec{X}_T) - f(\bvec{Z}_{T,L})} + \sum_{\ell=2}^L \E{f(\bvec{Z}_{T,\ell}) - f(\bvec{Z}_{T,\ell-1})},
\label{eq:mlts}
\end{equation}
where each of the difference terms acts to correct for the bias of the initial biased approximation, $ \E{f(\bvec{Z}_{T,1})}$. Giles~\cite{Giles2008} demonstrated that substantial computational improvements can be obtained if positive correlations can be induced between the terms in the bias corrections, that is between $\bvec{Z}_{T,\ell}$ and $\bvec{Z}_{T,\ell-1}$ for $\ell = 2, \ldots, L$, and between $\bvec{X}_T$ and $\bvec{Z}_{T,L}$. For SDEs and discrete-state continuous-time Markov processes, coupling schemes that induce sufficiently strong positive correlations have been well studied~\cite{Warne2019a,Giles2015,Anderson2012,Lester2015} and performance improvements of many orders of magnitude can be obtained without any loss in accuracy in terms of mean-square error. 

\subsection{Approximate Bayesian computation}

For the purposes of inference, the likelihood function for sequences of $n$ observations, denoted by\\  $\bvec{Y}_{\text{obs}}~=~[\bvec{y}_{\text{obs}}(t_1),\bvec{y}_{\text{obs}}(t_2), \ldots,\bvec{y}_{\text{obs}}(t_n)]$, is given by
\begin{align}
\CondPDF{\bvec{Y}_{\text{obs}}}{\paramvec}\, &= \int_{\mathbb{X}^{n+1}} \CondPDF{\bvec{Y}_{\text{obs}},\bvec{X}_{t_0},\ldots,\bvec{X}_{t_n}}{\paramvec}\prod_{i=0}^n \text{d} \bvec{x}_{t_i} \notag\\ &=\int_{\mathbb{X}^{n+1}} \PDF{\bvec{X}_{t_0}}\text{d}\bvec{X}_{t_0}\prod_{i=1}^{n} \obsProc{\bvec{y}_{\text{obs}}(t_i)}{\bvec{X}_{t_{i}},\paramvec}
	\CondPDF{\bvec{X}_{t_i}}{\bvec{X}_{t_{i-1}},\paramvec} \,\text{d}\bvec{X}_{t_i}, 
\label{eq:like} 
\end{align} 
where $\obsProc{\bvec{y}_{\text{obs}}(t_i)}{\bvec{X}_{t_{i}},\paramvec}$ is the probability density of noisy observation $\bvec{y}_{\text{obs}}(t_i)$, given the true state, $\bvec{X}_{t_i} \in \mathbb{X} \subseteq \mathbb{N}^\mathcal{N}$, at time $t_i$,  and $\CondPDF{\bvec{X}_{t_i}}{\bvec{X}_{t_{i-1}},\paramvec}$ is the solution to the CME for the state transition over the time interval $(t_{i-1},t_i]$. The intractability of the CME immediately implies the intractability of the Bayesian inference problem. However, depending on the structure of the observation process, \eqref{eq:like} may still be intractable even when \eqref{eq:genCME} has an analytical solution. Therefore, almost any practical application of partially observed continuous-time Markov processes will require some form of likelihood-free inference.

ABC, pseudo-marginal methods and BSL are all possible approaches to avoid the dependence of Bayesian inference on the likelihood (\eqref{eq:like}) by generating simulated data from the likelihood using exact stochastic simulation along with the observation noise process. Here, we focus on the ABC approach that is based on the posterior approximation,
\begin{equation}
\CondPDF{\paramvec}{\bvec{Y}_{\text{obs}}} \approx \CondPDF{\paramvec}{\discrep{\bvec{Y}_{\text{obs}}}{\bvec{Y}_{s}} \leq \epsilon} \propto \CondProb{\discrep{\bvec{Y}_{\text{obs}}}{\bvec{Y}_{s}} \leq \epsilon}{\paramvec}\PDF{\paramvec},
\label{eq:approxBayes}
\end{equation}
where $\discrep{\bvec{Y}_{\text{obs}}}{\bvec{Y}_{s}}$ is a discrepancy metric between noisy observations, $\bvec{Y}_{\text{obs}}$, and simulated noisy observations, $\bvec{Y}_{s}$, and $\epsilon$ is a sufficiently small discrepancy threshold. The most direct method to implement ABC is to use rejection sampling (Algorithm~\ref{alg:ABC-rej}). 
\begin{algorithm}
	\caption{ABC rejection sampling to generate $N$ approximated posteriors samples}
	\begin{algorithmic}[1]
		\For{$i \in [1,2, \ldots, N]$}
		\Repeat
		\State Sample the prior $\paramvec^* \sim \PDF{\paramvec}$;
		\State Generate simulated data $\bvec{Y}_s \sim \simProc{\cdot}{\paramvec^*}$;
		\Until{$\discrep{\bvec{Y}_{\text{obs}}}{\bvec{Y}_s} \leq \epsilon$}
		\State$\paramvec^i  \leftarrow \paramvec^*$.
		\EndFor
	\end{algorithmic}
	\label{alg:ABC-rej}
\end{algorithm}
Here, independent identically distributed samples of rate parameters are drawn from the prior, and the sample is accepted if a resulting stochastic simulation of the model is within $\epsilon$ of the observations under the discrepancy metric. In practice, rejection sampling may not be computationally feasible due to the prohibitively high rejection rates when the data dimensionality is high. There are many techniques that can be applied to improve the computational efficiency of ABC methods that we do not discuss here due to the wealth of available literature~\cite{Sisson2018,Sunnaker2013,Marjoram2003,Sisson2007,Drovandi2011,Fearnhead2012}. Instead we specifically investigate recent developments in MLMC and multifidelity methods~\cite{Prescott2020,Warne2018}, as they provide a new avenue to explore computational improvements for ABC inference. While we focus on these methods as applied directly to ABC rejection sampling (Algorithm~\ref{alg:ABC-rej}), we note that our work is also applicable to other schemes such as SMC~\cite{Prescott2021,Jasra2019}.

\subsection{Multilevel Monte Carlo and multifidelity methods for ABC inference}
In this section, we describe two recent methods for acceleration of ABC inference that are foundational to the main contribution of this work. These methods are presented by Warne et al.~\cite{Warne2018} and Prescott and Baker~\cite{Prescott2020}, respectively. Both methods build upon the MLMC work of Giles~\cite{Giles2008} and Rhee and Glynn~\cite{Rhee2015} for efficient computation of expectations with respect to a stochastic process~\cite{Giles2008,Anderson2012}. We refer the reader to   Warne et al.~\cite{Warne2019a}, Lester et al.~\cite{Lester2014}, Schnoerr et al.~\cite{Schnoerr2017}, and Peherstorfer et al. \cite{Peherstorfer_2018}, for accessible introductions to simulation and inference methods including MLMC and multifidelity methods.

\subsubsection{Multilevel rejection sampling}

For ABC inference, Warne et~al.~\cite{Warne2019a,Warne2018} consider the direct application of the telescoping summation (\eqref{eq:mlts}) to a sequence of $L$ ABC rejection samplers (Algorithm~\ref{alg:ABC-rej}). That is, 
\begin{equation}
\E{f(\paramvec_L)} = \E{f(\paramvec_1)} + \sum_{\ell=2}^L \E{f(\paramvec_\ell) - f(\paramvec_{\ell-1})},
\label{eq:mlmcabc}
\end{equation}
with $\paramvec_\ell \sim \CondPDF{\cdot}{\discrep{\bvec{Y}_{\text{obs}}}{\bvec{Y}_s} \leq \epsilon_\ell}$  for threshold $\epsilon_\ell = \epsilon_0 m^{-\ell}$ for $\ell = 1,\ldots, L$, where $\epsilon_0$ is a large discrepancy threshold leading to a high acceptance rate (typically close to the prior) and $m > 1$. The greatest challenge in application of MLMC for inference is the construction of a coupling to generate positively correlated sample pairs $(\paramvec_\ell,\paramvec_{\ell-1})$ without introducing additional bias that would violate the telescoping summation. For this Warne et~al.~\cite{Warne2018} apply a novel construction using the marginal empirical distribution at level $\ell$ and the inverse marginal distributions obtained from the bias corrections up to level $\ell-1$. That is, given $N_\ell$ i.i.d. samples $\paramvec_\ell^1,\ldots,\paramvec_\ell^{N_\ell}$, one generates 
\begin{equation}
\tilde{\paramvec}_{\ell-1}^{i} = [\hat{F}_{\ell-1,1}^{-1}(\bar{F}_{\ell,1}(\param_{\ell,1})),\ldots,\hat{F}_{\ell-1,k}^{-1}(\bar{F}_{\ell,k}(\param_{\ell,k}))], \quad \text{ for } i = 1,\ldots, N_{\ell},
\label{eq:mlmcabccouple}
\end{equation}
where $\bar{F}_{\ell,j}(\cdot)$ is the marginal empirical distribution of the $j$th dimension of $\paramvec_\ell$, denoted by $\param_{\ell,j}$,  and $\hat{F}_{\ell-1,j}^{-1}(\cdot)$ is an estimate of the marginal distribution inverse for the $j$th dimension of $\paramvec_{\ell-1}$. The result is a sampling procedure for $(\paramvec_\ell,\tilde{\paramvec}_{\ell})$ with a positive correlation induced between pairs. 
Warne et al.~\cite{Warne2018} note that this approach only strictly satisfies the telescoping summation in each marginal rather than the full distribution of $\paramvec_L$. As a result, the sequence of discrepancies, $\epsilon_1,\ldots,\epsilon_L$, needs to be chosen so that the correlation structures between $\CondPDF{\cdot}{\discrep{\bvec{Y}_{\text{obs}}}{\bvec{Y}_s} \leq \epsilon_\ell}$  and $\CondPDF{\cdot}{\discrep{\bvec{Y}_{\text{obs}}}{\bvec{Y}_s} \leq \epsilon_{\ell-1}}$ are similar. This requires the bias correction terms in \eqref{eq:mlmcabc} to be computed in order from $\ell = 2 \to L$ due to the dependence on the previous level for the estimation of the marginal distribution inverses at level $\ell-1$ (\eqref{eq:mlmcabccouple}). 
The complete process is given in Algorithm~\ref{alg:mlmcabc}. Given an appropriate discrepancy sequence, the sample size sequence, $N_1, \ldots, N_L$, can be optimised to achieve improved convergence rates in mean squared error (see Section 6.5.1 in Warne~\cite{Warne2020b}, Theorem 3.1 in Giles~\cite{Giles2008}, and Section 3.2 in Lester et al.~\cite{Lester2014}). However, the target accuracy needs to be sufficiently small for this improved convergence rate to take effect due to an overhead computational cost in generating the trial simulations needed to optimise the sequence $N_1, \ldots, N_L$. 
\begin{algorithm}
	\caption{Multilevel Monte Carlo for ABC rejection sampling (MLMC-ABC)}
	\begin{algorithmic}[1]
		\State Initialise $\epsilon_1,\ldots, \epsilon_L$, $N_1,\dots, N_L$ and prior $\PDF{\paramvec}$.
		\For{$\ell = 1,\ldots, L$}
		\State{Sample $\paramvec_{\ell}^{1},\ldots,\paramvec_{\ell}^{N_\ell} \sim \CondPDF{\cdot}{\discrep{\bvec{Y}_{\text{obs}}}{\bvec{Y}_s} \leq \epsilon_\ell}$ using ABC rejection sampling};
		\State Set $\bar{F}_{\ell,j}(s) \leftarrow \sum_{i=1}^{N_\ell} \ind{(-\infty,s]}{\param_{\ell,j}}/N_\ell$ for $j = 1,\ldots,k$;
		\If{$\ell = 1$}
		\State $\hat{f}_{\ell} \leftarrow \sum_{i=1}^{N_\ell} f(\paramvec^{i}_{\ell})/N_\ell$;
		\Else
		\For{$i = 1, \ldots, N_\ell$}
		
		\State Set $\tilde{\paramvec}_{{\ell-1}}^{i} \leftarrow \left[\hat{F}_{{\ell-1},1}^{-1}\left(\bar{F}_{\ell,1}^{N_\ell}\left(\param_{\ell,1}^{i}\right)\right), \ldots, \hat{F}_{{\ell-1},k}^{-1}\left(\bar{F}_{\ell,k}^{N_\ell}\left(\param_{\ell,k}^{i}\right)\right)\right]$;
		\EndFor
		\State Set $\hat{F}_{\ell,j}(s) \leftarrow \hat{F}_{{\ell-1},j}(s)+ \sum_{i=1}^{N_\ell} \left[\ind{(-\infty,s]}{\param_{\ell,j}^i} -\ind{(-\infty,s]}{\tilde{\param}_{\ell-1,j}^i}\right]/N_\ell$, $j = 1, \ldots, k$;
		\State Set $\hat{f}_{\ell} \leftarrow \hat{f}_{{\ell-1}} + \sum_{i=1}^{N_\ell} \left[f(\paramvec^{i}_{\ell}) - f(\tilde{\paramvec}^{i}_{{\ell-1}})\right]/N_\ell$.
		\EndIf
		\EndFor
	\end{algorithmic}
	\label{alg:mlmcabc}
\end{algorithm}

\FloatBarrier
\subsubsection{Multifidelity rejection sampling}
An alternative approach developed by Prescott and Baker~\cite{Prescott2020} utilises the telescoping summation in a probabilistic manner akin to the de-biasing approach of Rhee and Glynn~\cite{Rhee2015}. Instead of considering a sequence of ABC samplers defined in terms of acceptance thresholds, they consider ABC rejection samples with different simulator fidelities. That is, a high fidelity simulator $\bvec{Y}_s \sim \simProc{\cdot}{\paramvec}$ that is computationally expensive, such as the Gillespie direct method, and a computationally cheaper low fidelity simulator $\tilde{\bvec{Y}}_s \sim \approxsimProc{\cdot}{\paramvec}$, such as tau-leaping approximation with time step $\tau$. The idea is to perform ABC rejection sampling with the low fidelity simulator with discrepancy $\approxdiscrep{\bvec{Y}_{\text{obs}}}{\tilde{\bvec{Y}}_s}$ and acceptance threshold $\tilde{\epsilon}$, and then perform a probabilistic bias correction that requires ABC rejection using a high fidelity simulation with discrepancy $\discrep{\bvec{Y}_{\text{obs}}}{\bvec{Y}_s}$ and acceptance threshold $\epsilon$. The resulting estimator is given by
\begin{equation}
\E{f(\paramvec)} \approx \frac{\sum_{i=1}^N w(\paramvec^i)f(\paramvec^i)}{\sum_{i=1}^N w(\paramvec^i)},
\label{eq:mfabc}
\end{equation}  
where $\paramvec^1,\ldots, \paramvec^N$ are samples from the prior and the weight function is
\begin{equation}
w(\paramvec) = \ind{(0,\tilde{\epsilon}]}{\approxdiscrep{\bvec{Y}_{\text{obs}}}{\tilde{\bvec{Y}}_s}} + \xi \left[\ind{(0,\epsilon]}{\discrep{\bvec{Y}_{\text{obs}}}{\bvec{Y}_s}} - \ind{(0,\tilde{\epsilon}]}{\approxdiscrep{\bvec{Y}_{\text{obs}}}{\tilde{\bvec{Y}}_s}}\right].
\end{equation}
Here, $\xi$ is a random variable given by  
\begin{equation}
\xi = \frac{\ind{(0,\eta(\tilde{\bvec{Y}}_s)]}{U}}{\eta(\tilde{\bvec{Y}}_s)},
\end{equation}
where $U \sim \mathcal{U}(0,1)$ and $\eta(\tilde{\bvec{Y}}_s)$ is the probability of generating a high-fidelity simulation given a realisation from the low-fidelity simulation. This so-called continuation probability can take many forms, however, the method originally proposed by Prescott and Baker~\cite{Prescott2020} is
\begin{equation}
\eta(\tilde{\bvec{Y}}_s) = \eta_1\ind{(0,\tilde{\epsilon}]}{\approxdiscrep{\bvec{Y}_{\text{obs}}}{\tilde{\bvec{Y}}_s}} + \eta_2\ind{(\tilde{\epsilon},\infty)}{\approxdiscrep{\bvec{Y}_{\text{obs}}}{\tilde{\bvec{Y}}_s}},
\end{equation}   
where $\eta_1$ and $\eta_2$ are, respectively, the continuation probabilities when the low-fidelity simulation, $\tilde{\bvec{Y}_s} \sim \approxsimProc{\cdot}{\paramvec}$, is accepted and rejected. The sampler proceeds according to Algorithm~\ref{alg:mfabc}.
\begin{algorithm}[h]
	\caption{Multifidelity ABC rejection sampling (MF-ABC)}
	\begin{algorithmic}[1]
		\State Initialise $\eta_1,\eta_2$, $N$, $\epsilon, \tilde{\epsilon}$, $\discrep{\bvec{Y}_{\text{obs}}}{\cdot}$, $\approxdiscrep{\bvec{Y}_{\text{obs}}}{\cdot}$ and prior $p(\paramvec)$;
		\For {$i = 1,2,\ldots, N$}
		\State{Sample the prior $\paramvec^i \sim p(\paramvec)$};
		\State{Simulate the low-fidelity model $\tilde{\bvec{Y}_s} \sim \approxsimProc{\cdot}{\paramvec_i}$};
		\State Set $\tilde{w} \leftarrow \ind{(0,\tilde{\epsilon}]}{\approxdiscrep{\bvec{Y}_{\text{obs}}}{\tilde{\bvec{Y}}_s}}$ and $\eta \leftarrow \eta_1 \tilde{w} + \eta_2 (1 - \tilde{w})$;
		\If {$U < \eta$ where $U \sim \mathcal{U}(0,1)$}
		\State{Simulate the high-fidelity model $\bvec{Y}_s \sim \simProc{\cdot}{\paramvec_i}$};
		\State Set $w_i \leftarrow \tilde{w} + (\ind{(0,\epsilon]}{\discrep{\bvec{Y}_{\text{obs}}}{\bvec{Y}_s}} - \tilde{w})/\eta$;
		\Else
		\State Set $w_i \leftarrow \tilde{w}$;
		\EndIf
		\EndFor
		\State Set $\hat{f} \leftarrow \sum_{i=1}^N w_i f(\paramvec^i)/\sum_{i=1}^N w_i$.
	\end{algorithmic}
	\label{alg:mfabc}
\end{algorithm}

Prescott and Baker~\cite{Prescott2020} prove that when $\eta_1 >0$ and $\eta_2 > 0$ the multifidelity estimator (\eqref{eq:mfabc}) is asymptotically unbiased. The multifidelity ABC estimator may be viewed as a form of importance sampling and thereby has a bias of order $\mathcal{O}(1/N)$ (Supplementary Material). Unlike standard rejection and importance samplers, the multifidelity weights can be negative. Despite this fact, the multifidelity ABC estimator is both asymptotically unbiased and consistent~\cite{Prescott2021b}. Under certain conditions $\eta_1, \eta_2$ may be optimised such that the multifidelity ABC rejection sampler (Algorithm~\ref{alg:mfabc}) is more computationally efficient than direct rejection sampling with the high fidelity simulator (Algorithm~\ref{alg:ABC-rej}).  This improvement represents a decrease in the average computational cost for a given target mean-square error, but it does not improve the convergence rate. However, the multifidelity approach will not incur a significant overhead for selecting the continuation probabilities since adaptive schemes can be applied~\cite{Prescott2021b} (Supplementary Material).

\subsection{Multifidelity MLMC for ABC inference}

The MLMC and multifidelity approaches to ABC inference, MLMC-ABC and MF-ABC,  obtain computational improvements in distinct ways that are also complementary. MLMC-ABC (Algorithm~\ref{alg:mlmcabc}) combines a sequence of ABC rejection samplers with different discrepancy thresholds~\cite{Warne2018} while the stochastic simulation scheme is fixed and must be unbiased. Conversely, MF-ABC (Algorithm~\ref{alg:mfabc}) combines two stochastic simulation schemes~\cite{Prescott2020} with the other elements of the ABC sampler remaining largely unchanged. Computationally, MLMC-ABC empirically improves the convergence rate of the mean-square error as $\epsilon \to 0$ at the expense of a tuning step that incurs an additional cost~\cite{Warne2018}. This is consistent with theoretical and empirical results from other MLMC applications~\cite{Efendiev2015,Guha2017,Jasra2019}. In addition, MF-ABC reduces the average simulation cost without improvements in the convergence rate~\cite{Prescott2020,Prescott2021b}. Our novel contribution in this work is to show how these methods can be combined to exploit the computational advantages of both.

We now derive our new method, called \emph{multifidelity MLMC for ABC}~(MF-MLMC-ABC). Assume we have two sequences of $L$ ABC acceptance thresholds, $\{\epsilon_\ell\}_{\ell = 1}^{\ell = L}$ and $\{\tilde{\epsilon}_\ell\}_{\ell = 1}^{\ell = L}$, with $\epsilon_\ell > \epsilon_{\ell+1}$ and $\tilde{\epsilon}_\ell > \tilde{\epsilon}_{\ell+1}$ for all $\ell = 1,\ldots, L-1$, and a sequence of $L$ time-step lengths, $\{\tau_\ell\}_{\ell = 1}^{\ell = L}$. Note that there are no constraints on the relation between the two acceptance thresholds sequences, nor any requirement that the sequence of time-step lengths is strictly decreasing or even monotonic. Given a set of model parameters, $\paramvec \in \paramspace$, let $\bvec{Y}_s \sim \simProc{\cdot}{\paramvec}$ and $\bvec{Y}^\tau_s \sim \approxsimProct{\cdot}{\paramvec}{\tau}$ denote, respectively, exact stochastic simulation (i.e., using the Gillespie direct method~\cite{Gillespie1977}) and approximate stochastic simulation with time-step $\tau$ (i.e., using the tau-leaping method~\cite{Gillespie2001}). Finally, let $\discrep{\bvec{Y}_{\text{obs}}}{\cdot}$  and $\approxdiscrept{\bvec{Y}_{\text{obs}}}{\cdot}{\tau}$ be discrepancy metrics used to compare observed data $\bvec{Y}_{\text{obs}}$ with, respectively, exact simulation output, $\bvec{Y}_s$, or approximate simulation output, $\bvec{Y}^\tau_s$.

\begin{algorithm}[h]
	\caption{Multifidelity multilevel Monte Carlo for ABC rejection sampling (MF-MLMC-ABC)}
	\begin{algorithmic}[1]
		\State Initialise $\{\epsilon_\ell\}_{\ell = 1}^{\ell = L}$, $\{\tilde{\epsilon}_\ell\}_{\ell = 1}^{\ell = L}$, $\{\tau_\ell\}_{\ell = 1}^{\ell = L}$, $\{N_\ell\}_{\ell = 1}^{\ell = L}$,$\{(\eta_{\ell,1},\eta_{\ell,2})\}_{\ell = 1}^{\ell = L}$, $\discrep{\bvec{Y}_{\text{obs}}}{\cdot}$, $\approxdiscrept{\bvec{Y}_{\text{obs}}}{\cdot}{\tau_\ell}$ and prior $\PDF{\paramvec}$;
		\For{$\ell = 1,\ldots, L$}
		\For {$i = 1,2,\ldots, N_\ell$}
		\State{Sample the prior $\paramvec^i_\ell \sim p(\paramvec)$};
		\State{Simulate the low-fidelity model with $\tau = \tau_\ell$, $\bvec{Y}^{\tau_\ell}_s \sim \approxsimProct{\cdot}{\paramvec^i_\ell}{\tau_\ell}$};
		\State Set $\tilde{w} \leftarrow \ind{(0,\tilde{\epsilon}_\ell]}{\approxdiscrept{\bvec{Y}_{\text{obs}}}{\bvec{Y}^{\tau_\ell}_s}{\tau_\ell}}$ and $\eta^{\tau_\ell} \leftarrow \eta_{\ell,1} \tilde{w} + \eta_{\ell,2} (1 - \tilde{w})$;
		\If {$U < \eta^{\tau_\ell}$ where $U \sim \mathcal{U}(0,1)$}
		\State{Simulate the high-fidelity model $\bvec{Y}_s \sim \simProc{\cdot}{\paramvec^i_\ell}$};
		\State Set $w_\ell^i \leftarrow \tilde{w} + (\ind{(0,\epsilon]}{\discrep{\bvec{Y}_{\text{obs}}}{\bvec{Y}_s}} - \tilde{w})/\eta^{\tau_\ell}$;
		\Else
		\State Set $w_\ell^i \leftarrow \tilde{w}$;
		\EndIf
		\EndFor
		\State{Set $W_\ell = 1/\sum_{i=1}^{N_\ell} w_\ell^i$};
		\State Set $\bar{F}_{\ell,j}(s) \leftarrow W_\ell\sum_{i=1}^{N_\ell} w_\ell^i \ind{(-\infty,s]}{\param_{\ell,j}}$ for $j = 1,\ldots,k$;
		\If{$\ell = 1$}
		\State $\hat{f}_{\ell} \leftarrow W_\ell \sum_{i=1}^{N_\ell} w_\ell^i f(\paramvec^{i}_{\ell})$;
		\Else
		\For{$i = 1, \ldots, N_\ell$}
		
		\State Set $\tilde{\paramvec}_{{\ell-1}}^{i} \leftarrow \left[\hat{F}_{{\ell-1},1}^{-1}\left(\bar{F}_{\ell,1}^{N_\ell}\left(\param_{\ell,1}^{i}\right)\right), \ldots, \hat{F}_{{\ell-1},k}^{-1}\left(\bar{F}_{\ell,k}^{N_\ell}\left(\param_{\ell,k}^{i}\right)\right)\right]$;
		\EndFor
		\State Set $\hat{F}_{\ell,j}(s) \leftarrow \hat{F}_{{\ell-1},j}(s)+ W_\ell \sum_{i=1}^{N_\ell}  w_\ell^i \left[\ind{(-\infty,s]}{\param_{\ell,j}^i} -\ind{(-\infty,s]}{\tilde{\param}_{\ell-1,j}^i}\right]$, $j = 1, \ldots, k$;
		\State Set $\hat{f}_{\ell} \leftarrow \hat{f}_{{\ell-1}} + W_\ell\sum_{i=1}^{N_\ell} w_\ell^i \left[f(\paramvec^{i}_{\ell}) - f(\tilde{\paramvec}^{i}_{{\ell-1}})\right]$.
		\EndIf
		\EndFor
	\end{algorithmic}
	\label{alg:mfmlmcabc}
\end{algorithm} 

Given the above notation, for any $\ell = 1,\ldots, L$ we can apply MF-ABC sampling using weights,
\begin{equation}
w^{\tau_\ell}(\paramvec_\ell) = \ind{(0,\tilde{\epsilon}_\ell]}{\approxdiscrept{\bvec{Y}_{\text{obs}}}{\bvec{Y}^{\tau_\ell}_s}{\tau_\ell}} + \frac{\ind{(0,\eta^{\tau_\ell}(\bvec{Y}^{\tau_\ell}_s)]}{U}}{\eta^{\tau_\ell}(\bvec{Y}^{\tau_\ell}_s)} \left[\ind{(0,\epsilon_\ell]}{\discrep{\bvec{Y}_{\text{obs}}}{\bvec{Y}_s}} - \ind{(0,\tilde{\epsilon}_\ell]}{\approxdiscrept{\bvec{Y}_{\text{obs}}}{\bvec{Y}^{\tau_\ell}_s}{\tau_\ell}}\right],
\label{eq:mfmlmcw}
\end{equation} 
with continuation probability function
\begin{equation}
\eta^{\tau_\ell}(\bvec{Y}^{\tau_\ell}_s) = 
\eta_{\ell,1}\ind{(0,\tilde{\epsilon}_\ell]}{\approxdiscrept{\bvec{Y}_{\text{obs}}}{\bvec{Y}^{\tau_\ell}_s}{\tau_\ell}} + \eta_{\ell,2}\ind{(\tilde{\epsilon}_\ell,\infty)}{\approxdiscrept{\bvec{Y}_{\text{obs}}}{\bvec{Y}^{\tau_\ell}_s}{\tau_\ell}},
\end{equation}
for constants $\eta_{\ell,1} > 0$, $\eta_{\ell,2} > 0$ for each $\ell = 1,\ldots, L$ and $U \sim \mathcal{U}(0,1)$. Based on the results of Prescott and Baker~\cite{Prescott2020}, an expectation estimated this way would be aysmptotically unbiased with respect to the ABC posterior under the exact stochastic simulation and discrepancy measure, that is,  $\paramvec_\ell \sim \CondPDF{\cdot}{\discrep{\bvec{Y}_{\text{obs}}}{\bvec{Y}_s} \leq \epsilon_\ell}$. This provides a connection to the MLMC-ABC telescoping summation in \eqref{eq:mlmcabc}. Therefore, we can apply the MF-ABC estimator (\eqref{eq:mfabc}) using weights defined by \eqref{eq:mfmlmcw} to each of the $L$ terms in the MLMC-ABC telescoping summation (\eqref{alg:mlmcabc}) and thereby arrive at the MF-MLMC-ABC estimator,
\begin{equation}
\E{f(\paramvec_L)} \approx \hat{f} = \sum_{\ell=1}^L\sum_{i=1}^{N_\ell} \frac{w^{\tau_\ell}(\paramvec_\ell^i) g_\ell(\paramvec_\ell^i)}{\sum_{j=1}^{N_\ell}w^{\tau_\ell}(\paramvec_\ell^j)},
\label{eq:mfmlmcabc}
\end{equation}
where 
\begin{equation}
g_\ell(\paramvec_\ell^i) = \begin{cases}
f(\paramvec_\ell^i) & \text{if }\ell = 1 \\
f(\paramvec_\ell^i) - f(\tilde{\paramvec}_{\ell-1}^i) & \text{if }\ell > 1
\end{cases},
\end{equation}
and $\tilde{\paramvec}_{\ell-1}^i$ is constructed from $\paramvec_\ell^i$ and estimated marginal distribution functions obtained from the previous $\ell -1$ terms, as given in \eqref{eq:mlmcabccouple}, to implement an approximate coupling between levels~\cite{Warne2019a,Warne2018}. Due to the properties of MF-ABC, \eqref{eq:mfmlmcabc} is an asymptotically unbiased estimator of \eqref{eq:mlmcabc}~\cite{Prescott2020} and therefore an asymptotically unbiased estimator of $\E{f(\paramvec_L)}$ up to the approximate coupling scheme~\cite{Warne2018}.

We therefore arrive at the MF-MLMC-ABC method presented in Algorithm~\ref{alg:mfmlmcabc}. Note that the proposed approach, just as with MLMC-ABC and MF-ABC, requires a number of algorithmic hyperparameters to be selected appropriately to ensure efficient sampling. In the next section we discuss theoretical results that guide how these parameters should be selected. 

\FloatBarrier

\subsection{Optimal algorithm configuration}

There are several important algorithmic hyperparameters that must be appropriately chosen to practically apply the MF-MLMC-ABC method. Each of these has various influences on the accuracy and performance of the MF-MLMC-ABC method. In this section, we will step through each of these algorithmic hyperparameters and provide theoretical results to optimally configure the method. The algorithmic hyperparameters that require optimisation are the number of levels, $L$, the form of the sequences $\{\epsilon_\ell\}_{\ell=1}^{\ell=L}$  and $\{\tau_\ell\}_{\ell=1}^{\ell=L}$, the sequence of samples to draw from each level, $\{N_\ell\}_{\ell=1}^{\ell=L}$, and the sequence of continuation probabilities, $\{(\eta_{\ell,1},\eta_{\ell,2})\}_{\ell = 1}^{\ell = L}$. Of these, $L$ and $\{\tau_\ell\}_{\ell=1}^{\ell=L}$ need to be selected heuristically, however, for any given $L$ and $\{\tau_\ell\}_{\ell=1}^{\ell=L}$, the sequences  $\{N_\ell\}_{\ell=1}^{\ell=L}$ and $\{(\eta_{\ell,1},\eta_{\ell,2})\}_{\ell = 1}^{\ell = L}$ may be optimised. 

For a given level $\ell$ and assuming $\tau_\ell$ is selected, the optimal $\eta_{\ell,1}$ and $\eta_{\ell,2}$ can be determined through optimising the limiting efficiency as the number of samples $N_\ell \to \infty$. Prescott and Baker~\cite{Prescott2020} show that this corresponds to minimising the function
\begin{equation}
\phi(\eta_{\ell,1},\eta_{\ell,2} ; g_\ell) = \E{w^{\tau_\ell}(\paramvec_\ell)^2(g_\ell(\paramvec_\ell) - \E{g_\ell(\paramvec_\ell)})^2}\E{C_\ell(\paramvec_\ell)},
\label{eq:phi}
\end{equation}
where $C_\ell(\paramvec_\ell)$ is the cost of computing the weight $w^{\tau_\ell}(\paramvec_\ell)$ (\eqref{eq:mfmlmcw}). If the acceptance state of an approximate simulation is interpreted as a classifier for the predicted acceptance state for an exact simulation, and the true positive rate exceeds the false positive rate, then \eqref{eq:phi} can be minimised for $(\eta_{\ell,1},\eta_{\ell,2}) \in (0,1]^2$ (See Lemmas 4.2 and 4.3 in Prescott and Baker~\cite{Prescott2020}). The optimal continuation probabilities are given by (See Corollary 4.4 in Prescott and Baker~\cite{Prescott2020})
\begin{equation}
(\eta_{\ell,1}^*,\eta_{\ell,2}^*) = \begin{cases}
\left(\sqrt{\dfrac{R_p^\ell}{R_0^\ell}},\sqrt{\dfrac{R_n^\ell}{R_0^\ell}}\right), & \text{if } \max{\{R_p^\ell,R_n^\ell\}} \leq R_0^\ell, \\
\left(1,\bar{\eta}_{\ell,2}\right), & \text{if } \max{\{R_p^\ell,R_n^\ell\}} > R_0^\ell \text{ and } \phi(1,\bar{\eta}_{\ell,2}) \leq \phi(\bar{\eta}_{\ell,1},1), \\
\left(\bar{\eta}_{\ell,1},1\right), & \text{otherwise},
\end{cases}
\label{eq:opteta}
\end{equation} 
where
\begin{equation}
\begin{split}
R_p^\ell = \frac{p^\ell_{fp} \E{c^{\tau_\ell}(\paramvec_\ell)}}{c^\ell_p}, \quad  R_n^\ell = \frac{p^\ell_{fn} \E{c^{\tau_\ell}(\paramvec_\ell)}}{c^\ell_n}, \quad R_0^\ell = p^\ell_{tp} - p^\ell_{fp}, \\  \bar{\eta}_{\ell,1} = \min\left\{1,\sqrt{\frac{R_p^\ell + p^\ell_{fp}c^\ell_n/c^\ell_p}{R^\ell_0 + p^\ell_{fn}}}\right\}, \quad \bar{\eta}_{\ell,2} = \min\left\{1,\sqrt{\frac{R_n^\ell + p^\ell_{fn}c^\ell_p/c^\ell_n}{R^\ell_0 + p^\ell_{fp}}}\right\},\label{eq:Retabar}
\end{split}
\end{equation}
and

\begin{equation}
\begin{split}
p^\ell_{tp} &= \E{\ind{(0,\epsilon_\ell]}{\discrep{\bvec{Y}_{\text{obs}}}{\bvec{Y}_s}}\ind{(0,\tilde{\epsilon}_\ell]}{\approxdiscrept{\bvec{Y}_{\text{obs}}}{\bvec{Y}^{\tau_\ell}_s}{\tau_\ell}}(g_\ell(\paramvec_\ell) - \E{g_\ell(\paramvec_\ell)})^2},\\
p^\ell_{fp} &= \E{\ind{(\epsilon_\ell,\infty]}{\discrep{\bvec{Y}_{\text{obs}}}{\bvec{Y}_s}}\ind{(0,\tilde{\epsilon}_\ell]}{\approxdiscrept{\bvec{Y}_{\text{obs}}}{\bvec{Y}^{\tau_\ell}_s}{\tau_\ell}}(g_\ell(\paramvec_\ell) - \E{g_\ell(\paramvec_\ell)})^2},\\
p^\ell_{fn} &= \E{\ind{(0,\epsilon_\ell]}{\discrep{\bvec{Y}_{\text{obs}}}{\bvec{Y}_s}}\ind{(\tilde{\epsilon}_\ell,\infty]}{\approxdiscrept{\bvec{Y}_{\text{obs}}}{\bvec{Y}^{\tau_\ell}_s}{\tau_\ell}}(g_\ell(\paramvec_\ell) - \E{g_\ell(\paramvec_\ell)})^2},\\
c^\ell_p &= \CondE{c(\paramvec_\ell)}{\approxdiscrept{\bvec{Y}_{\text{obs}}}{\bvec{Y}^{\tau_\ell}_s}{\tau_\ell} \leq \tilde{\epsilon}_\ell}, \\
c^\ell_n &= \CondE{c(\paramvec_\ell)}{\approxdiscrept{\bvec{Y}_{\text{obs}}}{\bvec{Y}^{\tau_\ell}_s}{\tau_\ell} > \tilde{\epsilon}_\ell}.\\
\end{split}
\label{eq:pcs}
\end{equation}
In Equations (\ref{eq:Retabar}) and (\ref{eq:pcs}), $c(\paramvec_\ell)$ and $c^{\tau_\ell}(\paramvec_\ell)$ denote, respectively, the cost of generating an exact realisation, $\bvec{Y}_s \sim \simProc{\cdot}{\paramvec_\ell}$, and an approximate realisation, $\bvec{Y}_s^{\tau_\ell} \sim \approxsimProct{\cdot}{\paramvec_\ell}{\tau_\ell}$.

To optimise the sequence of samples $\{N_\ell\}_{\ell=1}^{\ell= L}$, we aim to minimise the total expected computational cost of computing the MF-MLMC-ABC estimator, that is, $\E{C(\hat{f})} = \sum_{\ell=1}^L N_\ell\E{C_\ell(\paramvec_\ell)}$, subject to the constraint $\V{\hat{f}} \propto h^2$ where $h^2$ is the target variance. Using a Lagrange multiplier, it can be shown (See Giles~\cite{Giles2008}, Lester et al.~\cite{Lester2014}, and Warne et al.~\cite{Warne2018}) that the following scaling is optimal,
\begin{equation}
N_\ell \propto h^{-2}\sqrt{\frac{v_\ell}{c_\ell}}\sum_{m=1}^{L}\sqrt{v_m c_m},\quad \text{for } \ell = 1,2, \ldots, L,
\label{eq:mloptN}
\end{equation} 
where $c_\ell = \E{C_\ell(\paramvec_\ell)}$ and $v_\ell = \E{w^{\tau_\ell}(\paramvec_\ell)^2(g_\ell(\paramvec_\ell) - \E{g_\ell(\paramvec_\ell)})^2} / \E{w^{\tau_\ell}(\paramvec_\ell)}^2$.

Relative to ABC rejection sampling, the MF-MLMC-ABC method is asymptotically unbiased, with the bias at level $\ell$ being of the order importance sampling $\mathcal{O}(1/N_\ell)$ (Supplementary Material). The effect of this bias to be considered along with the optimal sample size scaling in \eqref{eq:optNphi}, especially for the terminal level $L$. That is, we require the bias due to MF-ABC sampling to be small compared with the bias incurred from the ABC-based approximations, which are $\mathcal{O}(\epsilon)$. This is difficult to test in practice, however, it is a common feature of any ABC method based on importance sampling.


\subsection{Practical algorithm tuning}

A practical choice for the selection of most components of MF-MLMC-ABC is immediately available from the target ABC-based inference problem, for example, the exact stochastic simulation process $\bvec{Y}_s \sim \simProc{\cdot}{\paramvec}$, the prior probability density, $\PDF{\paramvec}$, the discrepancy metric $\discrep{\bvec{Y}_{\text{obs}}}{\cdot}$, and the target acceptance threshold $\epsilon_L$. Other choices are easily motivated. The largest acceptance threshold, $\epsilon_1$, can be chosen so that $\CondPDF{\paramvec}{\discrep{\bvec{Y}_{\text{obs}}}{\bvec{Y}_s} \leq \epsilon_1} \approx \PDF{\paramvec}$. The approximate stochastic simulation scheme, $\bvec{Y}_s^\tau \sim \approxsimProct{\cdot}{\paramvec}{\tau}$, could be chosen from a range of candidates, but a first-order method such as tau-leaping will be an appropriate default choice in many cases. Given a tau-leaping scheme for the approximate model, it will often be appropriate to take $\approxdiscrept{\bvec{Y}_{\text{obs}}}{\cdot}{\tau} = \discrep{\bvec{Y}_{\text{obs}}}{\cdot}$, and therefore $\tilde{\epsilon}_\ell = \epsilon_\ell$ for all $\ell = 1,\ldots,L$ is applicable. Of course, there is freedom and flexibility in all of these choices, and the discussion section highlights some potentially useful alternative strategies.

Next the number of levels, $L$, needs to be determined. Unfortunately, there is no general theory for this choice. However one common approach from the MLMC literature is to consider a geometric sequence $\epsilon_\ell = \epsilon_1 m^{-\ell+1}$ then set $L = 1 - \log_m(\epsilon_L/\epsilon_1)$. Some heuristics do exist to determine the appropriate scale factor $m > 1$. For MLMC with SDEs Giles~\cite{Giles2008} demonstrated $m = 4$ is optimal, however, due to the approximate coupling scheme for MLMC-ABC, Warne et al.,~\cite{Warne2019a,Warne2018} propose $m \in [1.5,2.5]$ as a practical choice for inference. 

Choice of the sequence of time-steps, $\{\tau_\ell\}_{\ell=1}^{\ell=L}$, can be guided by Equations (\ref{eq:opteta})--(\ref{eq:pcs}). Firstly, we wish $\tau_\ell$ to be small enough to exhibit low false positive and false negative rates; this will result in smaller optimal continuation probabilities (\eqref{eq:Retabar}) and reduce the number of times the high-fidelity model is simulated~\cite{Prescott2020}. However, the speed-up factor for low-fidelity simulations over high-fidelity simulations needs to be sufficiently high, otherwise there is not enough of a computational benefit~\cite{Prescott2021b}. Therefore, $\tau_\ell$ cannot be arbitrary small. To tune $\tau_\ell$, some experimentation is required to compare the differences computation time and the acceptance state between pairs of exact and approximate simulations. Fortunately, it is possible to identify poor choices of $\tau_\ell$, since the adaptive tuning scheme for the optimal continuation probabilities will be unable to improve upon standard ABC rejection sampling. It is also important to note that $\{\tau_\ell\}_{\ell=1}^{\ell=L}$ need not be a strictly decreasing or even monotonic sequence, and we demonstrate in the results section that $\tau_1 = \tau_2 = \cdots = \tau_L$ works quite well for many applications. Furthermore, introducing a coupling scheme between the exact and approximate simulations~\cite{Prescott2020,Anderson2012,Lester2014} can reduce mis-classification rates of larger values of $\tau_\ell$ and lead to improved performance.

In practice, \eqref{eq:opteta} is solved for optimal continuation probabilities through the generation of initial trial samples~\cite{Prescott2020}. In this work, we extend this through adaptive updates to $(\eta_{\ell,1},\eta_{\ell,2})$ while generating samples at level $\ell$. For the trial samples, we generate $M$ samples using Algorithm~\ref{alg:mfabc} with $\eta_{\ell,1} = \eta_{\ell,2} = 1$, then initial estimates of the expectations in \eqref{eq:Retabar} are produced through direct Monte Carlo estimates. Next, gradient descent is applied to iteratively update $\eta_{\ell,1}$ and $\eta_{\ell,2}$ toward the optimum in \eqref{eq:opteta} while also refining the estimates in Equations (\ref{eq:Retabar}). While it is possible to iteratively refine the solution to \eqref{eq:opteta} directly as the sampling proceeds~\cite{Prescott2020}, we find that this is extremely sensitive to the initial estimates. Therefore we utilise our adaptive gradient descent MF-ABC sampler based on recent developments~\cite{Prescott2021b} as a robust alternative to Algorithm~\ref{alg:mfabc} (Supplementary Material).

Finally, to apply MF-MLMC-ABC the sequence of sample numbers, $\{N_\ell\}_{\ell=1}^{\ell= L}$, are needed. Fortunately, we can rewrite \eqref{eq:mloptN} as
\begin{equation}
N_\ell \propto \frac{1}{h^{-2}} \frac{\sqrt{\phi(\eta_{\ell,1},\eta_{\ell,2}; g_\ell)}}{\E{C_\ell(\paramvec_\ell)}\E{w^{\tau_\ell}(\paramvec_\ell)}} \sum_{m=1}^L \frac{\sqrt{\phi(\eta_{m,1},\eta_{m,2}; g_m)}}{\E{w^{\tau_m}(\paramvec_m)}},
\label{eq:optNphi}
\end{equation}
to highlight that the optimal sequence of sample numbers, $\{N_\ell\}_{\ell=1}^{\ell = L}$, is dependent on the optimal continuation probabilities obtained using the adaptive gradient descent MF-ABC scheme at each level. Therefore, we can estimate the terms required for optimal $\{N_\ell\}_{\ell=1}^{\ell = L}$ directly from the same trial samples used to estimate optimal $\{(\eta_{\ell,1},\eta_{\ell,2})\}_{\ell=1}^{\ell=L}$.

\subsection{Summary}

The MF-ABC and MLMC-ABC methods exploit the multilevel telescoping summation, and properties of approximate stochastic simulation and ABC sampling in distinct and complementary ways. In the first instance, MF-ABC can improve on the expected cost of stochastic simulation in ABC rejection sampling using the randomized bias correction term~\cite{Prescott2020,Rhee2015}. MLMC-ABC assumes exact stochastic simulation for ABC rejection sampling, but applies MLMC techniques to a sequence of correlated samplers with discrepancy thresholds $\epsilon_1 > \epsilon_2 > \cdots > \epsilon_L$ with the effect of improving the convergence rate~\cite{Jasra2019,Warne2018}. We develop a new method, MF-MLMC-ABC, that results from the application of MF-ABC sampling for each of the terms in the MLMC-ABC telescoping summation. In the next section we demonstrate how to tune these methods practically and show the computational benefits of both can be exploited to achieve improvements of two orders of magnitude.

\FloatBarrier

\section{Results}

Using a variety of biologically relevant stochastic biochemical reaction network models, we demonstrate the substantial computational improvements using our new MF-MLMC-ABC method (Algorithm~\ref{alg:mfmlmcabc}). First we consider the properties of the MLMC-ABC (Algorithm~\ref{alg:mlmcabc}) and MF-ABC (Algorithm~\ref{alg:mfabc}) methods for a fundamental biochemical building block, the Michaelis--Menten model, to show practically how to tune these methods. Then we apply these guidelines to tune our new MF-MLMC-ABC method for the repressilator gene regulatory network to show the performance benefits over MF-ABC and MLMC-ABC. Finally, we perform a realistic test on the two-step MAPK cascade network that is of fundamental importance in cell biology. For two computationally challenging networks we show that MF-MLMC-ABC effectively combines the advantages of both MF-ABC and MLMC-ABC to accelerate ABC rejection sampling by two orders of magnitude.  

In the sections that follow, we focus on the performance improvements of our new approach for the purpose of estimating posterior means of unknown parameters. However, it is important to note that our approach, arising from \eqref{eq:mfmlmcabc}, handles expectations of an arbitrary function of the unknown parameters, $f(\paramvec)$. This could include posterior probabilities or densities, or even central moments of the posterior predictive distribution, that is the mean and variance of the biochemical network state $\bvec{X}_T$ given $\paramvec$. We demonstrate the use of MF-MLMC-ABC for estimation of marginal densities for the two-stage MAPK example.




\subsection{Initial explorations of MF-ABC and MLMC-ABC: Michaelis--Menten kinetics}

Using a stochastic network of Michaelis--Menten enzyme kinetics~\cite{Michaelis1913,Rao2003}, we demonstrate the essential requirements and computational benefits of the MF-ABC and MLMC-ABC methods in order to inform the configuration of MF-MLMC-ABC for more challenging networks. The Michaelis--Menten enzyme kinetics model 
describes the catalytic conversion of a substrate, $S$, into a product, $P$, via an enzymatic reaction involving enzyme, $E$, 
\begin{equation}
\label{eq:michment}
E + S \overset{k_1}{\rightarrow} [ES], \quad [ES] \overset{k_2}{\rightarrow} E + S, \quad [ES] \overset{k_3}{\rightarrow} E+ P,
\end{equation}
with kinetic rate parameters, $k_1$, $k_2$, and $k_3$.
Biologically, this network is of interest since many intracellular processes are built from Michaelis--Menten sub-components. Computationally, the Michaelis--Menten model is a minimal example of a network without a closed-form solution to the CME, however, with only three rate parameters and four chemical species, ABC inference is feasible even with rejection sampling~\cite{Warne2019a}.

\begin{figure}[h]
	\centering
	\includegraphics[width=0.6\linewidth]{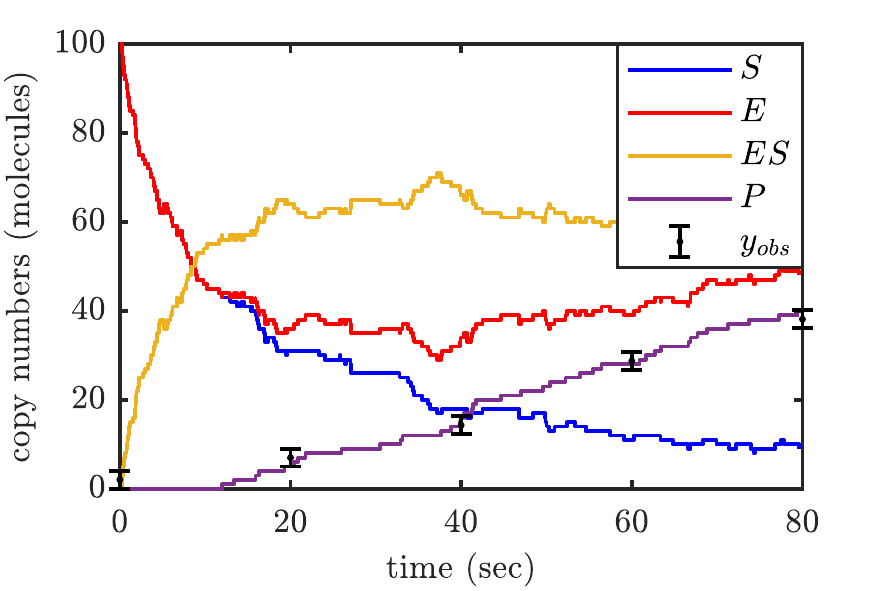}
	\caption{Example realisation of the Michaelis--Menten model along with noisy observations of the product molecules $y_{\text{obs}}(t) \sim \mathcal{N}(P_t, \sigma^2)$ (error bars indicate $y_{\text{obs}}(t) \pm \sigma$). Here, the initial condition is $E_0 = S_0 = 1000$, the true rate parameters are $k_1 = 0.001$, $k_2 = 0.005$ and $k_3 = 0.01$, and observations are taken at $t_i = 20i$, for $i = 1, 2, 3, 4$ with standard deviation $\sigma = 10$.   }
	\label{fig:michmentdata}
\end{figure}

For our simulated data we consider the realistic scenario where only the product molecules are directly observed (typically via fluorescent tagging of proteins~\cite{Finkenstadt2008,Iafolla2008,Bajar2016}), that is,
\begin{equation}
y_{\text{obs}}(t)  \sim \mathcal{N}(P_t,\sigma^2),
\label{eq:michmentobs}
\end{equation} 
where $\sigma$ is the standard deviation of the additive Gaussian observation noise. In real applications, especially for low copy numbers, it may be more appropriate to consider multinomial noise or multiplicative Gaussian noise~\cite{Simpson2020,Georgii2012}, however, for the purposes of the numerical experiments we present here, additive Gaussian noise is perfectly reasonable. Figure~\ref{fig:michmentdata} shows an example realisation of the Michaelis--Menten model with simulated observations indicated at discrete times, $t_1 = 20, t_2 = 40, t_3 = 60$, and $t_4 = 80$.

Using the Michaelis--Menten model and the noisy partial observations, we explore the effects of varying the parameters $L$ and $\tau$ on the performance of MLMC-ABC and MF-ABC, respectively, with a tau-leaping method assumed for the approximate simulation scheme. To this end, we consider the ABC inference problem,
\begin{equation}
\CondE{k_3}{\bvec{Y}_{\text{obs}}} \approx \int_{\mathbb{R}^3} k_3 \CondPDF{\paramvec}{\discrep{\bvec{Y}_{\text{obs}}}{\bvec{Y}_s} \leq \epsilon} \, \text{d}\paramvec,
\end{equation}
where $\paramvec = [k_1,k_2,k_3]$ is the vector of unknown rate parameters, $\bvec{Y}_{\text{obs}} = [y_{\text{obs}}(t_1),\ldots, y_{\text{obs}}(t_4)]$ is noisy observations of product copy numbers at discrete times $t_1, \ldots, t_4$, $\bvec{Y}_s \sim \simProc{\cdot}{\paramvec}$ is simulated data of the Michaelis--Menten model using the Gillespie direct method and simulating the observation process (\eqref{eq:michmentobs}), $\epsilon$ is the discrepancy threshold, and the discrepancy metric is $\discrep{\bvec{Y}_{\text{obs}}}{\bvec{Y}_s} = \|\bvec{Y}_{\text{obs}} -\bvec{Y}_s\|_2$ where $\| \cdot,\|_2$ is the Euclidean norm. Independent uniform priors are used with $k_1 \sim \mathcal{U}(0,0.003)$, $k_2 \sim \mathcal{U}(0,0.0015)$, and $k_3 \sim \mathcal{U}(0,0.05)$. 

\begin{figure}[h]
	\centering
	\includegraphics[width=0.8\linewidth]{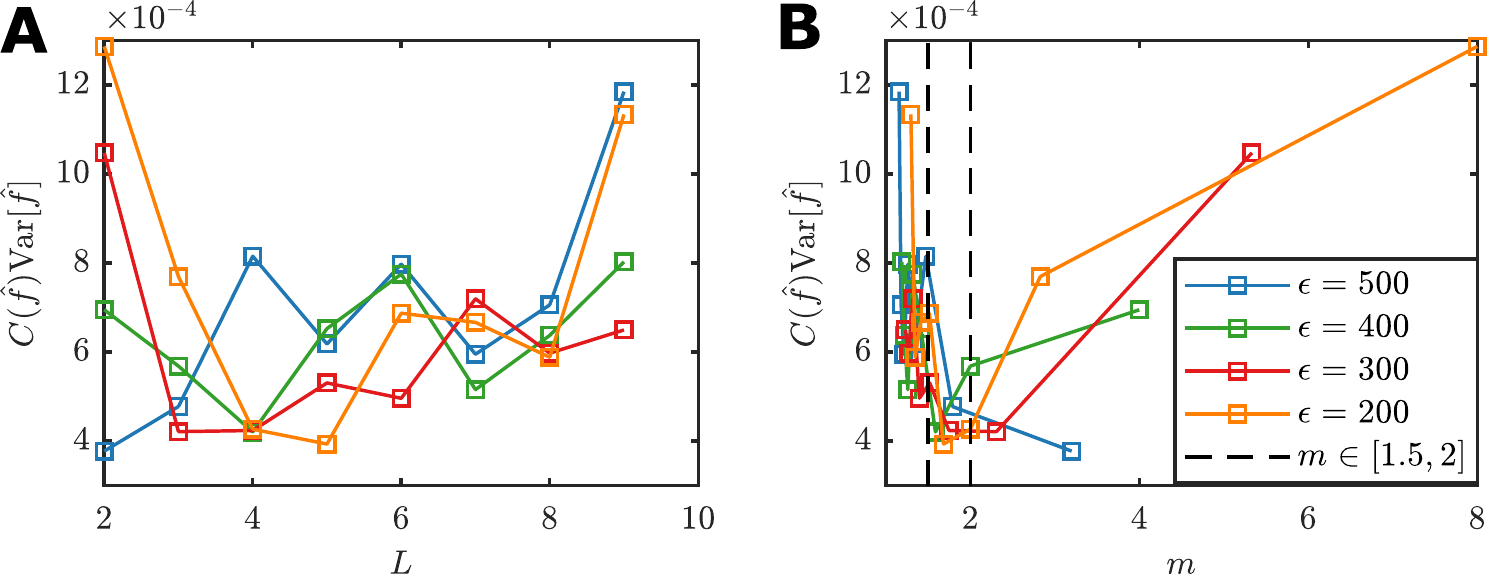}
	\caption{(A) The effect of the number of levels, $L$, in MLMC-ABC on the sampling efficiency (as determined by the product of computational cost and the estimator variance). (B) Equivalent plot in terms of $m = (\epsilon_1/\epsilon_L)^{1/(L-1)}$ to highlight that this product is minimised for the scaling $m \in [1.5,2]$ (black dashed lines) regardless of the target epsilon. Target discrepancy thresholds are indicated by dashed lines and in all cases $\epsilon_1 = 1600$.}
	\label{fig:tunemichment}
\end{figure}

We first explore MLMC-ABC (Algorithm~\ref{alg:mlmcabc}) in terms of the effect of the number of levels, $L$, using data generated with \eqref{eq:michmentobs} with $\sigma = 2$ at discrete times $t_1 = 20, t_2 = 40, t_3 = 60$, and $t_4 = 80$ using a single realisation given initial conditions $E_0 = S_0 = 1000$ and $[ES]_0 = P_0 = 0$, and rate parameters $k_1 = 0.001$, $k_2 = 0.005$ and $k_3 = 0.01$. We use $\epsilon_1 = 1600$ and apply MLMC-ABC for different numbers of levels $L \in [2,3,\ldots,9]$ and different target discrepancy thresholds $\epsilon_L = \epsilon \in [200, 300, \ldots, 600]$. Optimal sample numbers $\{N_\ell\}_{\ell=1}^{\ell=L}$ are obtained using a Lagrange multiplier approach~\cite{Giles2008,Lester2014,Warne2018} and rescaled so that $N_L = 16$. All stochastic simulations, $\bvec{Y}_s \sim \simProc{\cdot}{\paramvec}$, are exact using the Gillespie direct method (Algorithm~\ref{alg:essa}). Figure~\ref{fig:tunemichment} shows how the trade-off between computational cost and estimator variance is affected by the number of levels, $L$. The trends in terms of $L$ (Figure~\ref{fig:tunemichment}(A)) are not very meaningful because for a fixed $L$ the scale factor $m$ will be different for each target threshold. However, when the same data are presented in terms of the scale factor, $m$ (Figure~\ref{fig:tunemichment}(B)), the optimal choice is consistently contained within $m \in [1.5,2]$. This result aligns closely with the previous work on MLMC-ABC~\cite{Warne2018}. Therefore, we conclude that a good heuristic for a given $\epsilon_1$ and $\epsilon_L$ is to choose $m$ within this interval such that $L = 1 - \log_m(\epsilon_L/\epsilon_1)$ is a positive integer.

Next we look at the more nuanced problem of tuning the tau-leap time-step, $\tau$, in the context of MF-ABC (Algorithm~\ref{alg:mfabc}) for performance. Using the same data configuration and ABC problem definition as for MLMC-ABC we apply MF-ABC for different time-steps $\tau \in [0.005,0.01,0.02,0.04,0.08,0.16,0.32,0.64,1.28]$ and for the same set of target discrepancy thresholds as for MLMC-ABC. We also take $\tilde{\epsilon} = \epsilon$, and $\approxdiscrept{\bvec{Y}_{\text{obs}}}{\bvec{Y}^\tau_s}{\tau} = \|\bvec{Y}_{\text{obs}} - \bvec{Y}^\tau_s\|_2$ with $\bvec{Y}_s^\tau \sim \approxsimProct{\cdot}{\paramvec}{\tau}$ is the approximate stochastic simulation process using the tau-leaping method (Algorithm~\ref{alg:assa}).

Figure~\ref{fig:amftautunemichment} shows the effect of varying $\tau$ on the computational cost of generating the $N$ weighted samples, and the optimal continuation probabilities $\eta_1$ and $\eta_2$, as determined by the adaptive update scheme. Note that some choices of $\tau$ result in very small values for both $\eta_1$ and $\eta_2$ (Figure~\ref{fig:amftautunemichment}(B)--(C)). However, this does not translate into a computational improvement, since the total compute time is always larger thatn the worst case with $\eta_1 = \eta_2 = 1$ (Figure~\ref{fig:amftautunemichment}(A)). We therefore conclude that there is no computational advantage in using MF-ABC for the Michaelis--Menten model as specified here. Effectively, the tau-leaping method does not provide enough of a computational benefit over the Gillespie direct method for the Michaelis--Menten model with the given initial conditions. 

\begin{figure}[h]
	\centering
	\includegraphics[width=0.8\linewidth]{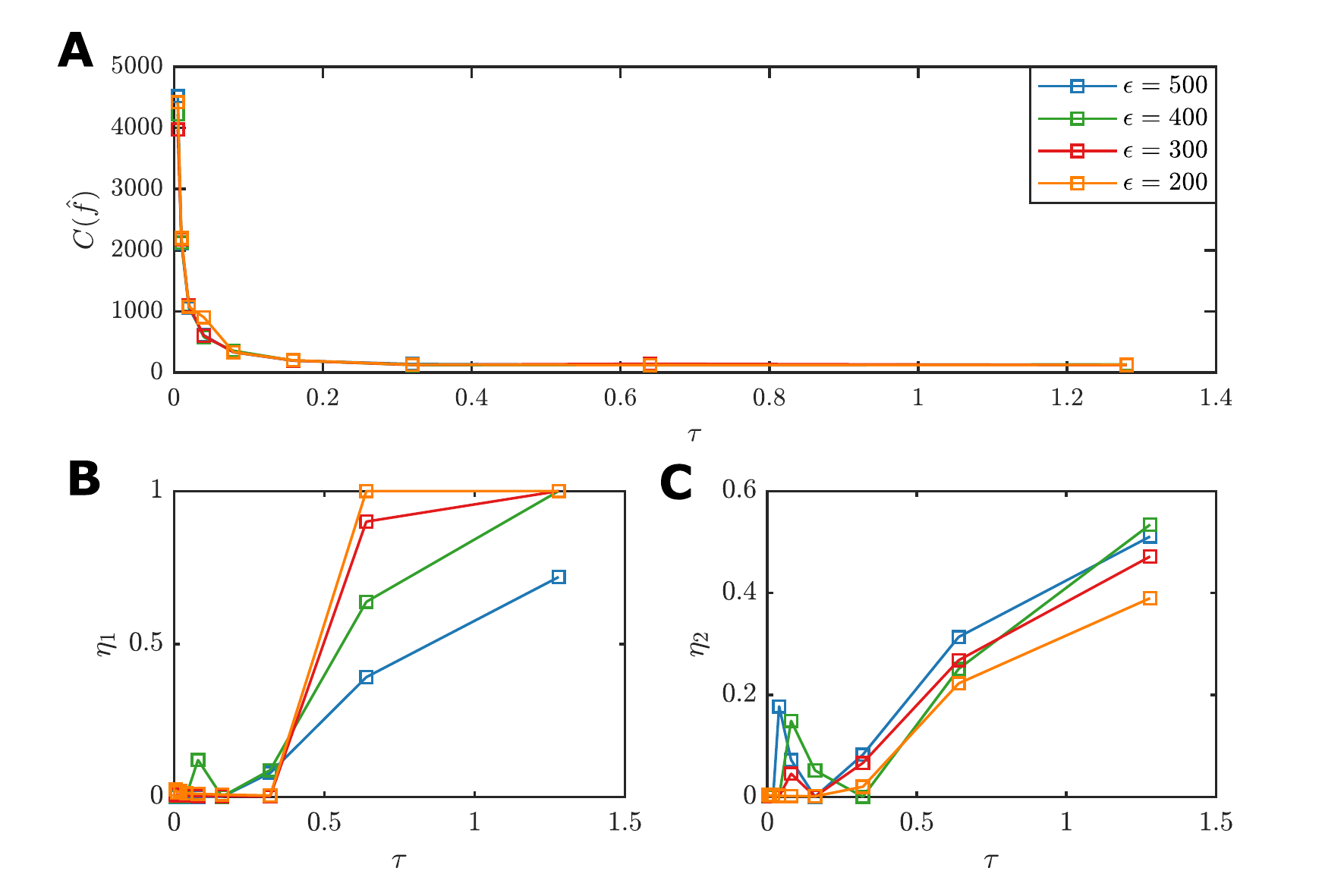}
	\caption{The relationship between the approximate stochastic simulation time-step, $\tau$ and: (A) the cost of generating $N = 10000$ weighted samples; (B) the continuation probability when an approximate stochastic simulation is accepted; and (C) the continuation probability when an approximate stochastic simulation is rejected. Coloured lines indicate the effect of $\tau$ for different discrepancy thresholds $\epsilon$.}
	\label{fig:amftautunemichment}
\end{figure}

A key message from this section is that before employing MLMC-ABC or MF-ABC, and by extension our new method MF-MLMC-ABC, some initial exploration should be performed. While MF-ABC provides no benefit for the Michaelis--Menten model, this is largely due to the simplicity of the network. In more complex models that we consider in the next two sections, MF-ABC provides a substantial improvement leading to high efficiency with the MF-MLMC-ABC method.

\subsection{Tuning and performance of MF-MLMC-ABC: Repressilator gene regulatory network}

We now demonstrate the computational benefits of MF-MLMC-ABC using a stochastic gene regulatory network called the repressilator~\cite{Elowitz2000}. The repressilator describes the expression levels of three genes, $G_1, G_2,$ and $G_3$, in which the expression of $G_i$ inhibits the expression of $G_{(i\mod 3) + 1}$, forming a cycle that results in stochastic oscillations. Each gene, $G_i$, consists of two reactions that describe gene expression, through the transcription of mRNA, $M_i$, and translation into protein, $P_i$, and two reactions that describe the degradation of mRNA and protein molecules. For the $i$th gene we have
\begin{equation}
G_i \xlongrightarrow{\alpha_0 + \alpha K^n/(K^n+P_j^n)} G_i +  M_i,\quad  M_i \overset{\beta}{\rightarrow} M_i + P_i, \quad  P_i \overset{\beta}{\rightarrow} \emptyset, \quad\text{and}\quad M_i \overset{\gamma}{\rightarrow} \emptyset, \label{eq:repr}
\end{equation}
where $j = (i+1 \mod 3) +1$, $\alpha_0 \geq 0$ is the leakage transcription rate (the transcription rate of a maximally inhibited gene), $\alpha + \alpha_0 > 0$ is the free transcription rate (uninhibited transcription rate), $n \geq 0$ is the Hill coefficient that describes the strength of the repressive effect of the inhibitor protein $P_j$, $K$ is the number of $P_j$ inhibitor proteins required to reduce the transcription rate of $G_i$ by 50\% (excluding leakage), $\beta > 0$ is the protein translation and degradation rate, and $\gamma > 0$ is the mRNA degradation rate.

We consider noisy observations of the protein copy numbers since these will be the only observables via fluorescent markers. This yields the observation process
\begin{equation}
\bvec{y}_{\text{obs}}(t)  \sim \mathcal{N}([P_{1,t},P_{2,t},P_{3,t}]^\text{T},\sigma^2\bvec{I}),
\label{eq:repobs}
\end{equation} 
where $\bvec{I}$ the $3 \times 3$ identity matrix. Just as with the Michaelis--Menten example, alternate noise models could be utilised. Discrete observations are taken at regular one time unit intervals, $t_i = i$, for $i = 0,1, \ldots, 10$. Figure~\ref{fig:repdata} shows an example realisation of the repressilator model along with discrete observations of the protein molecules.

\begin{figure}[h]
	\centering
	\includegraphics[width=\linewidth]{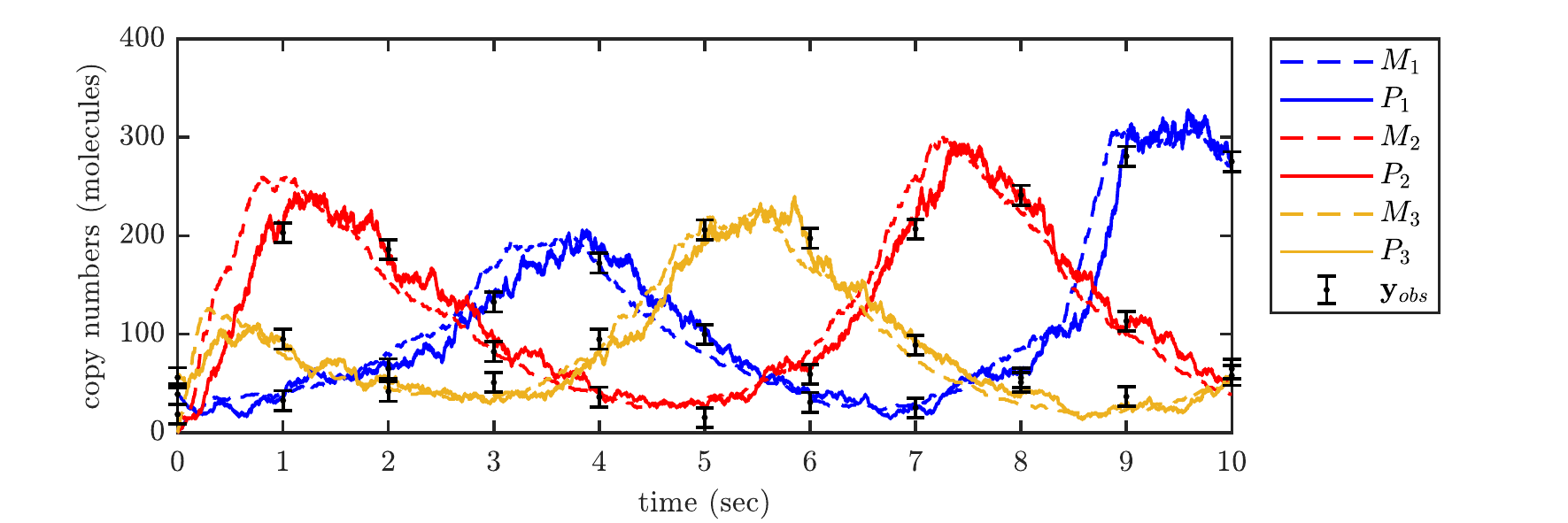}
	\caption{Example realisation of the repressilator gene regulatory network model along with noisy observations of the three protein molecules, $\bvec{y}_{\text{obs}}(t) \sim \mathcal{N}([P_{1,t},P_{2,t},P_{3,t}]^{\text{T}}, \sigma^2\bvec{I})$ (error bars indicate $\bvec{y}_{\text{obs}}(t) \pm \sigma$). Here, the initial condition is $M_{1,0} = M_{2,0} = M_{3,0} = 0$, $P_{1,0} = 40$, $P_{2,0} = 20$, $P_{3,0} = 60$, the true rate parameters are $\alpha_0 = 1$, $\alpha = 1000$, $K = 20$, $n = 2$, $\beta = 5$ and $\gamma = 1$, and observations are taken at $t_i = i$, for $i = 1, 2,\ldots 10$ with standard deviation $\sigma = 10$}
	\label{fig:repdata}
\end{figure}

This model is a common choice to benchmark the performance of likelihood-free inference methods since the oscillatory behaviour renders the acceptance probability to be very low~\cite{Warne2020,Prescott2020,Toni2009}. For our MF-MLMC-ABC method, the repressilator is particularly interesting due to the possibility that coupled pairs of exact simulations and approximate simulations will go out of phase with each other. This means the time-step sequence $\{\tau_\ell\}_{\ell=1}^{\ell = L}$ must be carefully chosen.

In the target ABC inference problem we consider assumes only the parameters of Hill functions to be unknown and evaluates
\begin{equation}
\CondE{K}{\bvec{Y}_{\text{obs}}} \approx \int_{\mathbb{R}^2} K\CondPDF{\paramvec}{\discrep{\bvec{Y}_{\text{obs}}}{\bvec{Y}_s} \leq \epsilon} \, \text{d}\paramvec,
\label{eq:repABC}
\end{equation}
where $\paramvec = [K,n]$ is the vector of unknown Hill function parameters, $\bvec{Y}_{\text{obs}} = [\bvec{y}_{\text{obs}}(t_0),\ldots, \bvec{y}_{\text{obs}}(t_{10})]$ are noisy observations of protein copy numbers at discrete times $t_1,\ldots, t_{10}$ (\eqref{eq:repobs}), $\bvec{Y}_s \sim \simProc{\cdot}{\paramvec}$ is simulated data of the repressilator model generated using the Gillespie direct method and the observation process (\eqref{eq:repobs}), $\epsilon$ is the discrepancy threshold, and the discrepancy metric is $\discrep{\bvec{Y}_{\text{obs}}}{\bvec{Y}_s} = \|\bvec{Y}_{\text{obs}} -\bvec{Y}_s\|_2$ where $\| \cdot,\|_2$ is the Euclidean norm. Independent uniform priors are used with $K \sim \mathcal{U}(10,30)$, $n \sim \mathcal{U}(1,4)$. We treat the rate parameters as known with $\alpha_0 = 1$, $\alpha = 1000$, $\beta = 5$ and $\gamma = 1$.

To apply MF-MLMC-ABC we require an appropriate sequence of time-steps $\{\tau_\ell\}_{\ell=1}^{\ell=L}$. To tune this sequence, we draw a small number, $N = 10000$, of MF-ABC weighted samples using the adaptive optimisation scheme for continuation probabilities (Supplementary Material) for a range of discrepancy thresholds $\epsilon \in [200,300,\ldots, 600]$ and time-steps $\tau \in [0.005,0.01,0.02,\ldots, 0.64]$.  We also take $\tilde{\epsilon} = \epsilon$ and $\approxdiscrept{\bvec{Y}_{\text{obs}}}{\bvec{Y}^\tau_s}{\tau} = \|\bvec{Y}_{\text{obs}} - \bvec{Y}^\tau_s\|_2$, with $\bvec{Y}_s^\tau \sim \approxsimProct{\cdot}{\paramvec}{\tau}$ the approximate stochastic simulation process using the tau-leaping method (Algorithm~\ref{alg:assa}). Figure~\ref{fig:amftautune} show the relationship between $\tau$, the computational cost, and continuation probabilities for each of the target discrepancy thresholds. If these discrepancies represented the MF-MLMC-ABC discrepancy sequence $\epsilon_1 = 600, \ldots, \epsilon_L = 300$ with $L = 4$, then Figure~\ref{fig:amftautune}(A) can be used to identify the optimal sequence for each $\tau_\ell$ by finding the value of $\tau_\ell$ with the lowest expected cost $C(\hat{f})$. This suggests a sequence of $\tau_1 = 0.08, \tau_ 2 = 0.04, \tau_3 = 0.01, \tau_4 = 0.02$. If we restrict our choice to a single time-step value to apply to all levels, then $\tau_1 = \tau_2 = \cdots \tau_L = 0.02$ is the best overall as this value for $\tau_\ell$ results in the lowest total cost for $N$ samples from each level.

\begin{figure}
	\centering
	\includegraphics[width=0.8\linewidth]{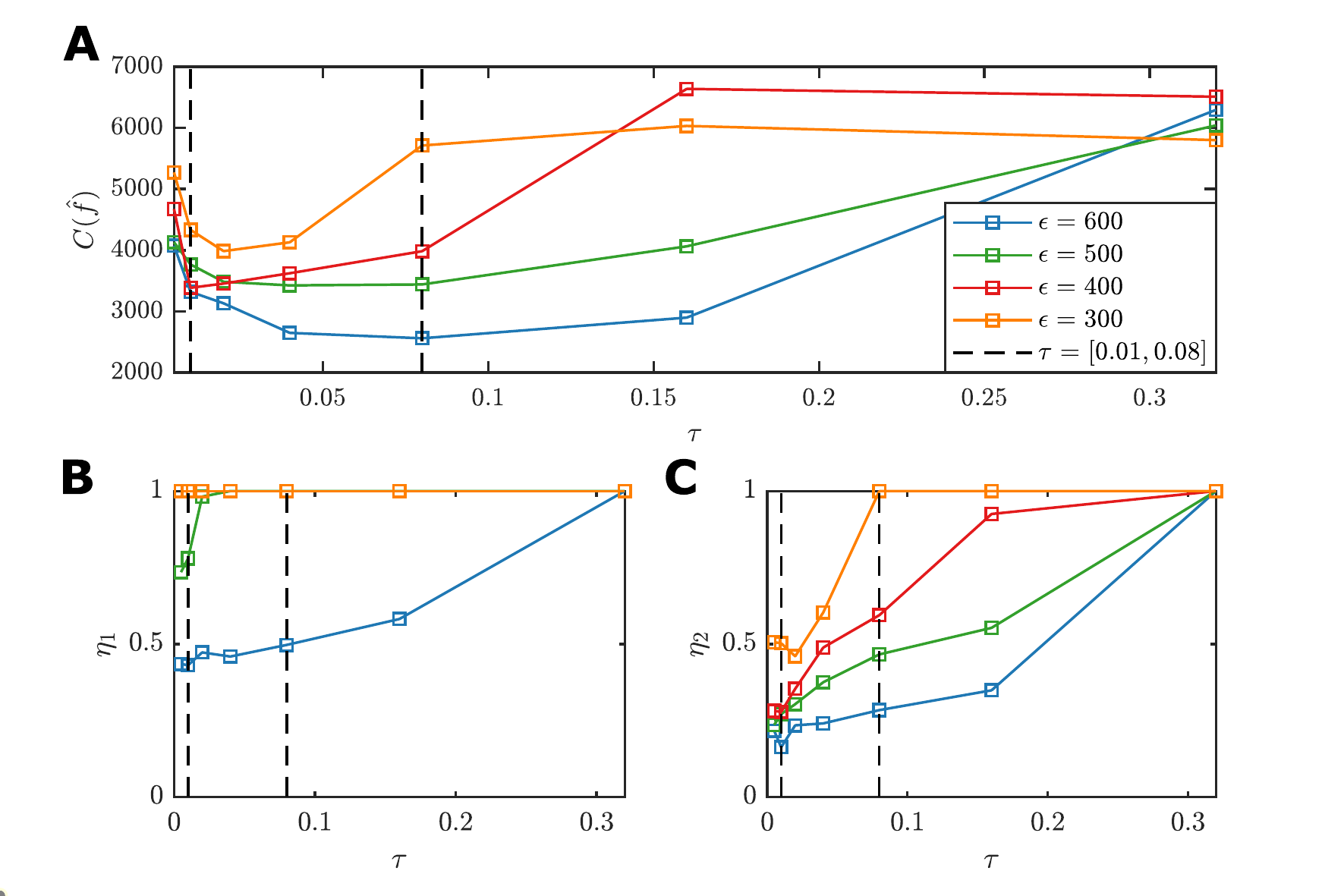}
	\caption{The relationship between the approximate stochastic simulation time-step, $\tau$, and: (A) the cost of generating $N = 10000$ weighted samples; (B) the continuation probability when an approximate stochastic simulation is accepted; and (C) the continuation probability when an approximate stochastic simulation is rejected. Solid lines indicate the effect of $\tau$ for different discrepancy thresholds $\epsilon$ and the dashed black lines indicate the range of values for $\tau$ that demonstrate performance improvement.}
	\label{fig:amftautune}
\end{figure}

Using the above heuristics we arrive at the choice of $L = 5 $ and $\tau = 0.04$. We apply our new MF-MLMC-ABC method (Algorithm~\ref{alg:mfmlmcabc}) to the ABC inference problem in \eqref{eq:repABC} and compare with MLMC-ABC (Algorithm~\ref{alg:mlmcabc}), MF-ABC (Algorithm~\ref{alg:mfabc}), and ABC rejection sampling (Algorithm~\ref{alg:ABC-rej})  for different values for the target discrepancy $\epsilon \in [350, 500]$, and in all cases $\epsilon_1 = 1600$. For each target, we perform optimal tuning steps for MF-MLMC-ABC, MLMC-ABC and MF-ABC to adapt the continuation probabilities $\{(\eta_{\ell,1},\eta_{\ell,1})\}_{\ell=1}^{\ell=L}$ and sample sizes $\{N_\ell\}_{\ell=1}^{\ell=L}$. We repeat this for different target variances, $h^2$, then estimate the variance, $\V{\hat{f}}$, and the computational cost, $C(\hat{f})$, to obtain an estimate of the convergence rate, $\gamma$, by least-squares fitting $\V{\hat{f}} \propto C(\hat{f})^{-\gamma}$ to align with theory from the MLMC literature~\cite{Giles2015,Giles2008}. Figure~\ref{fig:benchrep} demonstrates the substantial computational advantage of MF-MLMC-ABC. 

\begin{figure}[h]
	\centering
	\includegraphics[width=0.85\linewidth]{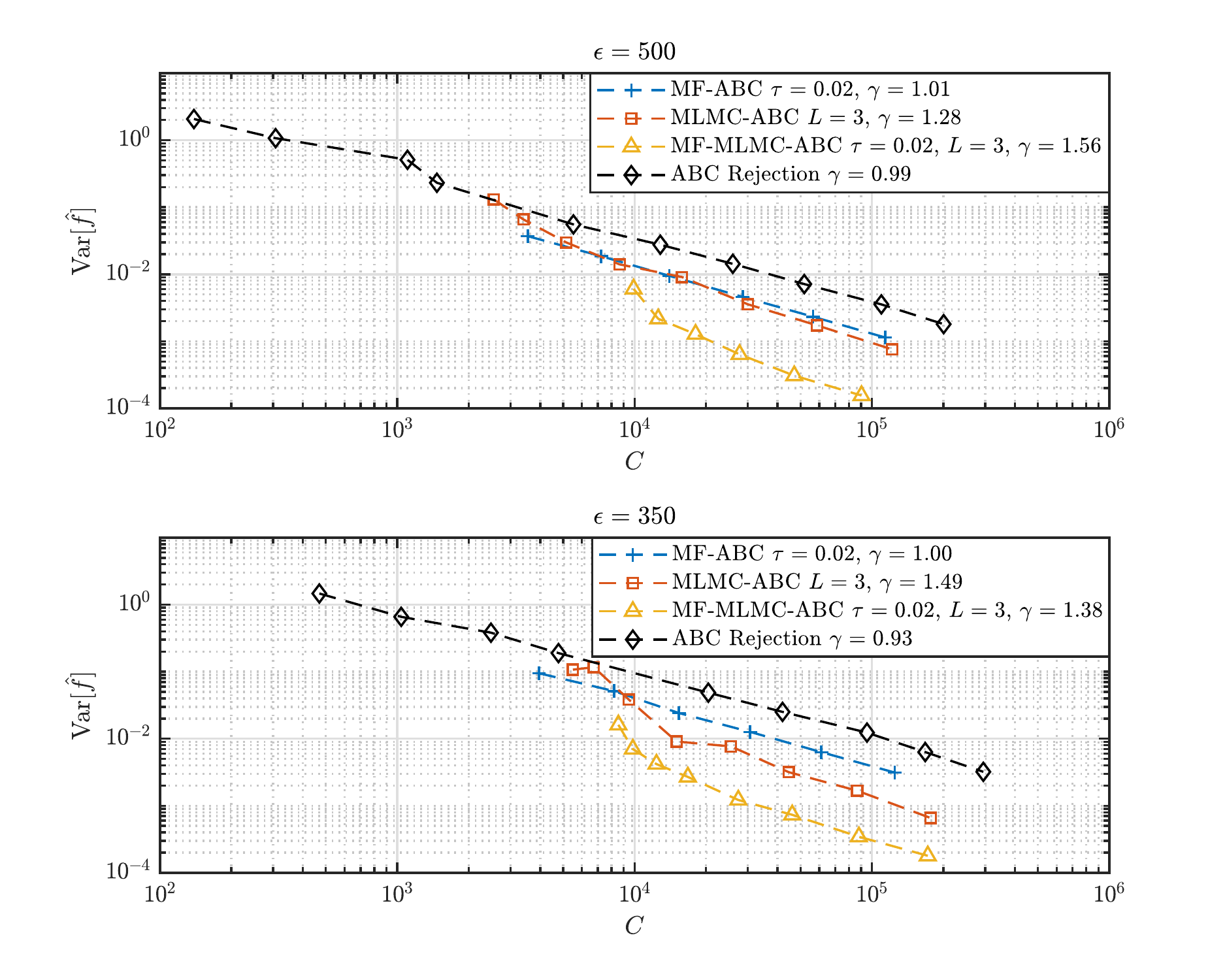}
	\caption{Comparison of convergence rates for MF-MLMC-ABC (yellow triangles) with MF-ABC (blue crosses), MLMC-ABC (red squares), and ABC rejection sampling (black diamonds) using the repressilator model with thresholds: (A) $\epsilon = 500$; and (B) $\epsilon = 350$. Rates are estimated by fitting $\V{\hat{f}} \propto C(\hat{f})^{-\gamma}$ to benchmark data using least squares.}
	\label{fig:benchrep}
\end{figure}

While MF-ABC is consistently computationally cheaper than ABC rejection sampling, the convergence rate is $\gamma \approx 1$ which is equivalent to the limiting behaviour of ABC rejection and the theoretical rate under the central limit theorem. MLMC-ABC achieves a higher convergence rate of around $\gamma \approx 1.5$. However, the computational benefit is not realised until the target variances are small due to the overhead of tuning for $\{N_\ell\}_{\ell=1}^{\ell=L}$ along with the exact stochastic simulations performed at every level. MF-MLMC-ABC out-performs the other methods in every target discrepancy and target variance, and succeeds in both improving the convergence rate and substantially reducing the overall computational cost, including the tuning steps. The convergence rate is improved in a similar way to MLMC-ABC with $\gamma \approx 1.5$, however, the overall computational reduction in MF-MLMC-ABC compared with MLMC-ABC is larger than the reduction in MF-ABC vs ABC rejection. This is largely due to the fact that MF-ABC is most effective for larger discrepancies, as noted by the tendency for optimal continuation probabilities to be smaller as the discrepancy threshold increases (Figure~\ref{fig:amftautune}(B)--(C)). Furthermore, MLMC will utilise variance reduction via the coupling in the telescoping summation to allocate fewer samples for the smaller discrepancies. As a result, the earlier levels benefit from both the higher acceptance rates that come with large discrepancy thresholds, and smaller continuation probabilities so exact stochastic simulations are rarely executed. The effect extends to the optimisation of $\{N_\ell\}_{\ell=1}^{\ell=L}$ through \eqref{eq:optNphi}. This results in a substantial reduction in the usual overheads associated with MLMC-ABC and computational benefit is realised for much smaller target variances. Consistently, for equivalent computational cost the MF-MLMC-ABC is between one and two orders of magnitude lower in terms of variance and the improvement increases for larger computation times due to the convergence rate. This example demonstrates practically how  MF-MLMC-ABC may be tuned without substantial overhead to provide a very high-performance inference method. We also explore the effect of different choices of $L$ and $\tau$ to demonstrate the efficiency of our heuristics (Supplementary Material).      

One final aspect of the MF-MLMC-ABC approach that is important to consider is that of the additional bias that can be incurred from either selecting a poor choice of $L$ (that is, $m$ is too large), or from the bias in the MF-ABC scheme for small $N_\ell$. To explore this we compare the expectations of the parameter $K$ across algorithms. Figure~\ref{fig:Repmfmlmcbias}, demonstrates that this bias, as expected, decays with computational effort. However, it is also important to note that the magnitude of the bias is small compared with that of the ABC approximation as additional experimentation shows that a smaller discrepancy of $\epsilon < 200$ leads to $\CondE{K}{\bvec{Y}_{\text{obs}}} \approx 18$. 

\begin{figure}[h]
	\centering
	\includegraphics[width=0.85\linewidth]{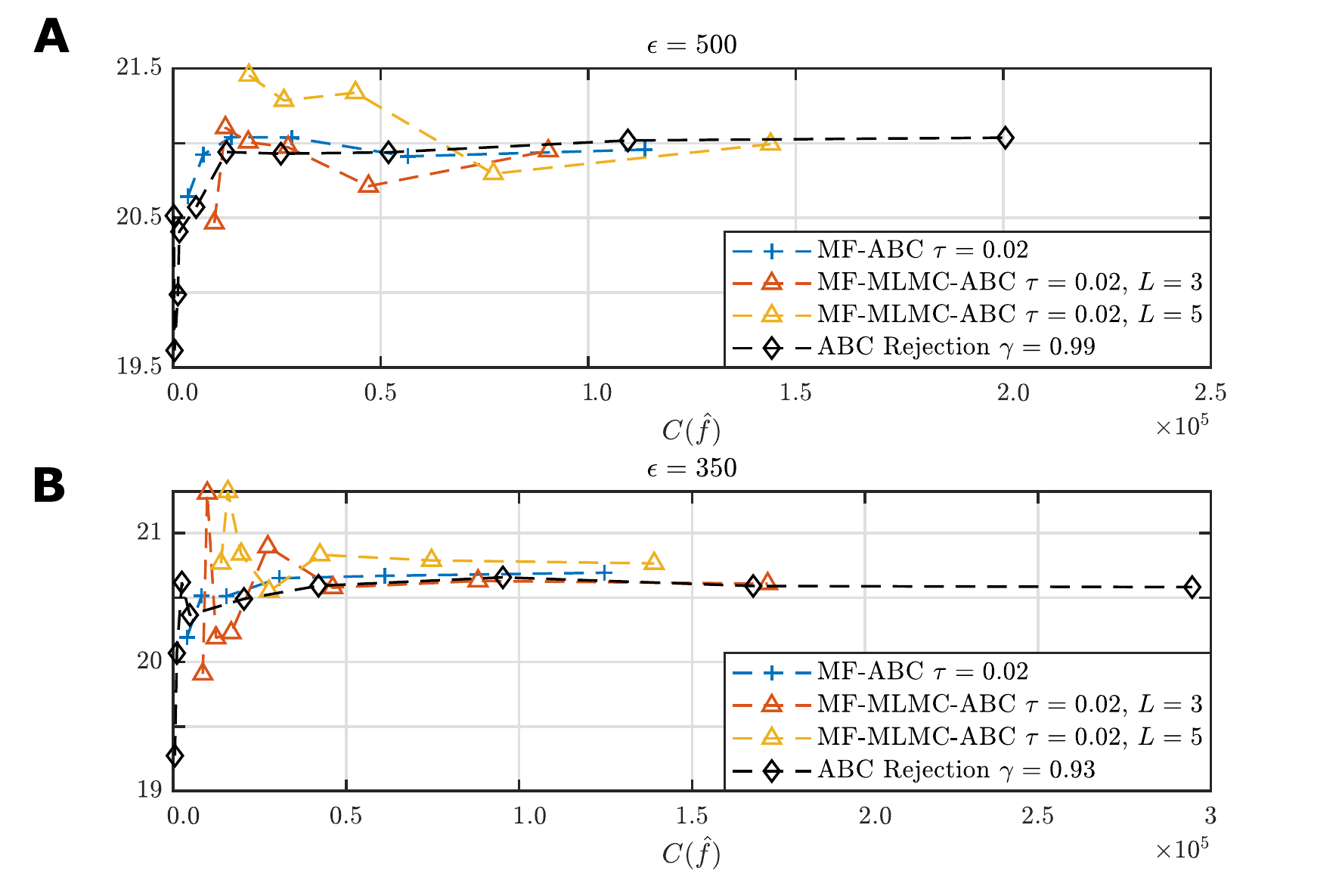}
	\caption{Comparison of expectations for MF-MLMC-ABC (yellow and red triangles) with MF-ABC (blue crosses), and ABC rejection sampling (black diamonds) using the repressilator model with thresholds: (A) $\epsilon = 500$; and (B) $\epsilon = 350$.}
	\label{fig:Repmfmlmcbias}
\end{figure}

\FloatBarrier
\subsection{A challenging problem: Two-step MAPK cascade reaction}

The last example we consider is a challenging biological network involving a two-step Mitogen Activated Protein Kinase (MAPK) enzymatic cascade~\cite{Dhananjaneyulu2012}. Such cascade reactions are essential components of cell signalling processes, such as Epidermal Growth Factor Receptor (EGFR) signalling, that regulates cell growth, death, proliferation, and differentiation in mammalian cells~\cite{Brown2004,Oda2005}. This two-step MAPK cascade model involves four coupled Michaelis--Menten components that govern the phosphorylation and dephosphorylation of two proteins $X$ and $Y$,   
\begin{equation}
\begin{split}
X + E \overset{k_1}{\rightarrow} [XE], \quad [XE] \overset{k_2}{\rightarrow} X + E, \quad [XE] \overset{k_3}{\rightarrow} X^*+ E, \\
X^* + P_1 \overset{k_4}{\rightarrow} [X^*P_1], \quad [X^*P_1] \overset{k_5}{\rightarrow} X^* + P_1, \quad [X^*P_1] \overset{k_6}{\rightarrow} X + P_1, \\
X^* + Y \overset{k_7}{\rightarrow} [X^*Y], \quad [X^*Y]\overset{k_8}{\rightarrow} X^* + Y, \quad [X^*Y] \overset{k_9}{\rightarrow} X^* + Y^*, \\
Y^* + P_2 \overset{k_{10}}{\rightarrow} [Y^*P_2], \quad [Y^*P_2] \overset{k_{11}}{\rightarrow} Y^* + P_2, \quad [Y^*P_2] \overset{k_{12}}{\rightarrow} Y + P_2, \\
\end{split}
\label{eq:step2MAPK}
\end{equation}
where $k_1, k_2, \ldots, k_{12}$ are kinetic rate parameters, $X^*,Y^*$ are the activated (phosphorylated) proteins, $E$ is the enzyme involved in the activation of the $X$ protein, and $P_1,P_2$ are phosphatase molecules that dephosphorylate $X^*, Y^*$. Finally note the two-step process where the activated $X^*$ protein acts as an enzyme in the activation of $Y$. 

In this case, we assume only activated proteins can be detected, therefore we consider the observation process
\begin{equation}
\bvec{y}_{\text{obs}}(t)  \sim \mathcal{N}([X_t^*,Y_t^*]^\text{T},\sigma^2\bvec{I}),
\label{eq:mapkobs}
\end{equation} 
where $\bvec{I}$ the $2 \times 2$ identity matrix. Discrete observations are taken at regular four time unit intervals. $t_i = 4i$, for $i = 0,1, \ldots 50$. Figure~\ref{fig:mapkdata}(A) shows an example realisation of the two-step MAPK cascade model along with discrete observations of the activated proteins.  Figures~\ref{fig:mapkdata}(B)--(D) provide additional detail and highlight the complex dynamics of the unobserved chemical species.  Given the very limited data in this realistic scenario, we do not have practical identifiability for all rate parameters. Based on the network structure, we only expect parameters $k_3$, $k_6$, $k_9$, and $k_{12}$ to be identifiable as these rates correspond to reactions that change the copy numbers of the observed variables.
\begin{figure}[h]
	\centering
	\includegraphics[width=0.9\linewidth]{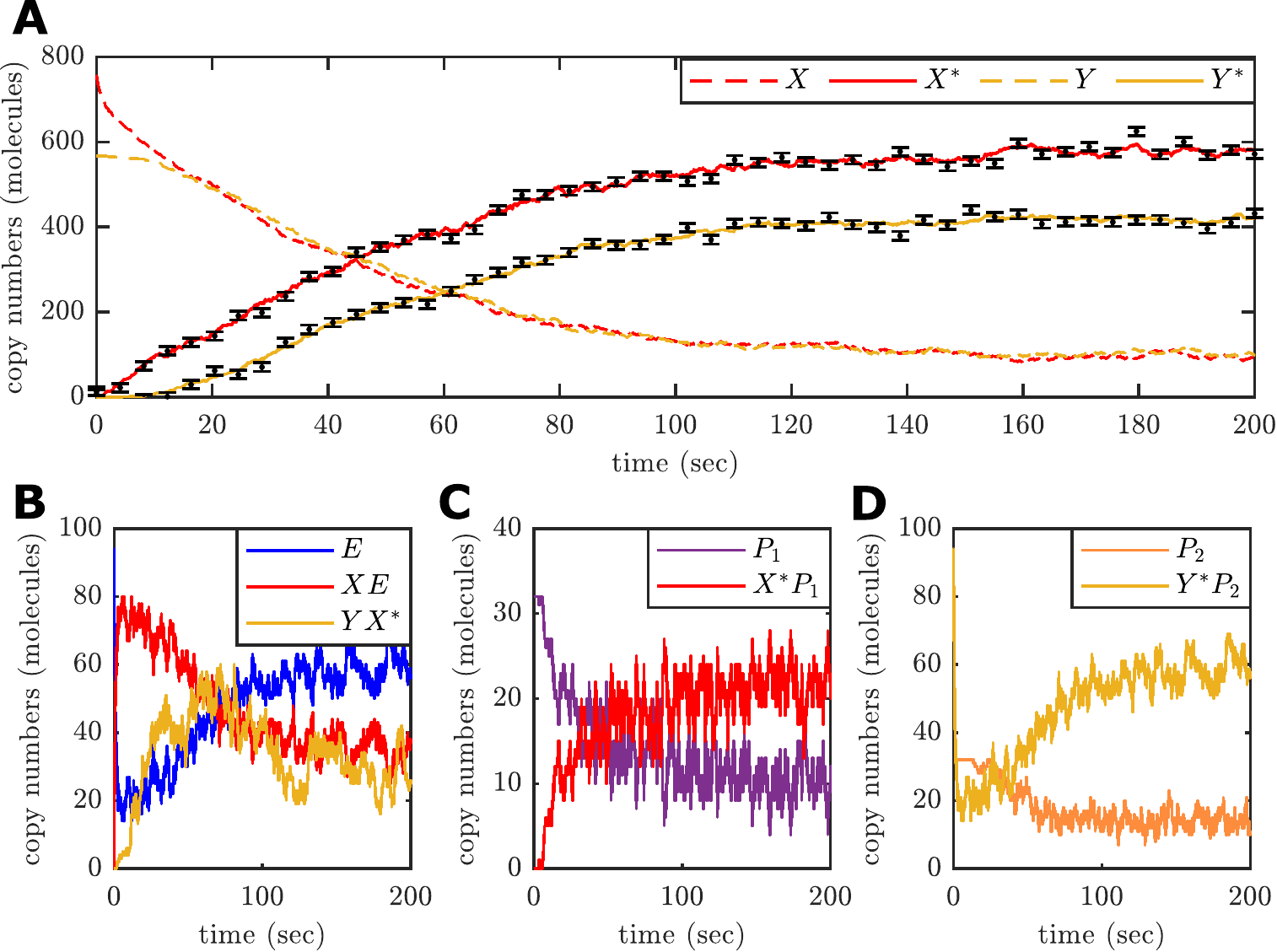}
	\caption{Example realisation of the two-step MAPK cascade reaction network model along with noisy observations of the two phosphorylated protein molecules, $\bvec{y}_{\text{obs}}(t) \sim \mathcal{N}([X^*_t,Y^*_t]^{\text{T}}, \sigma^2\bvec{I})$ (error bars indicate $\bvec{y}_{\text{obs}}(t) \pm \sigma$). Here, the initial condition is $E_{0} = 94$, $X_{0} = 757$, $Y = 567$, $P_1 = P_2 = 32$, and $X^* = Y^* = EX = X^*P_1 = YX^* = Y^*P_2$, and the true rate parameters are $k_1 = 0.001$, $k_2 = k_1/120$, $k_3 = 0.18$, $k_4 = 0.001$, $k_5 = k_4/22$, $k_6 = 0.3$, $k_7 = 0.0001$, $k_8 = k_7/110$, $k_9 = 0.2$, $k_{10} = 0.001$, $k_{11} = k_{10}/22$, and $k_{12} = 0.3$, and observations are taken at $t_i = 4i$, for $i = 1, 2,\ldots 50$ with standard deviation $\sigma = 10$.}
	\label{fig:mapkdata}
\end{figure}

The target ABC inference problem we consider is to estimate the reverse rate parameter of the dephosphorylation reaction for the deactivation of $Y^*$ and the marginal posterior distributions for the identifiable rate parameters
\begin{equation}
\begin{split}
\CondE{k_{11}}{\bvec{Y}_{\text{obs}}} &\approx \int_{\mathbb{R}^8} k_{11}\CondPDF{\paramvec}{\discrep{\bvec{Y}_{\text{obs}}}{\bvec{Y}_s} \leq \epsilon} \, \text{d}\paramvec, \\
\CondProb{k_i < s_i}{\bvec{Y}_{\text{obs}}} &\approx \int_{\mathbb{R}^8} \ind{(-\infty,s_i]}{k_i}\CondPDF{\paramvec}{\discrep{\bvec{Y}_{\text{obs}}}{\bvec{Y}_s} \leq \epsilon} \, \text{d}\paramvec, \quad i = 3,6,9,12
\end{split}
\label{eq:mapkABC}
\end{equation}
where $\paramvec = [k_2,k_3,k_5,k_6,k_8,k_9,k_{11},k_{12}]$ is the vector of unknown rate parameters, $\bvec{Y}_{\text{obs}} = [\bvec{y}_{\text{obs}}(t_0),\ldots, \bvec{y}_{\text{obs}}(t_{50})]$ are noisy observations of the activated protein copy numbers at discrete times (\eqref{eq:mapkobs}), $\bvec{Y}_s \sim \simProc{\cdot}{\paramvec}$ is simulated data of the two-step MAPK model using the Gillespie direct method and simulating the observation process (\eqref{eq:mapkobs}), $\epsilon$ is the discrepancy threshold, and the discrepancy metric is $\discrep{\bvec{Y}_{\text{obs}}}{\bvec{Y}_s} = \|\bvec{Y}_{\text{obs}} -\bvec{Y}_s\|_2$ where $\| \cdot\|_2$ is the Euclidean norm. Independent uniform priors are used with $k_2 \sim \mathcal{U}(0,k_1)$, $k_3 \sim \mathcal{U}(0,1)$, $k_5 \sim \mathcal{U}(0,k_4)$, $k_6 \sim \mathcal{U}(0,1)$, $k_8 \sim \mathcal{U}(0,k_7)$, $k_9 \sim \mathcal{U}(0,1)$, $k_{11} \sim \mathcal{U}(0,k_{10})$, $k_{12} \sim \mathcal{U}(0,1)$. We treat the rate parameters of the complex binding in all Michaelis--Menten reactions as known with $k_1 = k_4 = 0.001$, and $k_7 = k_{10}$.

We apply MF-MLMC-ABC to this problem with $L = 7$, $\tau = 0.5$, $\epsilon_1 = 1600$, $\epsilon_L = 300$, $\epsilon_\ell = \epsilon_{\ell-1}/m$ for all $\ell = 2, \ldots, L$ and $m \approx 1.32$, $\tilde{\epsilon}_\ell = \epsilon_\ell$ for $\ell = 1, \ldots,L$, $\approxdiscrept{\bvec{Y}_{\text{obs}}}{\bvec{Y}^\tau_s}{\tau} = \|\bvec{Y}_{\text{obs}} - \bvec{Y}^\tau_s\|_2$ with $\bvec{Y}_s^\tau \sim \approxsimProct{\cdot}{\paramvec}{\tau}$ the approximate stochastic simulation process using the tau-leaping method (Algorithm~\ref{alg:assa}), and $\tau_\ell = \tau$ for $\ell = 1,2, \ldots, L$. Just as with the repressilator, we observe significant improvements (Figure~\ref{fig:benchmapk}) against MLMC-ABC, MF-ABC, and ABC rejection sampling. We also show the estimated marginal posterior densities for the identifiable parameters, $k_3$, $k_6$ $k_9$, and $k_{12}$, in Figure~\ref{fig:mapkpost}.

\begin{figure}
	\centering
	\includegraphics[width=0.7\linewidth]{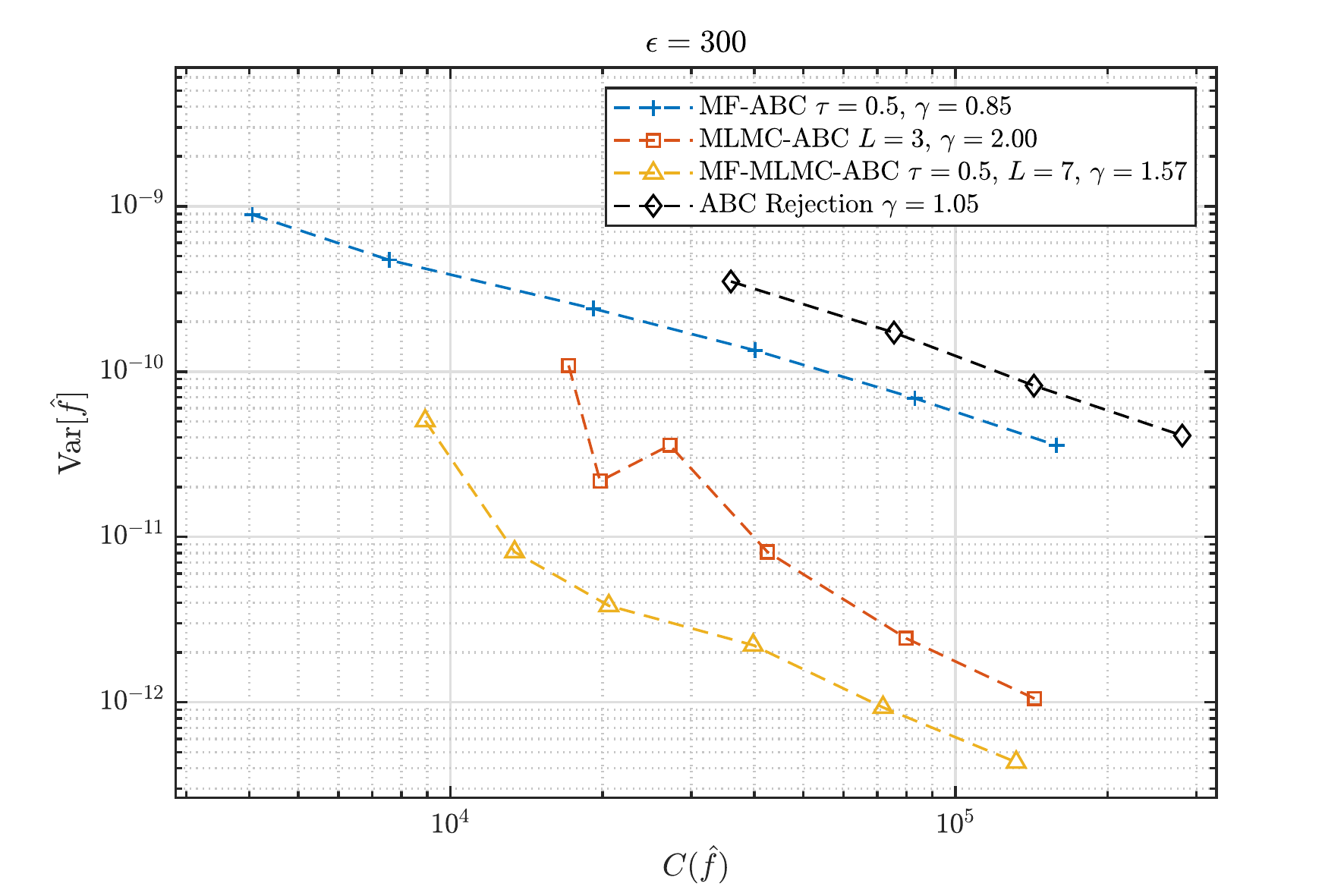}
	\caption{Comparison of convergence rates for MF-MLMC-ABC (yellow triangles) with MF-ABC (blue crosses), MLMC-ABC (red squares), and ABC rejection sampling (black diamonds) using the two-step MAPK model with threshold $\epsilon = 300$. Rates are estimated by fitting $\V{\hat{f}} \propto C(\hat{f})^{-\gamma}$ to benchmark data using least squares.}
	\label{fig:benchmapk}
\end{figure}

\begin{figure}
	\centering
	\includegraphics[width=0.8\linewidth]{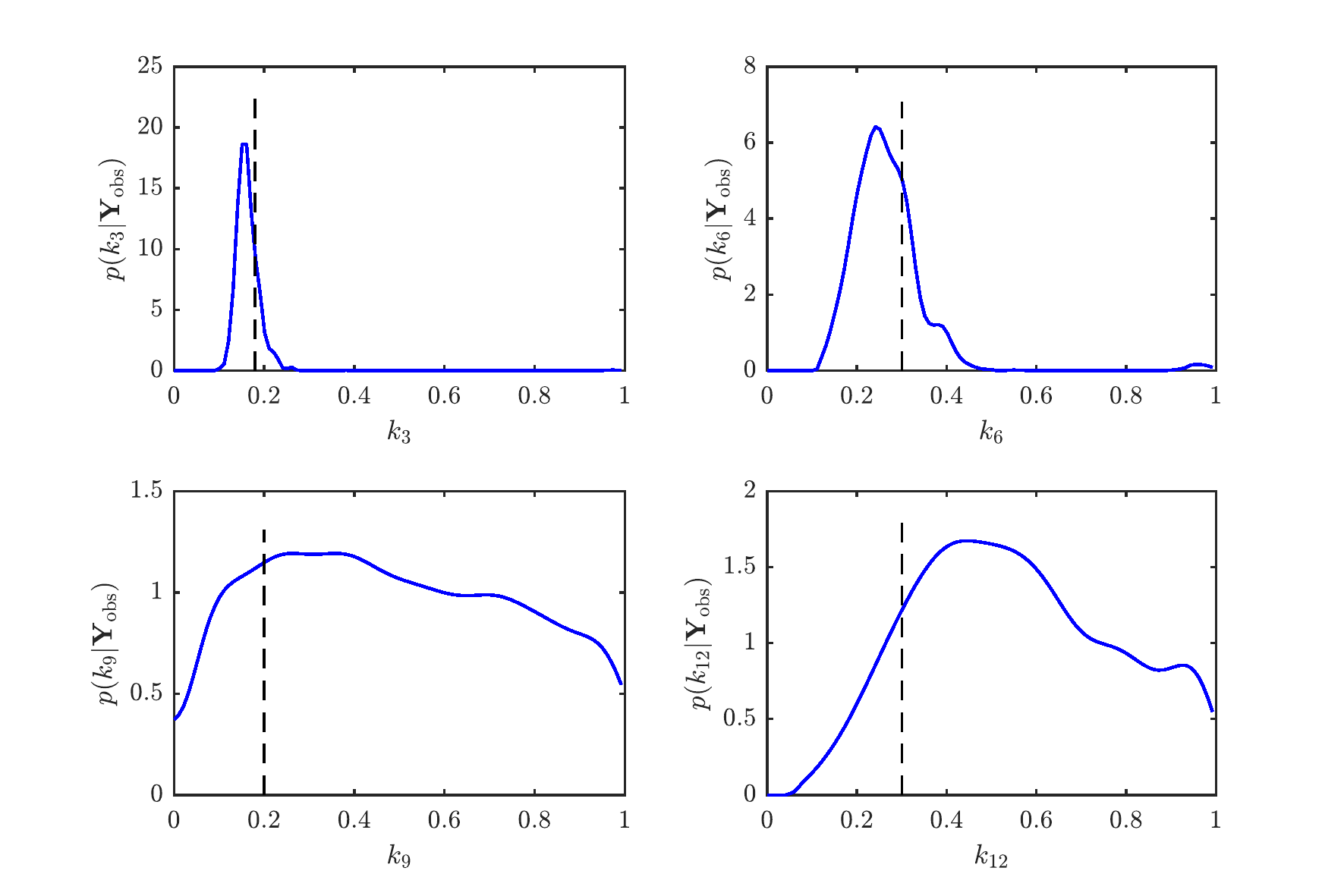}
	\caption{Marginal posterior probability density functions (blue lines) estimated by the for identifiable parameters from the two-stage MAPK cascade model. True parameter values are indicated (black dashed lines). Densities were estimated using the MF-MLMC-ABC method with $L = 7$, $\tau = 0.5$ and $\epsilon_1 = 1600$ and $\epsilon_7 = 300$.}
	\label{fig:mapkpost}
\end{figure}

In this case, MF-MLMC-ABC exceeds two orders of magnitude improvement over both MF-ABC and ABC rejection sampling. It is also clear that the MLMC-ABC approach is still very effective, with almost two orders of magnitude improvement over ABC rejection sampling. However, just as with the repressilator model, the MF-MLMC-ABC approach achieves the improved convergence rate of MLMC and reduced computation time for each term in the telescoping summation, notwithstanding a cost reduction in the MLMC tuning overhead.

\FloatBarrier


\section{Discussion}

In this work, we have introduced a new approach to ABC-based parameter inference for partially observed stochastic processes. Our approach combines the benefits of MLMC variance reduction techniques~\cite{Giles2008,Jasra2019,Warne2018} with the benefits of a multifidelity method for reducing stochastic simulation costs~\cite{Prescott2020,Prescott2021,Rhee2015}. We have developed the formulation by applying the multifidelity weighting scheme directly to the MLMC telescoping summation for ABC inference. Our practical implementation of the new algorithm demonstrates how various components of the algorithm are tuned and optimised, leading to the acceleration of ABC inference by two orders of magnitude  for realistic inference problems in systems biology.

These promising results open new avenues for accelerating various ABC-based schemes. While we focus on the acceleration of ABC rejection sampling, this is chosen as a first step toward other schemes. For example, Jasra et al.~\cite{Jasra2019} demonstrate an MLMC approach to SMC for ABC, and Prescott and Baker~\cite{Prescott2021} develop an SMC-based implementation of multifidelity ABC.  Both of these methods improve upon SMC for ABC and future work combining these methods could produce a multifidelity MLMC version of SMC. Further, ABC is not the only likelihood-free method that relies heavily on many stochastic simulations from the model, pseudo-marginal methods~\cite{Andrieu2009,Andrieu2010,Warne2020} and BSL~\cite{Price2017,Priddle2021} also use simulations to either estimate the likelihood function or construct a Gaussian approximation to the likelihood. There may be opportunities to combine MLMC and multifidelity methods for both these approaches to likelihood-free inference. For example, Jasra et al.~\cite{Jasra2018} consider a MLMC approach to particle MCMC and it may be possible to obtain additional benefit from multifidelity methods. 

We have also focussed our attention on ABC inference for partially observed discrete-state Markov processes since ABC methods are widely applied for such models in systems biology~\cite{Sunnaker2013,Toni2009,Liepe2014,Wu2014}, epidemiology \cite{Chinazzi2020,McKinley2018,Minter2019,Walker2019,Warne2020a}, ecology~\cite{Beaumont2010,Siren2018}, and physics~\cite{Akeret2015,BarajasSolano2019,Christopher2018}. Consequently, the high-fidelity and low-fidelity simulations within the MF-ABC framework are, respectively, assumed to be the natural choices of Gillespie's direct method and the tau-leaping method parameterised by $\tau$. However, our methodology supports many other variations, including differences in discrepancy metrics, summary statistics, and discrepancy thresholds, that could provide improvements in the ability for the low-fidelity simulation to predict the outcome of the high-fidelity simulation and thus reduce the optimal continuation probabilities. Furthermore, the approach is generally applicable to any ABC inference problem in which an appropriate approximate simulation scheme can be identified. For example, there is a wide range of potential choices for approximations that will lead to efficient multifidelity sampling~\cite{Peherstorfer_2018}, such as model reduction~\cite{Brown2004,Transtrum2014}, mean-field or linear approximations~\cite{Warne2019b,Schnoerr2017,Browning2020,Cao2018}, or surrogate models that act as emulators~\cite{Buzbas2015,Tripathy2018,Borowska2021}.  Finally, recent advances in deep learning provide potential for the automatic construction of surrogates~\cite{Lueckmann2017,Papamakarios2019,Cranmer2020}.

An additional novel modification we develop here, with further generalisations developed in Prescott et~al.~\cite{Prescott2021b}, is an adaptive scheme for iteratively updating the continuation probabilities (Supplementary Material). This approach is more robust than previous implementations~\cite{Prescott2020}, however, it relies upon a fixed choice of the time-step, $\tau$, in the approximation. While this scheme will identify poor choices in $\tau$ with $\eta_1,\eta_2 \to 1$, future work should investigate adaptive schemes to optimise $\tau$ along with $\eta_1,\eta_2$; this would lead to a near automatic tuning step for MF-MLMC-ABC. Furthermore, a similar adaptive optimisation approach could be used to adapt the optimal MLMC sample size sequence to reduce the MLMC tuning overhead. However, a more pressing issue with MLMC-based approaches to ABC is the assumption of a fixed sequence of discrepancy thresholds~\cite{Prescott2021,Jasra2019,Warne2018}, while the state-of-the-art in SMC-ABC is to adaptively select this sequence~\cite{Drovandi2011,Beaumont2009}. Such adaptive schemes may be possible for MLMC-based ABC especially if applied in an SMC setting. Finally, other variance reduction techniques, such as array randomised quasi-Monte Carlo methods for Markov chains, could be applied to accelerate estimation of the multifidelity expectations even further~\cite{LEcuyer2009,Puchhammer2021,Beentjes2018}.

Our implementations have been developed efficiently using a high-level programming environment using a single CPU core. While this is sufficient to demonstrate the computational benefits of our new algorithm, there are many additional optimisation techniques that could also be applied here. In particular, one substantial advantage of utilising MF-MLMC-ABC based upon rejection sampling, is that there are few synchronisation steps. Therefore, most of the stochastic simulations (especially the approximate simulations) can exploit many parallel computing architectures, such as general-purpose graphics processing units (GPGPUs)~\cite{Lee2010,Klingbeil2011,Hurn2016}, single instruction multiple data (SIMD) CPU processors~\cite{Warne2021,Mahani2015}, and recent advances in AI hardware~\cite{Kulkarni2020}. This leads to methods in which the statistical efficiency also directly scales to state-of-the-art massively parallel computing.

Finally, our numerical results demonstrate that high performance and high precision Bayesian inference can be performed for challenging partially observed stochastic processes such as those that arise in systems biology. We demonstrate this high performance for a large network with 12 parameters and 11 chemical species without any dimensionality reduction using summary statistics that is almost always required for ABC inference to be viable. With computational improvements of up to two orders of magnitude, our method is a significant advance in the use of approximations to accelerate ABC inference without incurring accuracy penalties. Furthermore, the success of the multifidelity and MLMC approach to inference will enable more realistic and complex models to be used to analyse modern, high resolution data.

\paragraph*{Acknowledgments} DJW thanks the Australian Mathematical Society for the Lift-off Fellowship. DJW and MJS acknowledge support from the Centre for Data Science at QUT, and the ARC Centre of excellence in Mathematical and Statistical Frontiers (ACEMS; CE140100049).   REB and TPP would like to thank BBSRC/UKRI for funding via grant number BB/R00816/1. REB is supported by a Royal Society Wolfson Research Merit Award. TPP is supported by Wave 1 of The UKRI Strategic Priorities Fund under the EPSRC Grant EP/W006022/1, particularly the ``Shocks and Resilience'' theme within that grant, and The Alan Turing Institute.  MJS is supported by the Australian Research Council (DP200100177). Computational resources were provided by the eResearch at QUT.

\paragraph*{Software Availability} Matlab source code with example implementations and demonstrations for all numerical examples presented in this work is available on  \href{https://github.com/davidwarne/MLMCandMultifidelityForABC}{GitHub}).

\begin{appendices}
	\appendix

	\section{Analysis of multifidelity rejection sampling}
	
	In this section, we present theoretical results for the bias and variance of the multifidelity ABC rejection sampler. We begin with a general importance sampler, then show that the multifidelity approach is a special case of importance sampling.
	\subsection{Importance sampling for approximate Bayesian computation}
	Let $\like{\dat}{\paramvec}$ be the true but intractable likelihood function and assume a likelihood-free approximation $\approxlike{\dat}{\paramvec}$. For approximate Bayesian Computation (ABC) methods we have the approximation to be chosen such that 
	\begin{equation*}
	\like{\dat}{\paramvec} \approx \alpha\approxlike{\dat}{\paramvec}= \CondE{\omega_\epsilon]}{\paramvec} = \int_{\mathbb{D}}\omega_\epsilon(\paramvec,\dat, \simdat)\simProc{\simdat}{\paramvec}\,\text{d}\simdat ,
	\end{equation*}
	where $\alpha > 0$ is a proportionality constant, $\omega_\epsilon(\paramvec,\dat, \simdat) = \ind{(-\infty, \epsilon]}{\discrep{\dat}{\simdat}}$ is the ABC accept/reject weighting function and $\simProc{\cdot}{\paramvec}$ is the model simulation process.
	
	Now consider an importance distribution $q(\paramvec) \geq 0$ with positive support that includes the positive support of the prior $\PDF{\paramvec}$. Suppose $N$ weighted samples generated by sampling $\paramvec^{1},\paramvec^{2}, \ldots,\paramvec^{N} \sim q(\paramvec)$, then evaluating the importance weights,
	\begin{equation*}
	w_i = \omega_i\frac{\PDF{\paramvec^i}}{q(\paramvec^i)},
	\end{equation*}
	where $\omega_i = \omega_\epsilon(\paramvec^i,\dat, \simdat^i)$ with $\simdat^i \sim \simProc{\cdot}{\paramvec^i}$. Since $\CondE{\omega_\epsilon}{\paramvec} = \alpha \approxlike{\dat}{\paramvec}$, we have,
	\begin{equation*}
	\CondE{w}{\paramvec} = \frac{\PDF{\paramvec}}{q(\paramvec)}\CondE{\omega_\epsilon}{\paramvec} = \alpha \frac{\PDF{\paramvec}}{q(\paramvec)}\approxlike{\dat}{\paramvec}.
	\end{equation*}
	This induces an approximate posterior distribution,
	\begin{equation*}
	\CondPDFsub{\paramvec}{\dat}{\epsilon} = \frac{1}{Z}\approxlike{\dat}{\paramvec}\PDF{\paramvec},
	\end{equation*}
	where $\E{w} = \E{\CondE{w}{\paramvec}} = \alpha Z$, with the outer expectation is taken with respect to the importance distribution $q(\cdot)$. Without loss of generality, take $\alpha = 1$.
	
	The weighted sample is then used to construct the weighted Monte Carlo estimate,
	\[
	\hat{f} = \frac{\sum_{i=1}^N w_i f(\paramvec_i)}{\sum_{j=1}^N w_j},
	\]
	of the posterior mean, 
	\begin{equation*}
	\bar{f} = \CondE{f(\paramvec)}{\dat} = \int_{\paramspace} f(\paramvec)\CondPDFsub{\paramvec}{\dat}{\epsilon}\, \mathrm{d}\paramvec,
	\end{equation*}
	for the arbitrary function $f: \paramspace \rightarrow \mathbb R$.
	The following results show that the estimate $\hat f$ is asymptotically unbiased and consistent estimator.
	Using the Delta method, we write the leading order behaviour of the bias and the variance of $\hat f$. These results apply independently of $\omega_\epsilon$.
	
	\begin{theorem}
		\label{thm:mse}
		We denote $\Delta(\paramvec) = f(\paramvec) - \bar f$ and $\hat \Delta = \hat f - \bar f$ as the recentred error. Then, to leading order, the bias of $\hat f$ and the MSE of $\hat f$ are given by
		\begin{align}
		\E{\hat \Delta} &= - \left[ \frac{\E{w^2 \Delta}}{\E{w}^2} \right] \frac{1}{N} + O(N^{-2}), \label{eq:bias} \\
		\E{\hat \Delta^2} &= \left[ \frac{\E{w^2 \Delta^2}}{\E{w}^2} \right] \frac{1}{N}  + O(N^{-2}), \label{eq:mse}
		\end{align}
		where expectations are taken over the importance distribution, $\paramvec \sim q(\cdot)$.
	\end{theorem}
	
	\begin{proof}
		Note that $\hat f = R / S$ for the random variables $R = \sum_{i=1}^N w_i f( \paramvec_i)$ and $S = \sum_{j=1}^N w_j$.
		We define the function $F_1(r,s) = r/s - \bar f$ and note that $\hat \Delta = F_1(R, S)$.
		The Delta method proceeds by taking the second-order Taylor expansion of $F_1(R,S)$ about $(\mu_R, \mu_S) = (\E{R}, \E{S})$.
		It is straightforward to show that
		\begin{align*}
		\mu_R &= N \E{wf} = N Z \bar f, \\
		\mu_S &= N \E{w} = N Z,
		\end{align*}
		and thus that the first possible non-zero terms in the expansion of $F_1$ are the second order terms,
		\begin{align*}
		\E{F_1(R,S)}
		&=  -\frac{1}{\mu_S^2}\E{(R-\mu_R)(S-\mu_S)} + \frac{1}{2} \E{((S-\mu_S)^2} \frac{2 \mu_R}{\mu_S^3} + O(N^{-2}) \\
		&= \frac{1}{N^2 Z^2} \left( \V{S} \bar f - \C{R}{S} \right) + O(N^{-2}) \\
		&= -\frac{1}{N^2 Z^2} \C{R - \bar f S}{S} + O(N^{-2}).
		\end{align*}
		
		Expanding $R$ and $S$ as finite sums, we have
		\begin{align*}
		\E{\hat \Delta} &= -\frac{1}{N^2 Z^2} \C{\sum_{i=1}^N w_i(f_i - \bar f)}{\sum_{j=1}^N w_j} + O(N^{-2}) \\
		&= -\frac{1}{N Z^2} \C{w \Delta}{w} + O(N^{-2}) \\
		&= -\frac{1}{N Z^2} \left( \E{w^2 \Delta} - \E{w} \E{w \Delta} \right) + O(N^{-2}).
		\end{align*}
		The first result follows from again observing that $\E{w} = Z$ and $\E{w \Delta} = \E{wf} - \E{w}\bar f = 0$. A very similar argument applies to expanding $F_2(R,S) = F_1(R,S)^2$ about $(\mu_R, \mu_S)$, in order to find $\E{\hat \Delta^2} = \E{F_2(R,S)}$ to leading order.
	\end{proof}
	
	The key observation about \Cref{eq:mse,eq:bias} is that they apply to any choice of weighting, so long as $\CondE{w}{\paramvec} = \PDF{\paramvec}\approxlike{\dat}{\paramvec} / q( \paramvec)$. Finally, note that if $q(\cdot) = \PDF{\cdot}$ then we obtain ABC rejection sampling.

	\subsection{Multifidelity ABC rejection sampling}
	We now have the weight $w_i =\PDF{\paramvec^i}\omega_i / q(\paramvec^i)$ defined by the multifidelity ABC likelihood estimation, given by
	\[
	\omega_i = \ind{(-\infty,\tilde{\epsilon}]}{\approxdiscrep{\approxsimdat^i}{\dat}} + \frac{\ind{(-\infty,\eta(\approxsimdat^i)]}{U}}{\eta(\approxsimdat^i)} \left(\ind{(-\infty,\epsilon]}{\discrep{\simdat^i}{\dat}} - \ind{(-\infty,\tilde{\epsilon}]}{\approxdiscrep{\approxsimdat^i}{\dat}} \right),
	\]
	where $U$ is a unit-uniform random variable and the function $\eta(\approxsimdat))$ is the piecewise-continuous function
	\[
	\eta(\approxsimdat) = \eta_1\ind{(-\infty,\tilde{\epsilon}]}{\approxdiscrep{\approxsimdat}{\dat}} + \eta_2\ind{(\tilde{\epsilon},\infty]}{\approxdiscrep{\approxsimdat}{\dat}},
	\]
	parametrised by $\eta_1, \eta_2 \in (0,1]$.
	We have $\CondE{\omega}{\paramvec} = \CondProb{\discrep{\simdat}{\dat} \leq \epsilon}{\paramvec} = \approxlike{\dat}{\paramvec}$, and so Monte Carlo estimators using this weighting asymptotically agree with the estimates produced using standard ABC.
	Previous results have shown that in trading off MSE against computational budget, we can achieve smaller MSE for the same budget~\cite{Prescott2020}.
	
	The leading-order bias coefficients in the bias and MSE are determined by
	\begin{align*}
	\E{w^2 \Delta} &=  \int_{\paramspace} \Delta( \paramvec) \left( \frac{\PDF{\paramvec}}{q( \paramvec)} \right)^2 \CondE{\omega^2}{\paramvec} q( \paramvec)\, \mathrm d \paramvec, \\
	\E{w^2 \Delta^2} &=  \int_{\paramspace} \Delta( \paramvec)^2 \left( \frac{\PDF{\paramvec}}{q( \paramvec)} \right)^2 \CondE{\omega^2}{\paramvec} q( \paramvec)\, \mathrm d \paramvec,
	\end{align*}
	both of which rely on the conditional expectation $\CondE{\omega^2}{\paramvec}$.
	Expanding $\omega^2$ for the multifidelity case, we write
	\begin{align*}
	\omega^2
	&=
	\ind{(-\infty,\tilde{\epsilon}]}{\approxdiscrep{\approxsimdat}{\dat}} -
	2 \frac{\ind{(-\infty,\eta(\approxsimdat)]}{U}}{\eta(\approxsimdat)} \ind{(\epsilon,\infty]}{\discrep{\simdat}{\dat}} \ind{(-\infty,\tilde{\epsilon}]}{\approxdiscrep{\approxsimdat}{\dat}} \\
	&+ \frac{\ind{(-\infty,\eta(\approxsimdat)]}{U}}{\eta(\approxsimdat)^2}\left[ \ind{(-\infty,\epsilon]}{\discrep{\simdat}{\dat}} \ind{(\epsilon,\infty]}{\discrep{\approxsimdat}{\dat}} + \ind{(\epsilon,\infty]}{\discrep{\simdat}{\dat}} \ind{(-\infty,\epsilon]}{\discrep{\approxsimdat}{\dat}}\right].
	\end{align*}
	We take expectations with respect to $U$ to write
	\begin{align*}
	\CondE{\omega^2}{\paramvec,\approxsimdat,\simdat} 
	&=)
	\ind{(-\infty,\epsilon]}{\discrep{\simdat}{\dat}} \ind{(-\infty,\tilde{\epsilon}]}{\approxdiscrep{\approxsimdat}{\dat}} -\ind{(\epsilon,\infty]}{\discrep{\simdat}{\dat}} \ind{(-\infty,\tilde{\epsilon}]}{\approxdiscrep{\approxsimdat}{\dat}} \\
	&+ \frac{1}{\eta(\approxsimdat)} \left[ \ind{(-\infty,\epsilon]}{\discrep{\simdat}{\dat}} \ind{(\epsilon,\infty]}{\discrep{\approxsimdat}{\dat}} + \ind{(\epsilon,\infty]}{\discrep{\simdat}{\dat}} \ind{(-\infty,\epsilon]}{\discrep{\approxsimdat}{\dat}}\right].
	\end{align*}
	Finally, taking expectations with respect to the pair $(\approxsimdat,\simdat)$, we have 
	\[
	\CondE{\omega^2}{\paramvec} = p_{\tp}( \paramvec) - p_{\fp}( \paramvec) + \frac{1}{\eta_1} p_{\fp}( \paramvec) + \frac{1}{\eta_2} p_{\fn}( \paramvec),
	\]
	for the true positive, false positive, and false negative probabilities
	\begin{align*}
	p_\tp( \paramvec) &= \CondProb{\approxdiscrep{\approxsimdat}{\dat} \leq \tilde{\epsilon},~\discrep{\simdat}{\dat} \leq \epsilon}{\paramvec}, \\
	p_\fp( \paramvec) &= \CondProb{\approxdiscrep{\approxsimdat}{\dat} \leq \tilde{\epsilon},~\discrep{\simdat}{\dat} > \epsilon}{\paramvec}, \\
	p_\fn( \paramvec) &= \CondProb{\approxdiscrep{\approxsimdat}{\dat} > \tilde{\epsilon},~\discrep{\simdat}{\dat} \leq \epsilon}{\paramvec}.
	\end{align*}
	Noting that $\approxlike{\dat}{\paramvec} = \CondProb{\discrep{\simdat}{\dat} \leq \epsilon}{\paramvec} = p_\tp( \paramvec) + p_\fn( \paramvec)$, this can be written as
	\begin{equation}
	\CondE{\omega^2}{\paramvec} = \approxlike{\dat}{\paramvec} 
	+ \left( \frac{1}{\eta_1} - 1 \right) p_\fp( \paramvec) 
	+ \left( \frac{1}{\eta_2} - 1 \right) p_\fn( \paramvec),
	\end{equation}
	and thus we can relate the bias and MSE of MF-ABC to the bias and MSE of standard ABC.
	
	\subsubsection{Multifidelity bias}
	We let the subscript $\cdot_\abc$ denote the output of standard ABC, while $\cdot_\mfabc$ denotes the output of MF-ABC.
	Then the bias of MF-ABC is determined by
	\begin{align*}
	\E{w_\mfabc^2 \Delta} = \E{w_\abc^2 \Delta}
	&+ \frac{1-\eta_1}{\eta_1} \int_{\paramspace} \Delta( \paramvec) \frac{\PDF{\paramvec}}{q( \paramvec)} p_\fp( \paramvec) \PDF{\paramvec} \mathrm d \paramvec \\
	&+ \frac{1-\eta_2}{\eta_2} \int_{\paramspace} \Delta( \paramvec) \frac{\PDF{\paramvec}}{q( \paramvec)} p_\fn( \paramvec) \PDF{\paramvec} \mathrm d \paramvec,
	\end{align*}
	where we note that $\Delta$ may be positive or negative, so that we cannot infer any inequalities from this expression.
	
	Indeed, the difference in bias between MF-ABC and standard ABC is then given by
	\[
	\E{\hat \Delta_\mfabc} = \E{\hat \Delta_\abc} - \frac{1}{N Z^2} \left[ \frac{1 - \eta_1}{\eta_1} B_\fp + \frac{1 - \eta_2}{\eta_2} B_\fn \right] + O(N^{-2})
	\]
	for the expectations based on false positives,
	\begin{align*}
	B_\fp
	&= \int_{\paramspace} \Delta( \paramvec) \left( \frac{\PDF{\paramvec}}{q( \paramvec)} \right)^2 p_\fp( \paramvec) q( \paramvec) \mathrm d \paramvec = \int_{\paramspace} \Delta( \paramvec) \frac{\PDF{\paramvec}}{q( \paramvec)} p_\fp( \paramvec) \PDF{\paramvec} \mathrm d \paramvec
	\end{align*}
	and false negatives,
	\begin{align*}
	B_\fn
	&= \int_{\paramspace} \Delta( \paramvec) \left( \frac{\PDF{\paramvec}}{q( \paramvec)} \right)^2 p_\fn( \paramvec) q( \paramvec) \mathrm d \paramvec = \int_{\paramspace} \Delta( \paramvec) \frac{\PDF{\paramvec}}{q( \paramvec)} p_\fn( \paramvec) \PDF{\paramvec} \mathrm d \paramvec
	\end{align*}
	
	\subsubsection{Multifidelity MSE}
	Similarly, we can write down the MSE of MF-ABC in terms of the MSE of ABC, since
	\begin{align*}
	\E{w_\mfabc^2 \Delta^2} = \E{w_\abc^2 \Delta^2} 
	&+ \frac{1-\eta_1}{\eta_1} \int_{\paramspace} \left( \Delta( \paramvec) \frac{\PDF{\paramvec}}{q( \paramvec)} \right)^2 p_\fp( \paramvec) q( \paramvec) \mathrm d \paramvec \\
	&+ \frac{1-\eta_2}{\eta_2} \int_{\paramspace} \left( \Delta( \paramvec) \frac{\PDF{\paramvec}}{q( \paramvec)} \right)^2 p_\fn( \paramvec) q( \paramvec) \mathrm d \paramvec.
	\end{align*}
	It follows that the MSE can be written as
	\[
	\E{\hat \Delta_\mfabc^2} = \E{\hat \Delta_\abc^2}+ \frac{1}{NZ^2} \left[ \frac{1 - \eta_1}{\eta_1} V_\fp + \frac{1-\eta_2}{\eta_2} V_\fn \right] + O(N^{-2})
	\]
	for the expectations based on false positives,
	\begin{align*}
	V_\fp &= \int_{\paramspace} \left( \Delta( \paramvec) \frac{\PDF{\paramvec}}{q( \paramvec)} \right)^2 p_\fp( \paramvec) q( \paramvec) \mathrm d \paramvec = \int_{\paramspace} \Delta( \paramvec)^2 \frac{\PDF{\paramvec}}{q( \paramvec)} p_\fp( \paramvec) \PDF{\paramvec} \mathrm d \paramvec,
	\end{align*}
	and false negatives,
	\begin{align*}
	V_\fn &= \int_{\paramspace} \left( \Delta( \paramvec) \frac{\PDF{\paramvec}}{q( \paramvec)} \right)^2 p_\fn( \paramvec) q( \paramvec) \mathrm d \paramvec = \int_{\paramspace} \Delta( \paramvec)^2 \frac{\PDF{\paramvec}}{q( \paramvec)} p_\fn( \paramvec) \PDF{\paramvec} \mathrm d \paramvec.
	\end{align*}
	
	\subsubsection{Optimal $\eta_1$ and $\eta_2$}
	.
	In the end, we want to minimise the product of MSE with simulation cost, which is equivalent to minimising the product of
	\[
	Z \CondE{\frac{\pi}{q} \Delta^2}{\dat} + \frac{1-\eta_1}{\eta_1} V_\fp + \frac{1-\eta_2}{\eta_2} V_\fn
	\]
	with the expected simulation cost,
	\[
	\E{\tilde C} + \eta_1 \CondE{C}{\approxdiscrep{\approxsimdat}{\dat} \leq \tilde{\epsilon}} \Prob{\approxdiscrep{\approxsimdat}{\dat} \leq \tilde{\epsilon}} + \eta_2 \CondE{C}{\approxdiscrep{\approxsimdat}{\dat} > \tilde{\epsilon}}\Prob{\approxdiscrep{\approxsimdat}{\dat} > \tilde{\epsilon}}.
	\]
	This leads to Equation (18) in the main manuscript and \eqref{eq:phi} in Section~\ref{sec:adaptMF} .

	\section{Adaptive multifidelity ABC rejection sampling}
	\numberwithin{equation}{section}
	\numberwithin{figure}{section}
	\numberwithin{algorithm}{section}
	\numberwithin{table}{section}
	\setcounter{equation}{0}
	\label{sec:adaptMF}
	
	For simplicity, the main text presents MF-ABC and MF-MLMC-ABC based upon direct multifidelity ABC rejection where continuation probabilities must be pre-specified. Prescott and Baker~\cite{Prescott2020} present an adaptive scheme that explicitly update these probabilities using Equations (19)--(21) in the main text, however, this approach is extremely sensitive to poor initial estimates of the expectations in Equation (21) (main text). Here we utilise an adaptive scheme using exponentiated gradient descent to guide the continuation probabilities toward the optimum. This approach is used in our numerical experiments for the MF-ABC and MF-MLMC-ABC cases.
	
	Firstly, after $M < N$ iterations we compute the following estimates for the expectations in Equation (21) of the main text, 
	\begin{equation}
	\begin{split}
	\hat{\mu} = \sum_{i=1}^M W_i f(\paramvec_i) /\sum_{i=1}^M W_i,&\quad \hat{p}_{tp} = \frac{\rho_m}{\rho_k}\frac{1}{k}\sum_{i \in I_k} (f(\paramvec_i) - \hat{\mu})^2 \tilde{w}_i w_i,  \\
	\hat{p}_{fp} = \frac{\rho_m}{\rho_k}\frac{1}{k}\sum_{i \in I_k} (f(\paramvec_i) - \hat{\mu})^2 \tilde{w}_i(1- w_i), \quad\text{and}&\quad \hat{p}_{fn} = \frac{1-\rho_m}{1-\rho_k}\frac{1}{k}\sum_{i \in I_k} (f(\paramvec_i) - \hat{\mu})^2 (1 - \tilde{w}_i) w_i, \\
	\end{split}
	\label{eq:roc}
	\end{equation}
	and,
	\begin{equation}
	\begin{split}
	\hat{c}^\tau &= \frac{1}{M}\sum_{i=1}^M c^\tau(\paramvec_i), \quad \hat{c}_p = \frac{\rho_m}{\rho_k}\frac{1}{k}\sum_{i \in I_k} c(\paramvec_i)\tilde{w}_i, \quad\text{and}\quad \hat{c}_{n} = \frac{1-\rho_m}{1-\rho_k}\frac{1}{k}\sum_{i \in I_k} c(\paramvec_i)(1 - \tilde{w}_i), \\
	\end{split}
	\label{eq:cost}
	\end{equation}
	where $W_i$,$\tilde{w}_i$,and $w_i$ are, respectively, the full multifidelity weights, the weight of the approximate simulation, and the weight of the exact simulation. Here, $I_k$ is the set of all samples for which an exact simulation has been performed, $k = |I_k|$, $\rho_m$ is the acceptance rate of the approximate simulation and $\rho_k$ is the acceptance rate of the approximate simulation conditional on an exact simulation is generated.
	
	The function to optimise becomes,
	\begin{equation}	
	\phi(\eta_1,\eta_2 ; f) = \left(R_0 + \frac{p_{fp}}{\eta_1}+ \frac{p_{fn}}{\eta_2}\right)\left(c^\tau + \eta_1 c_p + \eta_2 c_n\right).\label{eq:phi}
	\end{equation}
	This leads to the following partial derivatives in the continuation probabilities,
	\begin{equation}
	\begin{split}
	\pdydx{\phi}{\eta_1} &= \left(\frac{1}{\eta_1} +R_0\right) c_p - (c^\tau + \eta_2c_n)\frac{p_{fp}}{\eta_1^2}, \\
	\pdydx{\phi}{\eta_2} &= \left(\frac{1}{\eta_2} +R_0\right) c_n - (c^\tau + \eta_1c_p)\frac{p_{fn}}{\eta_2^2}. \\
	\end{split}
	\label{eq:grad}
	\end{equation}
	
	We then apply exponentiated gradient descent with a learning rate $\delta = 0.1/[(c^\tau+c_p+c_n)\hat{\mu}^2]$. The result is provided in Algorithm~\ref{alg:adgradmfabc}. We also note that this method has been further generalised to other multifidelity schemes by Prescott, Warne and Baker~\cite{Prescott2021b}.
	
	\begin{algorithm}[h]
		\caption{Adaptive Gradient Multifidelity ABC rejection sampling}
		\begin{algorithmic}[1]
			\State Initialise $M < N$, $\epsilon, \tilde{\epsilon}$, $\discrep{\bvec{Y}_{\text{obs}}}{\cdot}$, $\approxdiscrep{\bvec{Y}_{\text{obs}}}{\cdot}$ and prior $p(\paramvec)$;
			\State{Set $\eta_1 \leftarrow 1$ and$ \eta_2 \leftarrow 1$ and $I_k \leftarrow \emptyset$}
			\For {$i = 1,2,\ldots,N$}
			\State{Sample the prior $\paramvec^i \sim p(\paramvec)$};
			\State{Simulate low-fidelity model $\tilde{\bvec{Y}_s} \sim \approxsimProc{\cdot}{\paramvec_i}$};
			\State Set $\tilde{w}_i \leftarrow \ind{(0,\tilde{\epsilon}]}{\approxdiscrep{\bvec{Y}_{\text{obs}}}{\tilde{\bvec{Y}}_s}}$ and $\eta \leftarrow \eta_1 \tilde{w}_i + \eta_2 (1 - \tilde{w}_i)$;
			\If {$U < \eta$ where $U \sim \mathcal{U}(0,1)$}
			\State{Simulate high-fidelity model $\bvec{Y}_s \sim \simProc{\cdot}{\paramvec_i}$};
			\State Set $W_i \leftarrow \tilde{w}_i + (\ind{(0,\epsilon]}{\discrep{\bvec{Y}_{\text{obs}}}{\bvec{Y}_s}} - \tilde{w}_i)/\eta$;
			\State{Set $I_k \leftarrow I_k \cup \{i\}$}
			\Else
			\State Set $W_i \leftarrow \tilde{w}_i$;
			\EndIf
			
			\If {$i > M$}
			\State Compute $\hat{p}_{tp}$, $\hat{p}_{fp}$, $\hat{p}_{fn}$, $\hat{c}^\tau(\paramvec)$, $\hat{c}_p$, $\hat{c}_n$,$\pdydx{\phi}{\eta_1}$ and $\pdydx{\phi}{\eta_2}$ according to Equations (\ref{eq:roc})--(\ref{eq:grad})
			\State{Update $\eta_1 \leftarrow \min\left\{1,\eta_1\exp\left(-\delta \eta_1 \pdydx{\phi}{\eta_1}\right)\right\}$ and $
				\eta_2 \leftarrow \min\left\{1,\eta_2\exp\left(-\delta \eta_2 \pdydx{\phi}{\eta_2}\right)\right\}$}
			\EndIf
			\EndFor
			\State Set $\hat{f} \leftarrow \sum_{i=1}^N W_i f(\paramvec^i)/\sum_{i=1}^N W_i$.
		\end{algorithmic}
		\label{alg:adgradmfabc}
	\end{algorithm}
	
	\FloatBarrier

	\begin{landscape}
		\section{Additional results}
		
		\begin{figure}[h]
			\centering
			\includegraphics[width=0.7\linewidth]{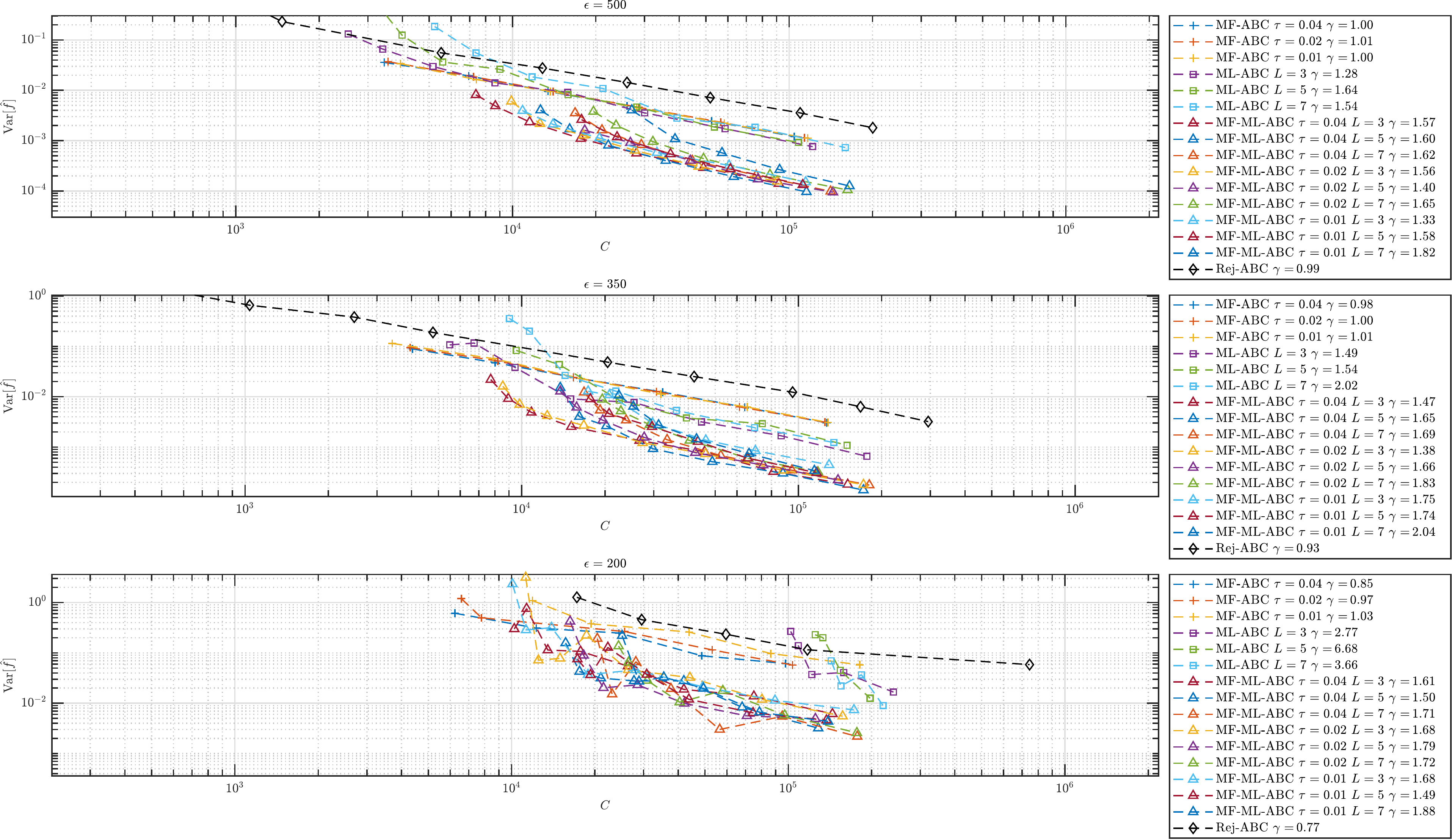}
			\caption{Comparison of convergence rates for a range of configurations MF-MLMC-ABC (triangles) with MF-ABC (crosses), MLMC-ABC (squares), and ABC rejection sampling (black diamonds) using the repressilator model with thresholds: (A) $\epsilon = 500$; (B) $\epsilon = 350$; and (C) $\epsilon = 200$. Rates are estimated by fitting $\V{\hat{f}} \propto C(\hat{f})^{-\gamma}$ to benchmark data using least squares.}
		
		\end{figure}
	\end{landscape}

\end{appendices}


\begin{thebibliography}{100}
	\expandafter\ifx\csname url\endcsname\relax
	\def\url#1{\texttt{#1}}\fi
	\expandafter\ifx\csname urlprefix\endcsname\relax\def\urlprefix{URL }\fi
	\expandafter\ifx\csname href\endcsname\relax
	\def\href#1#2{#2} \def\path#1{#1}\fi
	
	\bibitem{Meulen2017}
	F.~van~der Meulen, M.~Schauer, (2017) Bayesian estimation of incompletely observed
	diffusions 90~(5) 641--662.
	\newblock \href {http://dx.doi.org/10.1080/17442508.2017.1381097}
	{\path{doi:10.1080/17442508.2017.1381097}}.
	
	\bibitem{Dellaportas2006}
	P.~Dellaportas, N.~Friel, G.~O. Roberts, (2006) Bayesian model selection for partially
	observed diffusion models 93~(4) 809--825.
	\newblock \href {http://dx.doi.org/10.1093/biomet/93.4.809}
	{\path{doi:10.1093/biomet/93.4.809}}.
	
	\bibitem{Golightly2006}
	A.~Golightly, D.~J. Wilkinson, (2006) Bayesian sequential inference for nonlinear
	multivariate diffusions 16~(4) 323--338.
	\newblock \href {http://dx.doi.org/10.1007/s11222-006-9392-x}
	{\path{doi:10.1007/s11222-006-9392-x}}.
	
	\bibitem{Elowitz2002}
	M.~B. Elowitz, A.~J. Levine, E.~D. Siggia, P.~S. Swain, (2002)
	{Stochastic gene
		expression in a single cell}, Science 297 1183--1186.
	\newblock \href {http://dx.doi.org/10.1126/science.1070919}
	{\path{doi:10.1126/science.1070919}}.
	
	\bibitem{Raj2008}
	A.~Raj, A.~van Oudenaarden, (2008) Nature, nurture, or chance: Stochastic gene
	expression and its consequences, Cell 135~(2) 216--226.
	\newblock \href {http://dx.doi.org/10.1016/j.cell.2008.09.050}
	{\path{doi:10.1016/j.cell.2008.09.050}}.
	
	\bibitem{Wilkinson2009}
	D.~J. Wilkinson, (2009)
	{Stochastic
		modelling for quantitative description of heterogeneous biological systems},
	Nature Reviews Genetics 10 122--133.
	\newblock \href {http://dx.doi.org/10.1038/nrg2509}
	{\path{doi:10.1038/nrg2509}}.
	
	
	\bibitem{Fedoroff2002}
	N.~Fedoroff, W.~Fontana, (2002) Small numbers of big molecules, Science 297
	1129--1131.\\
	\newblock \href {http://dx.doi.org/10.1126/science.1075988}
	{\path{doi:10.1126/science.1075988}}.
	
	\bibitem{Blake2003}
	W.~J. Blake, M.~Kaern, C.~R. Cantor, J.~J. Collins, (2003)
	{Noise
		in eukaryotic gene expression}, Nature 422 633--637.
	\newblock \href {http://dx.doi.org/10.1038/nature01546}
	{\path{doi:10.1038/nature01546}}.
	
	\bibitem{Braichenko2021}
	S.~Braichenko, J.~Holehouse, R.~Grima, (2021) {Distinguishing between models of mammalian gene expression: telegraph-like models versus mechanistic models}, Journal of the Royal Society Interface 18 20210510 . \newblock \href{https://doi.org/10.1098/rsif.2021.0510}{\path{doi:10.1098/rsif.2021.0510}}.
	
	\bibitem{Gammaitoni1998}
	L.~Gammaitoni, P.~Hanggi, P.~Jung, F.~Marchesoni, (1998)
	{Stochastic
		resonance}, Reviews of Modern Physics 70~(1) 223--287.
	\newblock \href {http://dx.doi.org/10.1103/RevModPhys.70.223}
	{\path{doi:10.1103/RevModPhys.70.223}}.
	
	
	
	\bibitem{Muratov2005}
	C.~B. Muratov, E.~Vanden-Eijnden, E.~Weinan, (2005)
	{Self-induced
		stochastic resonance in excitable systems}, Physica D 210~(3--4)
	227--240.
	\newblock \href {http://dx.doi.org/10.1016/j.physd.2005.07.014}
	{\path{doi:10.1016/j.physd.2005.07.014}}.
	
	\bibitem{Wellens2004}
	T.~Wellens, V.~Shatokhin, A.~Buchleitner, (2004) Stochastic resonance, Reports on
	Progress in Physics 67~(1) 45--105.
	\newblock \href {http://dx.doi.org/10.1088/0034-4885/67/1/R02}
	{\path{doi:10.1088/0034-4885/67/1/R02}}.
	
	\bibitem{Paulsson2000}
	J.~Paulsson, O.~G. Berg, M.~Ehrenberg, (2000)
	{Stochastic focusing:
		Fluctuation-enhanced sensitivity of intracellular regulation}, Proceedings of
	the National Academy of Sciences 97~(13) 7148--7153.\\
	\newblock \href
	{http://dx.doi.org/10.1073/pnas.110057697}
	{\path{doi:10.1073/pnas.110057697}}.
	
	\bibitem{Schlogl1972}
	F.~Schl\"{o}gl, (1972)
	{Chemical reaction
		models for non-equilibrium phase transitions}, Zeitschrift für Physik
	253~(2) 147--161.
	\newblock \href {http://dx.doi.org/10.1007/BF01379769}
	{\path{doi:10.1007/BF01379769}}.
	
	\bibitem{Schnakenberg1979}
	J.~Schnakenberg, (1979)
	{Simple
		chemical reaction systems with limit cycle behaviour}, Journal of Theoretical
	Biology 81~(3) 389--400.
	\newblock \href {http://dx.doi.org/10.1016/0022-5193(79)90042-0}
	{\path{doi:10.1016/0022-5193(79)90042-0}}.
	
	
	\bibitem{Gelman2014}
	A.~Gelman, J.~B. Carlin, H.~S. Stern, D.~B. Dunson, A.~Vehtari, D.~B. Rubin, (2014)
	Bayesian Data Analysis, 3rd Edition, Chapman \& Hall/CRC.
	
	\bibitem{Metroplis1953}
	N.~Metropolis, A.~W. Rosenbluth, M.~N. Rosenbluth, A.~H. Teller, E.~Teller, (1953)
	{Equation
		of state calculations by fast computing machines}, The Journal of Chemical
	Physics 21~(6) 1087--1092.
	\newblock \href {http://dx.doi.org/http://dx.doi.org/10.1063/1.1699114}
	{\path{doi:10.1063/1.1699114}}.
	
	\bibitem{Hastings1970}
	W.~K. Hastings, (1970) {Monte Carlo sampling
		methods using Markov chains and their applications}, Biometrika 57~(1)
	97--109.
	
	\bibitem{DelMoral2006}
	P.~Del~Moral, A.~Doucet, A.~Jasra, (2006)
	{Sequential Monte
		Carlo samplers}, Journal of the Royal Statistical Society: Series B
	(Statistical Methodology) 68~(3) 411--436.\\
	\newblock \href {http://dx.doi.org/10.1111/j.1467-9868.2006.00553.x}
	{\path{doi:10.1111/j.1467-9868.2006.00553.x}}.
	
	\bibitem{Sisson2018}
	S.~A. Sisson, Y.~Fan, M.~Beaumont, (2018) Handbook of Approximate Bayesian
	Computation, Chapman and Hall/CRC Press.
	
	\bibitem{Sunnaker2013}
	M.~Sunn{\aa}ker, A.~G. Busetto, E.~Numminen, J.~Corander, M.~Foll, C.~Dessimoz, (2013)
	Approximate Bayesian computation, PLOS Computational Biology 9~(1)
	e1002803.\\
	\newblock \href {http://dx.doi.org/doi:10.1371/journal.pcbi.1002803}
	{\path{doi:10.1371/journal.pcbi.1002803}}.
	
	\bibitem{Warne2019a}
	D.~J. Warne, R.~E. Baker, M.~J. Simpson, (2019) Simulation and inference algorithms
	for stochastic biochemical reaction networks: from basic concepts to
	state-of-the-art, Journal of The Royal Society Interface 16~(151)
	20180943.
	\newblock \href {http://dx.doi.org/10.1098/rsif.2018.0943}
	{\path{doi:10.1098/rsif.2018.0943}}.
	
	\bibitem{Andrieu2009}
	C.~Andrieu, G.~O. Roberts, (2009) {The
		pseudo-marginal approach for efficient Monte Carlo computations}, The Annals
	of Statistics 37~(2) 697--725.
	
	\bibitem{Andrieu2010}
	C.~Andrieu, A.~Doucet, R.~Holenstein, (2010)
	{Particle Markov chain Monte Carlo
		methods}, Journal of the Royal Statistical Society. Series B (Statistical
	Methodology) 72~(3) 269--342.
	
	\bibitem{Warne2020}
	D.~J. Warne, R.~E. Baker, M.~J. Simpson, (2020) A practical guide to pseudo-marginal
	methods for computational inference in systems biology, Journal of
	Theoretical Biology 496 110255.
	\newblock \href {http://dx.doi.org/10.1016/j.jtbi.2020.110255}
	{\path{doi:10.1016/j.jtbi.2020.110255}}.
	
	\bibitem{An2019}
	Z.~An, D.~J. Nott, C.~Drovandi, (2019) Robust Bayesian synthetic likelihood via a
	semi-parametric approach, Statistics and Computing 30~(3) 543--557.
	\newblock \href {http://dx.doi.org/10.1007/s11222-019-09904-x}
	{\path{doi:10.1007/s11222-019-09904-x}}.
	
	\bibitem{Ong2017}
	V.~M.~H. Ong, D.~J. Nott, M.-N. Tran, S.~A. Sisson, C.~C. Drovandi, (2017) Variational
	Bayes with synthetic likelihood, Statistics and Computing 28~(4)
	971--988.
	\newblock \href {http://dx.doi.org/10.1007/s11222-017-9773-3}
	{\path{doi:10.1007/s11222-017-9773-3}}.
	
	\bibitem{Price2017}
	L.~F. Price, C.~C. Drovandi, A.~Lee, D.~J. Nott, (2017) Bayesian synthetic likelihood,
	Journal of Computational and Graphical Statistics 27~(1) 1--11.
	\newblock \href {http://dx.doi.org/10.1080/10618600.2017.1302882}
	{\path{doi:10.1080/10618600.2017.1302882}}.
	
	\bibitem{Lueckmann2017}
	J.-M.~Lueckmann, P.~J. Goncalves, G.~Bassetto, K.~\"{O}cal, M.~Nonenmacher, J.~H. Macke, (2017) Flexible statistical inference for mechanistic models of neural dynamics, Advances in Neural Information Processing Systems 30.
	
	\bibitem{Papamakarios2019}
	G.~Papamakarios, D.~Sterratt, I.~Murray, (2019) Sequential neural likelihood: Fast likelihood-free Inference with autoregressive flows, Proceedings of Machine Learning Research 89 837--848.
	
	
	\bibitem{Cranmer2020}
	K.~Cranmer, J.~Brehmer, G.~Louppe, (2020) The frontier of simulation-based inference, Proceedings of the National Academy of Sciences of the United States of America 117~(48) 30055--30062.\\
		\newblock \href {http://dx.doi.org/10.1073/pnas.1912789117 }
		{\path{doi:10.1073/pnas.1912789117 }}.
	
	
	\bibitem{Pritchard1999}
	J.~K. Pritchard, M.~T. Seielstad, A.~Perez-Lezaun, M.~W. Feldman, (1999)
	{Population
		growth of human y chromosomes: a study of y chromosome microsatellites.},
	Molecular Biology and Evolution 16~(12) 1791--1798.
	\newblock
	\href {http://dx.doi.org/10.1093/oxfordjournals.molbev.a026091}
	{\path{doi:10.1093/oxfordjournals.molbev.a026091}}.
	
	
	\bibitem{Tavare1997}
	S.~Tavar{\'e}, D.~J. Balding, R.~C. Griffiths, P.~Donnelly, (1997)
	{Inferring coalescence times
		from DNA sequence data}, Genetics 145~(2) 505--518.
	
	\bibitem{Beaumont2002}
	M.~A. Beaumont, W.~Zhang, D.~J. Balding, (2002)
	{Approximate Bayesian
		computation in population genetics}, Genetics 162~(4) 2025--2035.
	
	
	\bibitem{Marjoram2003}
	P.~Marjoram, J.~Molitor, V.~Plagnol, S.~Tavar{\'e}, (2003) Markov chain Monte Carlo
	without likelihoods, Proceedings of the National Academy of Sciences of the
	United States of America 100~(26) 15324--15328.
	\newblock \href {http://dx.doi.org/10.1073/pnas.0306899100}
	{\path{doi:10.1073/pnas.0306899100}}.
	
	\bibitem{Sisson2007}
	S.~A. Sisson, Y.~Fan, M.~M. Tanaka, (2007) Sequential Monte Carlo without likelihoods,
	Proceedings of the National Academy of Sciences of the United States of
	America 104~(6) 1760--1765.\\
	\newblock \href {http://dx.doi.org/10.1073/pnas.0607208104}
	{\path{doi:10.1073/pnas.0607208104}}.
	
	\bibitem{DelMoral2012}
	P.~Del~Moral, A.~Doucet, A.~Jasra, (2012)
	{An adaptive sequential
		Monte Carlo method for approximate Bayesian computation}, Statistics and
	Computing 22~(5) 1009--1020.\\
	\newblock \href {http://dx.doi.org/10.1007/s11222-011-9271-y}
	{\path{doi:10.1007/s11222-011-9271-y}}.
	
	\bibitem{Drovandi2011}
	C.~C. Drovandi, A.~N. Pettitt, (2011)
	{Estimation of
		parameters for macroparasite population evolution using approximate Bayesian
		computation}, Biometrics 67~(1) 225--233.\\
	\newblock \href {http://dx.doi.org/10.1111/j.1541-0420.2010.01410.x}
	{\path{doi:10.1111/j.1541-0420.2010.01410.x}}.
	
	\bibitem{Bon2020}
	J.~J. Bon, A.~Lee, C.~Drovandi, (2021) Accelerating sequential {Monte Carlo} with
	surrogate likelihoods, Statistics and Computing (In Press), \href
	{http://arxiv.org/abs/2009.03699} {\path{arXiv:2009.03699}}.
	
	\bibitem{Banterle2019}
	M.~Banterle, , C.~Grazian, A.~Lee, C.~P. Robert,  (2019) Accelerating
	Metropolis-Hastings algorithms by delayed acceptance, Foundations of Data
	Science 1~(2) 103--128.
	\newblock \href {http://dx.doi.org/10.3934/fods.2019005}
	{\path{doi:10.3934/fods.2019005}}.
	
	\bibitem{Everitt2020}
	R.~G. Everitt, P.~A. Rowi{\'{n}}ska, (2020) Delayed acceptance {ABC}-{SMC}, Journal of
	Computational and Graphical Statistics 1--12\href
	{http://dx.doi.org/10.1080/10618600.2020.1775617}
	{\path{doi:10.1080/10618600.2020.1775617}}.
	
	\bibitem{Golightly2014}
	A.~Golightly, D.~A. Henderson, C.~Sherlock, (2014) Delayed acceptance particle {MCMC}
	for exact inference in stochastic kinetic models, Statistics and Computing
	25~(5) 1039--1055.\\
	\newblock \href {http://dx.doi.org/10.1007/s11222-014-9469-x}
	{\path{doi:10.1007/s11222-014-9469-x}}.
	
	\bibitem{Prangle2014}
	D.~Prangle, (2014)
	{Lazy
		ABC}, Statistics and Computing 26~(1) 171--185.\\
	\newblock \href {http://dx.doi.org/0.1007/s11222-014-9544-3}
	{\path{doi:0.1007/s11222-014-9544-3}}.
	
	\bibitem{Parno2018}
	M.~D. Parno, Y.~M. Marzouk, (2018) Transport map accelerated Markov chain Monte Carlo,
	{SIAM}/{ASA} Journal on Uncertainty Quantification 6~(2) 645--682.
	\newblock \href {http://dx.doi.org/10.1137/17m1134640}
	{\path{doi:10.1137/17m1134640}}.
	
	\bibitem{Warne2019b}
	D.~J. Warne, R.~E. Baker, M.~J. Simpson, (2021) Rapid Bayesian inference for expensive stochastic models, Journal of Computational and Graphical Statistics (In Press),\href
	{http://arxiv.org/abs/1909.06540} {\path{arXiv:1909.06540}}.
	
	\bibitem{Prescott2020}
	T.~P. Prescott, R.~E. Baker, (2020) Multifidelity approximate Bayesian computation,
	{SIAM}/{ASA} Journal on Uncertainty Quantification 8~(1) 114--138.
	\newblock \href {http://dx.doi.org/10.1137/18m1229742}
	{\path{doi:10.1137/18m1229742}}.
	
	\bibitem{Prescott2021}
	T.~P. Prescott, R.~E. Baker, (2021) Multifidelity approximate Bayesian computation
	with sequential Monte Carlo parameter sampling, {SIAM}/{ASA} Journal on
	Uncertainty Quantification 9~(2) 788--817.
	\newblock \href {http://dx.doi.org/10.1137/20m1316160}
	{\path{doi:10.1137/20m1316160}}.
	
	\bibitem{Peherstorfer2016}
	B.~Peherstorfer, K.~Willcox, M.~Gunzburger, (2016)
	{Optimal model management for
		multifidelity Monte Carlo estimation}, SIAM Journal on Scientific Computing
	38~(5) A3163--A3194.\\
	\newblock  \href
	{http://dx.doi.org/10.1137/15M1046472} {\path{doi:10.1137/15M1046472}}.
	
	\bibitem{Heinrich1998}
	S.~Heinrich, (1998)
	{Monte
		Carlo complexity of global solution of integral equations}, Journal of
	Complexity 14~(2) 151--175.
	\newblock \href {http://dx.doi.org/10.1006/jcom.1998.0471}
	{\path{doi:10.1006/jcom.1998.0471}}.
	
	\bibitem{Giles2015}
	M.~B. Giles, (2015) Multilevel Monte Carlo methods, Acta Numerica 24 259--328.\\
	\newblock \href {http://dx.doi.org/doi:10.1017/S09624929}
	{\path{doi:10.1017/S09624929}}.
	
	\bibitem{Giles2008}
	M.~B. Giles,
	(2008) Multilevel Monte Carlo path simulation, Operations Research 56~(3) 607--617.
	\newblock \href {http://dx.doi.org/10.1287/opre.1070.0496}
	{\path{doi:10.1287/opre.1070.0496}}.
	
	\bibitem{Anderson2012}
	D.~F. Anderson, D.~J. Higham, (2012)
	{Multilevel Monte Carlo for
		continuous time Markov chains, with applications in biochemical kinetics},
	Multiscale Modeling \& Simulation 10~(1) 146--179.\\
	\newblock  \href
	{http://dx.doi.org/10.1137/110840546} {\path{doi:10.1137/110840546}}.
	
	\bibitem{Lester2014}
	C.~Lester, R.~E. Baker, M.~B. Giles, C.~A. Yates, (2016)
	{Extending the multi-level
		method for the simulation of stochastic biological systems}, Bulletin of
	Mathematical Biology 78~(8) 1640--1677.\\
	\newblock \href {http://dx.doi.org/10.1007/s11538-016-0178-9}
	{\path{doi:10.1007/s11538-016-0178-9}}.
	
	\bibitem{Lester2015}
	C.~Lester, C.~A. Yates, M.~B. Giles, R.~E. Baker, (2015)
	{An
		adaptive multi-level simulation algorithm for stochastic biological systems},
	The Journal of Chemical Physics 142~(2) 024113.\\
	\newblock \href {http://dx.doi.org/10.1063/1.4904980}
	{\path{doi:10.1063/1.4904980}}.
	
	
	\bibitem{Dodwall2015}
	T.~J. Dodwell, C.~Ketelsen, R.~Scheichl, A.~L. Teckentrup,
	(2015)
	{A hierarchical multilevel Markov
		chain Monte Carlo algorithm with applications to uncertainty quantification
		in subsurface flow}, SIAM/ASA Journal on Uncertainty Quantification 3~(1) 1075--1108.
	\newblock \href
	{http://dx.doi.org/10.1137/130915005} {\path{doi:10.1137/130915005}}.
	
	\bibitem{Dodwell2019}
	T.~J. Dodwell, C.~Ketelsen, R.~Scheichl, A.~L. Teckentrup, (2019) Multilevel Markov
	chain Monte Carlo, {SIAM} Review 61~(3) 509--545.
	\newblock \href {http://dx.doi.org/10.1137/19m126966x}
	{\path{doi:10.1137/19m126966x}}.
	
	\bibitem{Efendiev2015}
	Y.~Efendiev, B.~Jin, P.~Michael, X.~Tan, (2015) Multilevel Markov Chain Monte Carlo Method for High-Contrast Single-Phase Flow Problems. Communications in Computational Physics, 17(1), 259-286.
	\newblock \href{http://dx.doi.org/10.4208/cicp.021013.260614a}{\path{doi:10.4208/cicp.021013.260614a}}
	
	\bibitem{Beskos2016}
	A.~Beskos, A.~Jasra, K.~Law, R.~Tempone, Y.~Zhou,
	(2016)
	{Multilevel
		sequential Monte Carlo samplers}, Stochastic Processes and their Applications \href
	{http://dx.doi.org/http://dx.doi.org/10.1016/j.spa.2016.08.004}
	{\path{doi:10.1016/j.spa.2016.08.004}}.
	
	\bibitem{Latz2018}
	J.~Latz, I.~Papaioannou, E.~Ullmann, (2018) Multilevel sequential$^2$ Monte Carlo for
	Bayesian inverse problems 368 154--178.
	\newblock \href {http://dx.doi.org/10.1016/j.jcp.2018.04.014}
	{\path{doi:10.1016/j.jcp.2018.04.014}}.
	
	\bibitem{Guha2017}
	N.~Guha, X.~Tan, (2017)
	{Multilevel
		approximate Bayesian approaches for flows in highly heterogeneous porous
		media and their applications}, Journal of Computational and Applied
	Mathematics 317 700 -- 717.
	\newblock \href {http://dx.doi.org/http://dx.doi.org/10.1016/j.cam.2016.10.008}
	{\path{doi:10.1016/j.cam.2016.10.008}}.
	
	\bibitem{Jasra2019}
	A.~Jasra, S.~Jo, D.~Nott, C.~Shoemaker, R.~Tempone,
	(2019) Multilevel Monte Carlo in
	approximate Bayesian computation, Stochastic Analysis and Applications 37~(3) 346--360.\\
	\newblock \href {http://dx.doi.org/10.1080/07362994.2019.1566006}
	{\path{doi:10.1080/07362994.2019.1566006}}.
	
	\bibitem{Warne2018}
	D.~J. Warne, R.~E. Baker, M.~J. Simpson, (2018)
	{Multilevel
		rejection sampling for approximate Bayesian computation}, Computational
	Statistics and Data Analysis 124 71--86.\\
	\newblock 
	\href {http://dx.doi.org/https://doi.org/10.1016/j.csda.2018.02.009}
	{\path{doi:10.1016/j.csda.2018.02.009}}.
	
	\bibitem{Prescott2021b}
	T.~P. Prescott, D.~J. Warne, R.~E. Baker (2021) {Efficient multifidelity likelihood-free Bayesian inference with adaptive computational resource allocation},
	ArXiv e-prints.
	\newline\urlprefix\url{http://arxiv.org/abs/2112.11971}
	
	\bibitem{Dhananjaneyulu2012}
	V.~Dhananjaneyulu, V.~N.~S. P, G.~Kumar, G.~A. Viswanathan, (2012) Noise propagation
	in two-step series {MAPK} cascade, {PLoS} {ONE} 7~(5) e35958.
	\newblock \href {http://dx.doi.org/10.1371/journal.pone.0035958}
	{\path{doi:10.1371/journal.pone.0035958}}.
	
	\bibitem{Brown2004}
	K.~S. Brown, C.~C. Hill, G.~A. Calero, C.~R. Myers, K.~H. Lee, J.~P. Sethna, (2004)
	R.~A. Cerione, The statistical mechanics of complex signaling networks: nerve
	growth factor signaling, Physical Biology 1~(3) 184--195.
	\newblock \href {http://dx.doi.org/10.1088/1478-3967/1/3/006}
	{\path{doi:10.1088/1478-3967/1/3/006}}.
	
	\bibitem{Schnoerr2017}
	D.~Schnoerr, G.~Sanguinetti, R.~Grima, (2017)
	{Approximation and
		inference methods for stochastic biochemical kinetics—a tutorial review},
	Journal of Physics A: Mathematical and Theoretical 50~(9) 093001.
	
	\bibitem{Higham2008}
	D.~J. Higham, (2008) {Modeling and simulating
		chemical reactions}, SIAM Review 50~(2) 347--368.
	\newblock \href
	{http://dx.doi.org/10.1137/060666457} {\path{doi:10.1137/060666457}}.
	
	
	\bibitem{Gillespie1977}
	D.~T. Gillespie, (1977) {Exact
		stochastic simulation of coupled chemical reactions}, The Journal of Physical
	Chemistry 81~(25) 2340--2361.
	\newblock \href {http://dx.doi.org/10.1021/j100540a008}
	{\path{doi:10.1021/j100540a008}}.
	
	\bibitem{Kurtz1972}
	T.~G. Kurtz, (1972)
	{The
		relationship between stochastic and deterministic models for chemical
		reactions}, The Journal of Chemical Physics 57~(7) 2976--2978.
	\newblock \href {http://dx.doi.org/10.1063/1.1678692}
	{\path{doi:10.1063/1.1678692}}.
	
	
	\bibitem{Erban2007}
	R.~Erban, S.~J. Chapman, P.~K. Maini, (2007) {A
		practical guide to stochastic simulation of reaction-diffusion processes},
	ArXiv e-prints.
	\newline\urlprefix\url{http://arxiv.org/abs/0704.1908}
	
	\bibitem{Gillespie1992}
	D.~T. Gillespie, (1992) A rigorous derivation of the chemical master equation, Physica
	A 188 404--425.
	\newblock \href {http://dx.doi.org/10.1016/0378-4371(92)90283-V}
	{\path{doi:10.1016/0378-4371(92)90283-V}}.
	
	\bibitem{Gibson2000}
	M.~A. Gibson, J.~Bruck, (2000)
	{Efficient exact
		stochastic simulation of chemical systems with many species and many
		channels}, The Journal of Physical Chemistry 104~(9) 1876--1889.\\
	\newblock \href {http://dx.doi.org/10.1021/jp993732q}
	{\path{doi:10.1021/jp993732q}}.
	
	\bibitem{Anderson2007}
	D.~F. Anderson,
	(2007)
	{A
		modified next reaction method for simulating chemical systems with time
		dependent propensities and delays}, The Journal of Chemical Physics 127 214107.\\
	\newblock \href {http://dx.doi.org/10.1063/1.2799998}
	{\path{doi:10.1063/1.2799998}}.
	
	
	\bibitem{Gillespie2001}
	D.~T. Gillespie, (2001) Approximate accelerated simulation of chemically reacting
	systems, The Journal of Chemical Physics 115~(4) 1716--1733.
	\newblock \href {http://dx.doi.org/10.1063/1.1378322}
	{\path{doi:10.1063/1.1378322}}.
	
	\bibitem{Gillespie2000}
	D.~T. Gillespie, (2000)
	{The
		chemical Langevin equation}, The Journal of Chemical Physics 113~(1)
	297--306.
	\newblock \href {http://dx.doi.org/10.1063/1.481811}
	{\path{doi:10.1063/1.481811}}.
	
	\bibitem{Tian}
	T.~Tian, K.~Burrage, (2004) {Binomial leap
		methods for simulating stochastic chemical kinetics}, The Journal of Chemical
	Physics 121~(21) 10356--10364.
	\newblock\href
	{http://dx.doi.org/10.1063/1.1810475} {\path{doi:10.1063/1.1810475}}.
	
	
	\bibitem{Cao2004}
	Y.~Cao, H.~Li, L.~Petzold, (2004)
	{Efficient
		formulation of the stochastic simulation algorithm for chemically reacting
		systems}, The Journal of Chemical Physics 121 4059--4067.
	\newblock \href {http://dx.doi.org/10.1063/1.1778376}
	{\path{doi:10.1063/1.1778376}}.
	
	\bibitem{Anderson2011}
	D.~F. Anderson, A.~Ganguly, T.~G. Kurtz, (2011)
	{Error analysis of tau-leap
		simulation methods}, Annals of Applied Probability 21~(6) 2226--2262.
	\newblock \href {http://dx.doi.org/10.1214/10-AAP756}
	{\path{doi:10.1214/10-AAP756}}.
	
	\bibitem{Wilson2016}
	D.~Wilson, R.~E. Baker, (2016)
		{Multi-level methods and approximating distribution  functions}, AIP Advances 6 075020.
		\newblock \href {https://doi.org/10.1063/1.4960118}
		{\path{doi:10.1063/1.4960118}}.
	
	\bibitem{Giles2015b}
	M.~B. Giles, T.~Nagapetyan, K.~Ritter, (2015)
		{Mulitlevel Monte Carlo approximation od distribution functions and densities}, SIAM/ASA Journal on Uncertainty Quantification, 3~(1) 267--295.\\ \newblock \href {https://doi.org/10.1137/140960086}
		{\path{doi:10.1137/140960086}}.
	
	\bibitem{Fearnhead2012}
	P.~Fearnhead, D.~Prangle, (2012)
	{Constructing
		summary statistics for approximate Bayesian computation: semi-automatic
		approximate Bayesian computation}, Journal of the Royal Statistical Society
	Series B (Statistical Methodology) 74~(3) 419--474.
	\newblock \href {http://dx.doi.org/10.1111/j.1467-9868.2011.01010.x}
	{\path{doi:10.1111/j.1467-9868.2011.01010.x}}.
	
	
	\bibitem{Rhee2015}
	C.-H. Rhee, P.~W. Glynn, (2015)
	{Unbiased estimation with
		square root convergence for SDE models}, Operations Research 63~(5)
	1026--1043.
	\newblock \href {http://dx.doi.org/10.1287/opre.2015.1404}
	{\path{doi:10.1287/opre.2015.1404}}.
	
	\bibitem{Peherstorfer_2018}
	B.~Peherstorfer, K.~Willcox, M.~Gunzburger,
	(2018) Survey of multifidelity methods in
	uncertainty propagation, inference, and optimization, {SIAM} Review 60~(3) 550--591.
	\newblock \href {http://dx.doi.org/10.1137/16m1082469}
	{\path{doi:10.1137/16m1082469}}.
	
	\bibitem{Warne2020b}
	D.~J. Warne (2020) {Computational inference in mathematical biology: Methodological developments and applications}, PhD Thesis, Queensland University of Technology.\\ \newblock \href {http://dx.doi.org/10.5204/thesis.eprints.202835}
	{\path{doi:10.5204/thesis.eprints.202835}}.
	
	\bibitem{Michaelis1913}
	L.~Michaelis, M.~L. Menten, (1913) Die kinetik der invertinwirkung, Biochem Z 49 333--369.
	
	\bibitem{Rao2003}
	C.~V. Rao, A.~P. Arkin, (2003) {Stochastic
		chemical kinetics and the quasi-steady-state assumption: Application to the
		Gillespie algorithm}, The Journal of Chemical Physics 118~(11)
	4999--5010.\\
	\newblock \href
	{http://dx.doi.org/10.1063/1.1545446} {\path{doi:10.1063/1.1545446}}.
	
	\bibitem{Finkenstadt2008}
	B.~Finkenstädt, E.~A. Heron, M.~Komorowski, K.~Edwards, S.~Tang, C.~V. Harper,
	J.~R.~E. Davis, M.~R.~H. White, A.~J. Millar, D.~A. Rand, (2008)
	{Reconstruction of
		transcriptional dynamics from gene reporter data using differential
		equations}, Bioinformatics 24~(24) 2901--2907.\\
	\newblock 
	\href {http://dx.doi.org/10.1093/bioinformatics/btn562}
	{\path{doi:10.1093/bioinformatics/btn562}}.
	
	\bibitem{Iafolla2008}
	M.~A.~J. Iafolla, M.~Mazumder, V.~Sardana, T.~Velauthapillai, K.~Pannu, D.~R.
	McMillen, (2008)
	{Dark
		proteins: Effect of inclusion body formation on quantification of protein
		expression}, Proteins: Structure, Function, and Bioinformatics 72~(4)
	1233--1242.
	\newblock 
	\href {http://dx.doi.org/10.1002/prot.22024} {\path{doi:10.1002/prot.22024}}.
	
	
	\bibitem{Bajar2016}
	B.~Bajar, E.~Wang, S.~Zhang, M.~Lin, J.~Chu, (2016) A guide to fluorescent protein
	{FRET} pairs, Sensors 16~(9) 1488.
	\newblock \href {http://dx.doi.org/10.3390/s16091488}
	{\path{doi:10.3390/s16091488}}.
	
	\bibitem{Simpson2020}
	M.~J. Simpson, R.~E. Baker, S.~T. Vittadello, O.~J. Maclaren, (2020) Practical
	parameter identifiability for spatio-temporal models of cell invasion,
	Journal of The Royal Society Interface 17~(164) 20200055.
	\newblock \href {http://dx.doi.org/10.1098/rsif.2020.0055}
	{\path{doi:10.1098/rsif.2020.0055}}.
	
	\bibitem{Georgii2012}
	H.-O. Georgii, (2013) Stochastics, {De} {Gruyter}.
	\newblock \href {http://dx.doi.org/10.1515/9783110293609}
	{\path{doi:10.1515/9783110293609}}.
	
	\bibitem{Elowitz2000}
	M.~B. Elowitz, S.~Leibler, (2000) A synthetic oscillatory network of transcriptional
	regulators, Nature 403~(6767) 335--338.
	\newblock \href {http://dx.doi.org/10.1038/35002125}
	{\path{doi:10.1038/35002125}}.
	
	\bibitem{Toni2009}
	T.~Toni, D.~Welch, N.~Strelkowa, A.~Ipsen, M.~P.~H. Stumpf, (2009)
	{Approximate
		Bayesian computation scheme for parameter inference and model selection in
		dynamical systems}, Journal of the Royal Society Interface 6 187--202.
	\newblock \href {http://dx.doi.org/doi:10.1098/rsif.2008.0172}
	{\path{doi:doi:10.1098/rsif.2008.0172}}.
	
	\bibitem{Oda2005}
	K.~Oda, Y.~Matsuoka, A.~Funahashi, H.~Kitano, (2005) A comprehensive pathway map of
	epidermal growth factor receptor signaling, Physical Biology 1~(1).
	\newblock \href {http://dx.doi.org/10.1038/msb4100014}
	{\path{doi:10.1038/msb4100014}}.
	
	\bibitem{Priddle2021}
	J.~W. Priddle, S.~A. Sisson, D.~T. Frazier, I.~Turner, C.~Drovandi, (2021) Efficient
	Bayesian synthetic likelihood with whitening transformations, Journal of
	Computational and Graphical Statistics 1--27 \href
	{http://dx.doi.org/10.1080/10618600.2021.1979012}
	{\path{doi:10.1080/10618600.2021.1979012}}.
	
	\bibitem{Jasra2018}
	A.~Jasra, K.~Kamatani, K.~Law, Y.~Zhou, (2018) 
	{Bayesian static parameter
		estimation for partially observed diffusions via multilevel Monte Carlo},
	SIAM Journal on Scientific Computing 40~(2)A887--A902.
	\newblock \href
	{http://dx.doi.org/10.1137/17M1112595} {\path{doi:10.1137/17M1112595}}.
	
	\bibitem{Liepe2014}
	J.~Liepe, P.~Kirk, S.~Filippi, T.~Toni, C.~P. Barnes, M.~P.~H. Stumpf, (2014) A
	framework for parameter estimation and model selection from experimental data
	in systems biology using approximate Bayesian computation, Nature Protocols
	9~(2) 439--456.
	\newblock \href {http://dx.doi.org/10.1038/nprot.2014.025}
	{\path{doi:10.1038/nprot.2014.025}}.
	
	\bibitem{Wu2014}
	Q.~Wu, K.~Smith-Miles, T.~Tian, (2014) Approximate Bayesian computation schemes for
	parameter inference of discrete stochastic models using simulated likelihood
	density, {BMC} Bioinformatics 15~(S12).
	\newblock \href {http://dx.doi.org/10.1186/1471-2105-15-s12-s3}
	{\path{doi:10.1186/1471-2105-15-s12-s3}}.
	
	\bibitem{Chinazzi2020}
	M.~Chinazzi, J.~T. Davis, M.~Ajelli, C.~Gioannini, M.~Litvinova, S.~Merler,
	A.~P. y~Piontti, K.~Mu, L.~Rossi, K.~Sun, C.~Viboud, X.~Xiong, H.~Yu, M.~E.
	Halloran, I.~M. Longini, A.~Vespignani, (2020) The effect of travel restrictions on
	the spread of the 2019 novel coronavirus ({COVID}-19) outbreak, Science
	368~(6489) 395--400.
	\newblock \href {http://dx.doi.org/10.1126/science.aba9757}
	{\path{doi:10.1126/science.aba9757}}.
	
	\bibitem{McKinley2018}
	T.~J. McKinley, I.~Vernon, I.~Andrianakis, N.~McCreesh, J.~E. Oakley, R.~N.
	Nsubuga, M.~Goldstein, R.~G. White, (2018) Approximate Bayesian computation and
	simulation-based inference for complex stochastic epidemic models,
	Statistical Science 33~(1).
	\newblock \href {http://dx.doi.org/10.1214/17-sts618}
	{\path{doi:10.1214/17-sts618}}.
	
	\bibitem{Minter2019}
	A.~Minter, R.~Retkute, (2019) Approximate Bayesian computation for infectious disease
	modelling, Epidemics 29 100368.
	\newblock \href {http://dx.doi.org/10.1016/j.epidem.2019.100368}
	{\path{doi:10.1016/j.epidem.2019.100368}}.
	
	\bibitem{Walker2019}
	J.~N. Walker, A.~J. Black, J.~V. Ross, (2019) Bayesian model discrimination for
	partially-observed epidemic models, Mathematical Biosciences 317 108266.
	\newblock \href {http://dx.doi.org/10.1016/j.mbs.2019.108266}
	{\path{doi:10.1016/j.mbs.2019.108266}}.
	
	\bibitem{Warne2020a}
	D.~J. Warne, A.~Ebert, C.~Drovandi, W.~Hu, A.~Mira, K.~Mengersen, (2020) Hindsight is
	2020 vision: a characterisation of the global response to the {COVID}-19
	pandemic, {BMC} Public Health 20~(1).
	\newblock \href {http://dx.doi.org/10.1186/s12889-020-09972-z}
	{\path{doi:10.1186/s12889-020-09972-z}}.
	
	\bibitem{Beaumont2010}
	M.~A. Beaumont, (2010) Approximate Bayesian computation in evolution and ecology,
	Annual Review of Ecology, Evolution, and Systematics 41~(1) 379--406.\\
	\newblock \href {http://dx.doi.org/10.1146/annurev-ecolsys-102209-144621}
	{\path{doi:10.1146/annurev-ecolsys-102209-144621}}.
	
	\bibitem{Siren2018}
	J.~Sir{\'{e}}n, L.~Lens, L.~Cousseau, O.~Ovaskainen, (2018) Assessing the dynamics of
	natural populations by fitting individual-based models with approximate
	Bayesian computation, Methods in Ecology and Evolution 9~(5) 1286--1295.
	\newblock \href {http://dx.doi.org/10.1111/2041-210x.12964}
	{\path{doi:10.1111/2041-210x.12964}}.
	
	\bibitem{Akeret2015}
	J.~Akeret, A.~Refregier, A.~Amara, S.~Seehars, C.~Hasner, (2015) Approximate Bayesian
	computation for forward modeling in cosmology, Journal of Cosmology and
	Astroparticle Physics 2015~(08) 043--043.
	\newblock \href {http://dx.doi.org/10.1088/1475-7516/2015/08/043}
	{\path{doi:10.1088/1475-7516/2015/08/043}}.
	
	\bibitem{BarajasSolano2019}
	D.~Barajas-Solano, A.~Tartakovsky, (2019) Approximate Bayesian model inversion for
	{PDEs} with heterogeneous and state-dependent coefficients, Journal of
	Computational Physics 395 247--262.\\
	\newblock \href {http://dx.doi.org/10.1016/j.jcp.2019.06.010}
	{\path{doi:10.1016/j.jcp.2019.06.010}}.
	
	\bibitem{Christopher2018}
	J.~D. Christopher, N.~T. Wimer, C.~Lapointe, T.~R.~S. Hayden, I.~Grooms, G.~B.
	Rieker, P.~E. Hamlington,
	(2018) Parameter estimation for complex thermal-fluid
	flows using approximate Bayesian computation, Physical Review Fluids 3~(10) 104602.
	\newblock \href {http://dx.doi.org/10.1103/physrevfluids.3.104602}
	{\path{doi:10.1103/physrevfluids.3.104602}}.
	
	\bibitem{Transtrum2014}
	M.~K. Transtrum, P.~Qiu, (2014) Model reduction by manifold boundaries, Physical
	Review Letters 113~(9) 098701.
	\newblock \href {http://dx.doi.org/10.1103/physrevlett.113.098701}
	{\path{doi:10.1103/physrevlett.113.098701}}.
	
	\bibitem{Browning2020}
	A.~P. Browning, D.~J. Warne, K.~Burrage, R.~E. Baker, M.~J. Simpson, (2020)
	Identifiability analysis for stochastic differential equation models in
	systems biology, Journal of The Royal Society Interface 17~(173)
	20200652.
	\newblock \href {http://dx.doi.org/10.1098/rsif.2020.0652}
	{\path{doi:10.1098/rsif.2020.0652}}.
	
	\bibitem{Cao2018}
	Z.~Cao, R.~Grima, (2018) Linear mapping approximation of gene regulatory networks with
	stochastic dynamics, Nature Communications 9~(1).
	\newblock \href {http://dx.doi.org/10.1038/s41467-018-05822-0}
	{\path{doi:10.1038/s41467-018-05822-0}}.
	
	\bibitem{Buzbas2015}
	E.~O. Buzbas, N.~A. Rosenberg, (2015)
	{AABC:
		Approximate approximate Bayesian computation for inference in
		population-genetic models}, Theoretical Population Biology 99 31.
	42.\\
	\newblock \href {http://dx.doi.org/https://doi.org/10.1016/j.tpb.2014.09.002}
	{\path{doi:10.1016/j.tpb.2014.09.002}}.
	
	\bibitem{Tripathy2018}
	R.~K. Tripathy, I.~Bilionis, (2018) Deep {UQ}: Learning deep neural network surrogate
	models for high dimensional uncertainty quantification, Journal of
	Computational Physics 375 565--588.\\
	\newblock \href {http://dx.doi.org/10.1016/j.jcp.2018.08.036}
	{\path{doi:10.1016/j.jcp.2018.08.036}}.
	
	\bibitem{Borowska2021}
	A.~Borowska, D.~Giurghita, D.~Husmeier, (2021) Gaussian process enhanced
	semi-automatic approximate Bayesian computation: parameter inference in a
	stochastic differential equation system for chemotaxis, Journal of
	Computational Physics 429 109999.
	\newblock \href {http://dx.doi.org/10.1016/j.jcp.2020.109999}
	{\path{doi:10.1016/j.jcp.2020.109999}}.
	
	\bibitem{Beaumont2009}
	M.~A. Beaumont, J.-M. Cornuet, J.-M. Marin, C.~P. Robert, (2009)
	{Adaptive approximate Bayesian
		computation}, Biometrika 96~(4) 983--990.
	
	\bibitem{LEcuyer2009}
	P.~L'Ecuyer, C.~L{\'{e}}cot, A.~L'Archev{\^{e}}que-Gaudet, (2009) On array-{RQMC} for
	Markov chains: Mapping alternatives and convergence rates, in: Monte Carlo
	and Quasi-Monte Carlo Methods 2008, Springer Berlin Heidelberg, pp.
	485--500.
	\newblock \href {http://dx.doi.org/10.1007/978-3-642-04107-5_31}
	{\path{doi:10.1007/978-3-642-04107-5_31}}.
	
	\bibitem{Puchhammer2021}
	F.~Puchhammer, A.~B. Abdellah, P.~L'Ecuyer, (2021) Variance reduction with
	array-{RQMC} for tau-leaping simulation of stochastic biological and chemical
	reaction networks, Bulletin of Mathematical Biology 83~(8).
	\newblock \href {http://dx.doi.org/10.1007/s11538-021-00920-5}
	{\path{doi:10.1007/s11538-021-00920-5}}.
	
	\bibitem{Beentjes2018}
	C.~H.~L. Beentjes, R.~E. Baker, (2018) Quasi-Monte Carlo methods applied to
	tau-leaping in stochastic biological systems, Bulletin of Mathematical
	Biology 81~(8) 2931--2959.\\
	\newblock \href {http://dx.doi.org/10.1007/s11538-018-0442-2}
	{\path{doi:10.1007/s11538-018-0442-2}}.
	
	\bibitem{Lee2010}
	A.~Lee, C.~Yau, M.~B. Giles, A.~Doucet, C.~C. Holmes,
	(2010) On the utility of
	graphics cards to perform massively parallel simulation of advanced Monte Carlo methods, Journal of Computational and Graphical Statistics 19~(4) 769--789.
	\newblock \href {http://dx.doi.org/10.1198/jcgs.2010.10039}
	{\path{doi:10.1198/jcgs.2010.10039}}.
	
	\bibitem{Klingbeil2011}
	G.~Klingbeil, R.~Erban, M.~Giles, P.~K. Maini, (2011) {STOCHSIMGPU}: parallel
	stochastic simulation for the systems biology toolbox 2 for {MATLAB},
	Bioinformatics 27~(8) 1170--1171.\\
	\newblock \href {http://dx.doi.org/10.1093/bioinformatics/btr068}
	{\path{doi:10.1093/bioinformatics/btr068}}.
	
	\bibitem{Hurn2016}
	S.~Hurn, K.~Lindsay, D.~J. Warne,
	(2016) A heterogeneous computing approach to maximum likelihood parameter estimation for the Heston model of stochastic volatility, {ANZIAM} Journal 57 364.\\
	\newblock \href {http://dx.doi.org/10.21914/anziamj.v57i0.10425}
	{\path{doi:10.21914/anziamj.v57i0.10425}}.
	
	\bibitem{Warne2021}
	D.~J. Warne, S.~A. Sisson, C.~Drovandi, (2021) Vector operations for accelerating
	expensive Bayesian computations {\textendash} a tutorial guide, Bayesian
	Analysis 17(2):593-622.
	\newblock \href {http://dx.doi.org/10.1214/21-ba1265}
	{\path{doi:10.1214/21-ba1265}}.
	
	\bibitem{Mahani2015}
	A.~S. Mahani, M.~T. Sharabiani, (2015) {SIMD} parallel {MCMC} sampling with
	applications for big-data Bayesian analytics, Computational Statistics {\&}
	Data Analysis 88 75--99.\\
	\newblock \href {http://dx.doi.org/10.1016/j.csda.2015.02.010}
	{\path{doi:10.1016/j.csda.2015.02.010}}.
	
	\bibitem{Kulkarni2020}
	S.~Kulkarni, A.~Tsyplikhin, M.~M. Krell, C.~A. Moritz, (2020) Accelerating
	simulation-based inference with emerging {AI} hardware, in: 2020
	International Conference on Rebooting Computing ({ICRC}), {IEEE}.
	\newblock \href {http://dx.doi.org/10.1109/icrc2020.2020.00003}
	{\path{doi:10.1109/icrc2020.2020.00003}}.
	
\end{thebibliography}
\end{document}